\documentclass[12pt]{article}
\usepackage{amsmath,amsthm,amssymb}
\usepackage{newpxtext}
\usepackage{newpxmath}
\usepackage{graphicx}
\usepackage{subcaption}
\usepackage{setspace}
\usepackage{hyperref}
\hypersetup{
    colorlinks=true,
    linkcolor=blue,
    citecolor=blue,
    filecolor=magenta,
    urlcolor=blue,
}
\usepackage{float}
\usepackage{listings}
\usepackage{color}
\usepackage{tikz}
\usepackage{pgfplots}
\pgfplotsset{compat=1.18}
\usepackage{pgfplotstable}
\usepackage{caption}
\usepackage{geometry}
\usepackage{booktabs}
\usepackage{multirow}
\usepackage{longtable}
\usepackage{tabularx}
\usepackage{array}
\usepackage{pdflscape}
\usepackage{rotating}
\usepackage{tcolorbox}
\usepackage{algorithm2e}
\usepackage{bbm}
\usepackage[authoryear,round]{natbib}

\usepackage{xcolor}

\newtheorem{definition}{Definition}
\newtheorem{lemma}{Lemma}
\newtheorem{proposition}{Proposition}
\newtheorem{theorem}{Theorem}
\newtheorem{corollary}{Corollary}
\newtheorem{remark}{Remark}
\newtheorem{assumption}{Assumption}
\newtheorem{example}{Example}

\definecolor{gammaCol}{RGB}{230,97,1}
\definecolor{expCol}{RGB}{0,119,187}
\definecolor{thetaCol}{RGB}{192,0,0}

\DeclareMathOperator*{\argmax}{arg\,max}

\newcommand{\E}{\mathbb{E}}
\renewcommand{\P}{\mathbb{P}}
\newcommand{\R}{\mathbb{R}}

\newcommand{\Law}{\mathcal{L}}
\DeclareMathOperator{\Sym}{Sym}

\begin{document}

\title{Dynamic Decision-Making under Model Misspecification: A Stochastic Stability Approach\footnote{This paper is a revised version of the working paper titled ``Dynamic Decision-Making under Model Misspecification''. We thank Pedro Dal Bó, Jack Fanning, Junnan He, Toru Kitagawa, Soonwoo Kwon, Heejun Lee, Govind Menon, Andriy Norets, Jonathan Roth, Susanne Schennach, Hui Wang and the participants of Brown Econometrics/Theory Seminar for helpful comments and suggestions. The second author is supported by  Vannevar Bush Faculty Fellowship ONR-N0014-21-1-2887.}}
\author{Xinyu Dai\footnote{Department of Economics, Brown University. Email: \href{mailto:xinyu_dai@brown.edu}{xinyu\_dai@brown.edu}} \hspace{2cm} Daniel Chen\footnote{Department of Applied Mathematics, Brown University. Email: \href{mailto:daniel_t_chen@brown.edu}{daniel\_t\_chen@brown.edu}} \hspace{2cm} Yian Qian\footnote{Department of Applied Mathematics, Brown University. Email: \href{mailto:yian_qian@brown.edu}{yian\_qian@brown.edu}}}
\newgeometry{left=1in, right=1in, top=0.5in, bottom=1in}
\maketitle
\thispagestyle{empty}
\vspace{-1.5em}

\begin{abstract}\small
    Firms, platforms, and policymakers increasingly rely on online learning algorithms to guide dynamic decisions, yet long-run consequences under model misspecification remain poorly understood. In a two-model bandit under Thompson Sampling, we provide a complete classification of posterior dynamics: beliefs may converge to a single model, or their distribution may converge to a nondegenerate stationary measure outside the Berk--Nash framework. We extend the analysis to general finite model classes through a stochastic-stability framework on the belief simplex, delivering sufficient conditions for ergodicity versus transience with recursive dimensional reduction, showing that endogenous experimentation can sustain persistent policy variation even with infinite data.
\end{abstract}


\medskip


\newpage
\restoregeometry

\section{Introduction}
Firms learn demand from sales data, platforms learn user responses from engagement metrics, and policymakers learn treatment effects from policy outcomes. In each setting, the decision maker’s actions shape the data available for learning: a firm that sets high prices observes demand only at high prices, a platform that recommends certain content observes engagement only with that content. When the decision maker’s model of the environment is also \emph{misspecified}, the interaction between endogenous data and incorrect beliefs can generate surprising long-run dynamics. This paper studies what happens to learning and decision-making in such environments.

The existing literature on misspecified learning characterizes conditions under which beliefs and actions converge to stable long-run outcomes, typically Berk-Nash equilibria or pseudo-true parameters \citep{espondaBerkNashEquilibrium2016,fudenbergLimitPointsEndogenous2021a}. These analyses generally focus on decision rules in which the agent deterministically best-responds to current beliefs. A large class of practical learning algorithms, however, use \emph{posterior sampling}: in each period the agent draws a model from her posterior and acts as if that draw were correct. The appeal is that posterior sampling automatically generates exploration---the agent occasionally tries actions that are not myopically optimal, driven by residual uncertainty---without requiring the decision maker to solve an intractable dynamic programming problem. The most prominent such algorithm is Thompson Sampling \citep{thompsonLikelihoodThatOne1933,russoTutorialThompsonSampling2020}, whose built-in randomization generates endogenous experimentation that interacts with misspecification in ways that deterministic rules cannot produce. Under correct specification, Thompson Sampling is known to perform well \citep{agrawalAnalysisThompsonSampling2012,kaufmannThompsonSamplingAsymptotically2012}. Under misspecification, the consequences of this interaction are poorly understood.

The central question of this paper is what kinds of long-run dynamics misspecified posterior-based learning can generate. We use the misspecified bandit as a clean laboratory: the decision maker chooses among actions, observes rewards, and updates beliefs over a finite set of candidate models, none of which need be correct. This setting isolates the feedback loop between beliefs, actions, and evidence while remaining tractable enough for a complete characterization. Specifically, we ask: \emph{What are the long-run belief dynamics induced by Thompson Sampling under misspecification, and what do they imply for action frequencies and regret?}

We address this question in two steps. First, we provide a complete dynamic classification of posterior evolution in the simplest nontrivial environment---a misspecified two-armed, two-model Gaussian bandit---and derive sharp predictions for limiting beliefs, action frequencies, and asymptotic regret. Second, we extend the analysis to general finite model classes via a stochastic-stability framework that treats posterior evolution as a Markov process on the belief simplex, yielding sufficient conditions for ergodic versus transient behavior and recursive dimensional reductions.

In the two-arm setting, the agent chooses $A_t\in\{1,2\}$ and observes $R_t\mid(A_t=i)\sim\mathcal N(g(i),1)$, while believing that under model $\theta\in\{\nu,\gamma\}$ the reward satisfies $R_t\mid(A_t=i)\sim\mathcal N(\theta_i,1)$. The dynamics of the posterior $\pi_t(\theta)$ are governed by two primitives: (i) whether $\nu$ and $\gamma$ recommend the same arm (agreement vs.\ disagreement), and (ii) for each arm $i$, the sign of the expected log-likelihood ratio drift $\Delta_i\equiv\mathbb E[\frac{f_\nu (R_t \mid A_t = i) }{f_\gamma (R_t \mid A_t=i)}]$, which indicates which model is closer to the true data-generating process in KL divergence when arm $i$ is pulled. When the models agree on the preferred arm, TS quickly becomes effectively deterministic and the posterior collapses to whichever model the evidence favors. The interesting case arises when the models disagree, so that TS induces a feedback loop between beliefs, actions, and evidence. We show that exactly three qualitative regimes arise:
\begin{itemize}
    \item A \emph{self-confirming} regime ($\Delta_1>0,\,\Delta_2<0$): each model is supported by data from the arm it recommends, so the posterior eventually concentrates on one model. Which model is selected is random and path-dependent.
    \item A \emph{uniform-dominance} regime ($\Delta_1$ and $\Delta_2$ have the same sign): one model is closer to the true DGP under both arms, so the posterior converges to it from any interior prior.
    \item A \emph{self-defeating} regime ($\Delta_1<0,\,\Delta_2>0$): each model is undermined by data from its recommended arm, preventing posterior convergence. Instead, $\pi_t$ evolves as an ergodic Markov process on $(0,1)$ with a nondegenerate stationary measure $\mu$, generating persistent oscillation between the two models.
\end{itemize}

As a concrete illustration, consider a firm choosing between two pricing strategies while learning demand from sales data. The firm entertains two demand models that disagree on whether a high or low price maximizes revenue. When each model's preferred price generates data that confirms it (self-confirming), the firm locks into one pricing rule forever, possibly the wrong one. When each model's preferred price instead generates refuting data (self-defeating), the firm oscillates persistently between pricing strategies, never settling down, producing stochastic price variation from a stationary demand environment. Our classification predicts exactly when each outcome arises from the primitives of the demand models.

Although our primary focus is on belief and action dynamics, we also study performance implications by linking regret to the limiting belief regime. The key determinant is whether TS eventually assigns full weight to a model whose prescribed action coincides with the true optimal arm. When the posterior concentrates on such a model, TS plays the optimal arm almost always and cumulative regret remains sublinear. Under misspecification, however, our classification shows two distinct ways in which this can fail. In the \emph{self-defeating} regime, the posterior does not concentrate and instead fluctuates persistently in the interior, implying that TS continues to sample suboptimal actions with positive long-run frequency and therefore incurs linear regret. In the \emph{self-confirming} regime, the posterior \emph{does} concentrate, but can concentrate on an incorrect model that recommends the wrong arm; in that case TS becomes asymptotically deterministic on a suboptimal action and regret is again linear. Thus, misspecification can yield either near-optimal performance or severe long-run losses, depending on the interaction between the model--action mapping and the direction of statistical evidence.

We then extend the analysis to a finite model class $\Theta=\{\theta^{(1)},\ldots,\theta^{(M)}\}$ with an arbitrary number of actions. Our starting point is that under TS the posterior belief $\{\pi_t\}$ is itself a Markov process on the simplex $\Delta_M$. We study its \emph{stochastic stability} by asking whether the law of the posterior settles down to a steady-state distribution or instead drifts toward the boundary (i.e., eliminates models). In this context, we say the posterior is \emph{ergodic} if there exists a unique stationary distribution $\mu\in\mathcal P(\mathrm{int}\,\Delta_M)$ such that $\mathcal L(\pi_t)\Rightarrow \mu$ as $t\to\infty$, or equivalently, long-run time averages of functions of $\pi_t$ converge to their $\mu$-expectations. In this case, TS exhibits persistent belief mixing: the posterior keeps fluctuating in the interior rather than concentrating on a single model.

This approach yields two sufficient conditions that are easy to interpret. 
The first is an \emph{angle condition}: relative to an interior mean-field fixed point $S^\star$, the expected drift of the log-odds points back toward $S^\star$ everywhere away from it $\langle \xi(S),\, S-S^\star\rangle<0$. Geometrically, this rules out systematic outward drift and ensures that beliefs experience a restoring force toward the interior-mixing region. 
The second is a \emph{spectral condition} based on softmax potential geometry, which implies the mean drift can be factorized as a quasi-gradient system; we show that if the symmetric part of the drift matrix $G$ is negative definite, $\mathrm{Sym}(G)\prec 0$, and stochastic fluctuations are not too large relative to the drift strength, then the interior configuration is stochastically stable and the posterior is ergodic. Conversely, if $\mathrm{Sym}(G)\succ 0$, the interior configuration is unstable and beliefs are driven toward the boundary. Both conditions are established using Lyapunov-type drift arguments for general-state Markov chains (see, e.g., \citet{meynMarkovChainsStochastic1993,khasminskiiStochasticStabilityDifferential2012}). 
The angle condition corresponds to a Lyapunov function based on distance to $S^\star$, while the spectral condition uses one tied to the softmax geometry \citep{amariMethodsInformationGeometry2000}.

In the transient case, posterior mass progressively shifts toward the boundary of the simplex: TS assigns vanishing probability to some models, so beliefs effectively concentrate on a smaller subset, which is typically a single vertex or a proper face.
Moreover, when trajectories approach a face $\Delta_I$, the induced TS dynamics restricted to the surviving models $I$ inherit the same log-odds and drift geometry, yielding an \emph{inductive reduction}: iterating the same interior-versus-boundary tests on successive faces identifies terminal vertices and, when it occurs, ergodicity on low-dimensional faces. 
Finally, we show under mild symmetry in misspecification, the geometric requirements for full-dimensional interior ergodicity become increasingly stringent as $M$ grows, so vertex selection or low-dimensional face ergodicity should dominate in large model classes.

\paragraph{Related literature.}
Our paper contributes to the economics literature on learning under model misspecification. The foundational equilibrium concepts in this area are self-confirming equilibrium \citep{fudenbergSelfConfirmingEquilibrium1993}, in which each player's beliefs are correct along the equilibrium path but may be wrong off-path, and Berk-Nash equilibrium \citep{espondaBerkNashEquilibrium2016}, in which agents optimize given beliefs that minimize KL divergence from the true data distribution. Our self-confirming regime produces outcomes that are strict Berk-Nash equilibria (Corollary~1): each limiting action is optimal under the limiting belief, and the limiting belief KL-minimizes against the stationary data that action generates. What Thompson Sampling adds is a selection mechanism: posterior sampling selects among self-confirming limits with well-defined probabilities that depend on the prior and the evidence geometry. This parallels the persistence of superstitious beliefs identified by \citet{fudenbergSuperstitionRationalLearning2006}, where false beliefs survive because the data they generate never refute them.

A growing body of work studies the dynamics of convergence to such equilibria. \citet{fudenbergLimitPointsEndogenous2021a} show that only uniform Berk-Nash equilibria can be long-run limits of misspecified Bayesian learning; \citet{fudenbergWhichMisspecificationsPersist2023a} characterize which misspecifications persist; \citet{espondaAsymptoticBehaviorBayesian2021a} characterize asymptotic behavior via ODE methods; \citet{frickBeliefConvergenceMisspecified2023} develop martingale-based convergence results; and \citet{heidhuesConvergenceMisspecifiedLearning2021} establish convergence in one-dimensional continuous misspecified learning settings using stochastic approximation. These papers all identify conditions under which beliefs converge to a point. Notably, \citet{frickBeliefConvergenceMisspecified2023} observe that their identification condition, which ensures long-run beliefs are concentrated on a single state, rules out active learning settings such as bandit problems, where the agent's actions determine which signals are observed and beliefs may remain mixed in the long run. Our paper directly addresses this open case. The self-defeating regime yields a qualitatively different outcome: the posterior does not converge to any point but instead converges in distribution to a nondegenerate invariant measure, producing persistent belief mixing with infinite data. This phenomenon, belief cycling under misspecification, was first identified by \citet{nyarkoLearningMisspecifiedModels1991}, who showed that a monopolist with misspecified demand can cycle forever. \citet{heidhuesUnrealisticExpectationsMisguided2018} exhibit a related ``misdirected learning'' spiral in which an overconfident agent's actions generate data that reinforce incorrect beliefs. Our contribution is to provide a general classification (tied to the sign pattern of the expected log-likelihood drifts) of when cycling arises versus when beliefs converge, and to establish the complete dynamics in each regime. The qualitative classification is robust across posterior-based decision rules (Remark~\ref{rem:role_decision_rule}); Thompson Sampling's smooth logistic link function is what enables the precise characterization.\footnote{%
\citet{fudenbergActivelearningMisspecifiedPrior2017} study active learning with a misspecified prior and show that a patient agent's beliefs may not converge, while a myopic agent's do. Our setting differs in that the misspecification is in the likelihood rather than the prior, and the non-convergence arises from the endogenous feedback between beliefs, actions, and evidence rather than from patience.}

Several papers study robustness to misspecification or welfare consequences of biased learning. \citet{bohrenLearningHeterogeneousMisspecified2021} provide a general framework for misspecified social learning with a prospective/retrospective bias decomposition. \citet{frickWelfareComparisonsBiased2024} develop welfare comparisons for biased learning. \citet{lanzaniDynamicConcernMisspecification2025} shows that an agent who is aware of potential misspecification may exhibit endogenous behavioral cycles, while \citet{baRobustMisspecifiedModels2023} characterizes which misspecifications are globally vs.\ locally robust. \citet{ghoshRobustComparativeStatics2025} derives monotone comparative statics for steady-state behavior under misspecified Bayesian learning. Our paper takes a different approach: rather than studying an agent who is aware of or robust to misspecification, we characterize the dynamics and welfare consequences of a naive posterior-sampling agent. For comprehensive surveys of this literature, see \citet{bohrenMisspecifiedModelsLearning2025} and \citet{espondaLearningEquilibriumModel2026}.

Our paper also relates to the evolutionary game theory literature that interprets Bayesian updating as a form of replicator dynamics \citep{shaliziDynamicsBayesianUpdating2009,weibullEvolutionaryGameTheory1997,heMisspecifiedLearningEvolutionary2025}. We extend this connection by studying posterior learning dynamics \emph{with endogenous experimentation} induced by Thompson Sampling.\footnote{%
Existing work often establishes posterior convergence under assumptions that guarantee a law of large numbers for log-likelihood ratios along the realized history \citep[e.g.,][]{shaliziDynamicsBayesianUpdating2009}. We do not impose such averaging conditions ex ante because the signal process is endogenous under posterior-based experimentation. Instead, we identify structural conditions under which stochastic stability and ergodicity arise, and we characterize the resulting long-run belief dynamics.}

We contribute to the literature on bandit algorithms with misspecification. The machine learning literature primarily studies misspecification through worst-case regret bounds and robustness guarantees \citep{fosterAdaptingMisspecificationContextual2021,liModelMisspecificationReinforcementLearning2024}. \citet{fanFragilityOptimizedBandit2024} show that optimized bandit algorithms are fragile to slight misspecification; \citet{fanDiffusionApproximationsThompson2025} develop diffusion approximations for Thompson Sampling. \citet{liDynamicSelectionAlgorithmic2023} identify self-fulfilling bias in endogenous contextual bandits. Our approach differs: rather than bounding worst-case regret, we classify the qualitative long-run dynamics (convergence vs.\ ergodicity vs.\ transience) and show that the asymptotic object of learning can be an invariant distribution rather than a point.\footnote{Related questions are also studied in the reinforcement learning literature, which often considers nonparametric environments \citep{lattimoreLearningGoodFeature2020,bogunovicMisspecifiedGaussianProcess2021}. Our focus differs because we study the structured bandit problem with typical parametrizations, well suited to economic applications in which model parameters have behavioral or policy interpretations.}

Finally, our work complements the econometrics literature on inference under misspecification \citep{berkLimitingBehaviorPosterior1966,whiteMaximumLikelihoodEstimation1982,giacominiRobustBayesianInference2021a,bonhommeMinimizingSensitivityModel2022a,andrewsMisspecifiedMomentInequality2024a,mullerLocallyRobustEfficient2024,armstrongAdaptingMisspecification2025}. That literature studies what an econometrician can learn from misspecified models with exogenous data; we study what happens when a decision maker acts on misspecified models over time, generating the data endogenously. The KL-minimizing pseudo-truth from \citet{berkLimitingBehaviorPosterior1966} enters our analysis as the determinant of drift direction: the sign of $\Delta_i$ reflects which model is the KL-closest to the true data-generating process under each action.

The remainder of the paper is organized as follows. Section~2 presents the general model setup. Section~3 provides the complete analysis of the two-arm, two-model case, including regret implications and a dynamic pricing illustration. Section~4 develops the general stochastic-stability framework for finite model classes. Section~5 concludes. Proofs, additional examples, and simulation details are in the Appendix.

\section{Model Setup}

We study a dynamic decision problem evolving over discrete time periods $t = 0,1,2,\ldots$. In each period, a decision maker chooses an action $A_t$ from a finite action set $\mathcal A$, and then observes a real-valued stochastic reward $R_t$ drawn from an unknown conditional distribution
\begin{equation}\label{eq:true_reward}
R_t \sim g(\cdot \mid A_t = a),
\end{equation}
where $g(\cdot \mid a)$ is a probability distribution on $\mathbb R$. We write $g(a) := \mathbb{E}_g[R_t \mid A_t = a]$ for the true conditional mean reward of action $a$.

The decision maker does not observe $g$ and instead evaluates actions using a parametric family of reward models. Specifically, rewards are assumed to follow a conditional density
\begin{equation}\label{eq:model_reward}
R_t \sim f_\theta(\cdot \mid A_t),
\end{equation}
where $\theta$ belongs to a finite parameter set $\Theta \subset \mathbb R^d$. However, the \emph{true} reward function \(g\) may not be representable within this parametric class, meaning the model may be misspecified.

\begin{definition}[Model Misspecification]\label{def:misspec}
The model is misspecified if there exists no parameter
$\theta \in \Theta$ such that
\[
g(\cdot \mid a) = f_\theta(\cdot \mid a)
\quad \text{for all } a \in \mathcal A.
\]
\end{definition}

The objective of the decision maker is to maximize cumulative rewards over time\footnote{Here we assume constant 1 discount factor for simplicity.}. To evaluate actions, the decision maker uses Bayesian updating to maintain a posterior distribution $\pi_t$ over $\Theta$ based on observed actions and rewards in past periods.

We use $\mathbb F = (\mathcal F_t)_{t \in \mathbb N}$ to denote the filtration generated by the history of actions and rewards up to time $t$:
\begin{equation}
\mathcal{F}_t = \sigma(A_0, R_0, A_1, R_1, \ldots, R_{t-1}, A_t).
\label{eq:filtration}
\end{equation}

A (possibly randomized) policy $\mu=\{\mu_t\}_{t\ge 0}$ is a sequence of decision rules such that, at each date $t$, the realized action $A_t$ is $\mathcal F_{t-1}$-measurable. Equivalently, $\mu_t(\cdot\mid\mathcal F_{t-1})$ denotes the conditional distribution over actions given past history. 

In principle, an optimal Bayesian decision maker chooses actions to maximize cumulative rewards by solving a dynamic programming problem over belief states. This problem involves a fundamental exploration–exploitation trade-off: at each date, the agent must balance the immediate expected payoff of an action against its informational value for improving future decisions. A myopic Bayesian policy, which selects the action maximizing the current posterior expected reward, ignores this dynamic learning incentive and may prematurely lock into suboptimal actions.

Solving the fully optimal dynamic program is, however, computationally intractable in general\footnote{In special cases with independent arms and discounted rewards, this problem admits an index representation via the Gittins index, yielding an exact solution under strong separability assumptions. Outside such settings—most notably in structural bandit problems where actions are informative about multiple alternatives—indexability fails and exact Bayesian optimality becomes infeasible.}, as the state space consists of posterior beliefs whose dimension grows with the complexity of the model \citep{lattimoreBanditAlgorithms2020}. We therefore focus on Thompson Sampling, a canonical Bayesian learning heuristic originally proposed by \citet{thompsonLikelihoodThatOne1933}. TS replaces dynamic optimization with a simple probability-matching rule: at each date, the agent samples a parameter from her posterior distribution and chooses the action that is optimal under that parameter. This procedure induces intrinsic randomization driven by posterior uncertainty, generating exploration without explicitly solving the underlying dynamic program. We next introduce additional notations and formally define the algorithm.

For each $\theta \in \Theta$ and $a \in \mathcal{A}$, let $
r_\theta(a) := \mathbb E_\theta[R_t \mid A_t = a]
$ denote the expected reward under parameter $\theta$. We define optimal action to be $\label{eq:optimal_action} \phi(\theta) \;:=\; \argmax_{a \in \mathcal{A}} \{ r_{\theta}(a) \}$. This mapping assigns each parameter an action that maximizes expected reward under that parameter. We further assume that the maximizer is unique for each $\theta$. \footnote{This assumption is made for notational convenience only. When multiple actions maximize $r_\theta(\cdot)$, one can fix any deterministic tie-breaking rule (e.g., the smallest index) and all results below extend verbatim; alternatively, one can allow randomized tie-breaking without changing the substance of the analysis.} Because the model is misspecified, the optimal action $\max_{\theta \in \Theta} \phi(\theta)$ may not coincide with the action that maximizes expected reward under the true data-generating process \eqref{eq:true_reward}. We further denote the oracle action as $a^* := \argmax_{a \in \mathcal{A}} r_g(a) = \mathbb E_g[R_t \mid A_t = a]$.

Thompson Sampling induces a randomized policy $\mu^{\mathrm{TS}}$ in which actions are chosen as a function of posterior beliefs rather than by directly solving the Bayes-optimal dynamic programming problem.

\vspace{1em}

\begin{algorithm}[H]
\SetAlgoLined
\DontPrintSemicolon
\caption{Thompson Sampling}
\label{Algorithm:TS}

\KwIn{Prior distribution $\pi_0$ supported on $\Theta$}

\For{$t = 0,1,2,\ldots$}{
    Sample a parameter $\theta_t \sim \pi_t$\;
    Choose action $A_t = \phi(\theta_t)$\;
    Observe reward $R_t$\;
    Update posterior belief
    \[
    \pi_{t+1}(\theta)
    =
    \frac{f_\theta(R_t \mid A_t)\,\pi_t(\theta)}
    {\sum_{\vartheta \in \Theta} f_\vartheta(R_t \mid A_t)\,\pi_t(\vartheta)},
    \qquad \forall \theta \in \Theta.
    \]
}
\end{algorithm}
\vspace{1em}

To evaluate the performance of a policy, we introduce the notion of cumulative regret, which compares a policy's expected payoffs to those
of a benchmark that, always selects the action with the highest expected reward under the true data-generating process. The cumulative regret of a policy $\mu$ over $T$ periods is defined as
\begin{equation}\label{eq:regret}
\mathrm{Regret}_T(\mu)
:= \sum_{t=0}^{T-1} \bigl( r_g(a^*) - r_g(A_t) \bigr).
\end{equation}

Under correct model specification, Thompson Sampling is known to achieve logarithmic (and in many settings near-optimal) regret rates; see, e.g., \citet{russoTutorialThompsonSampling2020} for a survey. However, its behavior and performance under model misspecification remain poorly understood. The remainder of this paper addresses this gap by characterizing the long-run dynamics of posterior beliefs and actions under TS when the model is misspecified.

\section{Misspecified Two-Arm Bandit Analysis}

We now present a complete analysis of Thompson Sampling in a two-arm bandit setting where the decision maker's models are potentially misspecified. This simple but nontrivial example reveals the fundamental mechanisms underlying Thompson Sampling's behavior under misspecification and provides the building blocks for a more general framework.

\subsection{Problem Setup}

Consider a simplified version of our general setup with two actions $\mathcal{A} = \{1, 2\}$ and two candidate models $\Theta = \{\nu, \gamma\}$ where $\nu = (\nu_1, \nu_2)$ and $\gamma = (\gamma_1, \gamma_2)$ specify expected rewards $r_\theta(i) = \theta_i$ for each action $i$ (see Figure \ref{fig:crossing}).
% ---- BEGIN crossing.tex ----
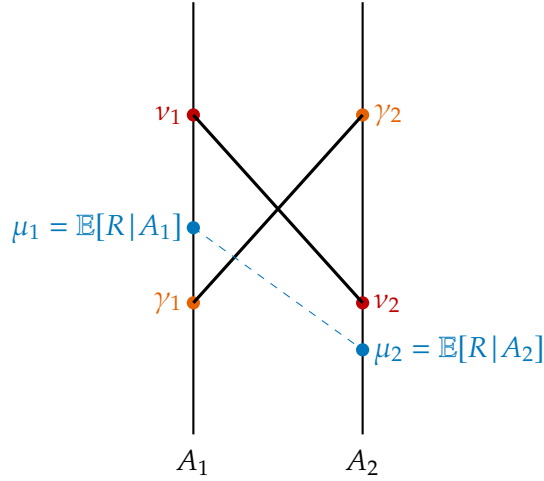
\begin{figure}[H]
    \centering
    \begin{tikzpicture}[>=latex,scale=1.25,
            every node/.style={font=\small}]

    % ------------------------------------------------------------------
    % 1.  two vertical "action'' lines
    % ------------------------------------------------------------------
    \draw[thick] (0, -0.4) -- (0, 4.2);   % A1
    \draw[thick] (1.8, -0.4) -- (1.8, 4.2);   % A2
    \node[below] at (0, -0.45){$A_1$};
    \node[below] at (1.8, -0.45){$A_2$};

    % ------------------------------------------------------------------
    % 2.  points on the left line  (A1)
    % ------------------------------------------------------------------
    \fill[gammaCol] (0, 1) circle (2.pt)
                    node[left] {$\gamma_1$};

    \fill[expCol]   (0, 1.8) circle (2.pt)
                    node[left] {$\mu_1 = \mathbb{E}[R\!\mid\!A_1]$};

    \fill[thetaCol] (0, 3) circle (2.pt)
                    node[left] {$\nu_1$};

    % ------------------------------------------------------------------
    % 3.  points on the right line  (A2)
    % ------------------------------------------------------------------
    \fill[gammaCol] (1.8, 3) circle (2.pt)
                    node[right] {$\gamma_2$};

    \fill[expCol]   (1.8, 0.5) circle (2.pt)
                    node[right] {$\mu_2 = \mathbb{E}[R\!\mid\!A_2]$};

    \fill[thetaCol] (1.8, 1) circle (2.pt)
                    node[right] {$\nu_2$};

    % ------------------------------------------------------------------
    % 4.  horizontal lines connecting points
    % ------------------------------------------------------------------
    \draw[very thick] (0, 1) -- (1.8, 3);   % γ1 → γ2
    \draw[very thick] (0, 3) -- (1.8, 1);   % θ1 → θ2
    \draw[dashed,expCol] (0, 1.8) -- (1.8, 0.5);   % Connect expectations with dashed line

    \end{tikzpicture}
    \caption{Crossing configuration showing misspecified models.}
    \label{fig:crossing}
\end{figure}
% ---- END crossing.tex ----
In the two-arm case, the posterior can be analyzed by the log-odds process, $\{S_t\}_{t \geq 0}$, defined by
\begin{equation}
S_t = \log\left(\frac{\pi_t}{1-\pi_t}\right), \quad t \geq 0
\end{equation}
where, the $\pi_t$, without loss of generality, is the posterior probability of picking model $\nu$ at time $t$. The posterior probability and log-odds are related through the sigmoid transformation, $\pi_t = \sigma(S_t) = \frac{1}{1 + e^{-S_t}}$. This transformation is strictly monotonic and continuous, mapping $\pi_t \in (0,1)$ to $S_t \in \mathbb{R}$. Specifically, $\pi_t \to 0$ if and only if $S_t \to -\infty$, and $\pi_t \to 1$ if and only if $S_t \to +\infty$. Therefore, analyzing the asymptotic behavior of the posterior $\pi_t$ is equivalent to studying the long-run dynamics of the log-odds process $S_t$.

The log-odds process evolves according to the additive structure:
\begin{equation}
S_{t+1} = S_t + Z_{t+1}
\label{eq:log-odds-two-arm}
\end{equation}
where $Z_{t+1} = \log[f_{\nu}(R_t|A_t) \pi_t(\nu)/f_{\gamma}(R_t|A_t) \pi_t(\gamma)]$ is the log-likelihood ratio increment. 

Notice that the posterior $\pi_{t-1}$ is measurable with respect to $\mathcal{F}_{t-1}$, and the action $A_t$ is chosen according to Thompson Sampling based on $\pi_{t-1}$. Therefore, the expected increment of the log-odds process can be expressed as a mixture of expected increments conditional on each action:
\begin{equation}
    \label{eq:two-arm-logodds-increment}
        \mathbb{E}[Z_t \mid \mathcal{F}_{t-1}] = \pi_{t-1} \Delta_{\phi(\nu)} + (1-\pi_{t-1}) \Delta_{\phi(\gamma)}
\end{equation}
where $\Delta_i = \mathbb{E}[Z_{t+1} \mid A_t = i]$ is the expected log-likelihood ratio increment under action $i$. Equivalently, letting $g_i := g(\cdot\mid i)$ and $f_{\nu,i}:=f_{\nu}(\cdot\mid i)$ (and similarly for $\gamma$),
\begin{equation}
\Delta_i = \mathbb{E}_{r\sim g_i}\big[\log f_{\nu,i}(r)-\log f_{\gamma,i}(r)\big] = -\mathrm{KL}(g_i\|f_{\nu,i})+\mathrm{KL}(g_i\|f_{\gamma,i}).
\end{equation}
Thus $\Delta_i>0$ means that, under the data generated by action $i$, model $\nu$ fits better than $\gamma$ in the KL sense.

The explicit law of $Z_t$ depends on the reward distribution and the model structure. However, if we further assume the Gaussian bandits setting, i.e. $R_t \mid A_t = i \sim \mathcal{N}(g(i), 1)$ while the decision maker believes $R_t \mid A_t = i \sim \mathcal{N}(\theta_i, 1)$ under model $\theta$. One can show that the increment distribution admits a tractable Gaussian-mixture form in this case (See Appendix \ref{app:B:gaussian}).\footnote{This Gaussian illustration is purely for transparency of the increment law. The regime classification results rely more generally on the signs of the mean increments $\Delta_i$ together with mild moment/regularity conditions on the log-likelihood ratios (see the regularity assumption~\ref{as:two-arm}), and do not hinge on normality per se.}

\subsection{Complete Classification of Asymptotic Behavior}
\label{sec:two_arm_dynamics}
As suggested by \eqref{eq:two-arm-logodds-increment}, the asymptotic behavior of Thompson Sampling depends critically on two key parameters: (i) $\Delta_{i}$, which measures the fitness for each model; (ii) $\phi(\theta)$, which determines whether the two models agree on the optimal action. In this subsection, we provide a complete characterization of all possible cases. We firstly divide the analysis into two main cases: the agreement case where both models agree on the optimal action, and the disagreement case where the models prefer different actions. Within each case, we further classify the behavior based on the signs of $\Delta_1$ and $\Delta_2$ (See in Figure \ref{fig:comparing_cases}).

% ---- BEGIN comparingCases.tex ----
\begin{figure}[H]
    \centering
    \begin{subfigure}[b]{0.48\textwidth}
        \centering
        \begin{tikzpicture}[>=latex,scale=1.0,
                every node/.style={font=\small}]

        % Two vertical "action" lines
        \draw[thick] (0, -0.4) -- (0, 4.2);   % A1
        \draw[thick] (1.8, -0.4) -- (1.8, 4.2);   % A2
        \node[below] at (0, -0.45){$A_1$};
        \node[below] at (1.8, -0.45){$A_2$};

        % Points on the left line (A1)
        \fill[gammaCol] (0, 1) circle (2.pt)
                        node[left] {$\gamma_1$};

        \fill[thetaCol] (0, 3) circle (2.pt)
                        node[left] {$\nu_1$};

        % Points on the right line (A2)
        \fill[gammaCol] (1.8, 0.5) circle (2.pt)
                        node[right] {$\gamma_2$};

        \fill[thetaCol] (1.8, 1) circle (2.pt)
                        node[right] {$\nu_2$};

        % Horizontal lines connecting points
        \draw[very thick, gammaCol] (0, 1) -- (1.8, 0.5);   % γ1 → γ2 (negative slope)
        \draw[very thick, thetaCol] (0, 3) -- (1.8, 1);   % θ1 → θ2 (negative slope)

        % Calculate midpoint
        \node at (-0.8, 2) {};
        \draw[dashed, expCol] (-0.5, 2) -- (0, 2);

        % Brackets and subcases
        \draw[thick, <-, expCol] (-0.5, 3.5) -- (-0.5, 2.1);
        \node at (-0.7, 3) {I};
        \node[right] at (-0.5, 3) {\textbf{}};

        \node at (-0.7, 2) {II};
        \node[right] at (-0.5, 2) {\textbf{}};

        \draw[thick, ->, expCol] (-0.5, 1.9) -- (-0.5, 0.5);
        \node at (-0.7, 1) {III};
        \node[right] at (-0.5, 1) {\textbf{}};

        \end{tikzpicture}
        \caption{Case 1: Both models agree $A_1$ is optimal}
        \label{fig:case1}
    \end{subfigure}
    \hfill
    \begin{subfigure}[b]{0.48\textwidth}
        \centering
        \begin{tikzpicture}[>=latex,scale=1.0,
                every node/.style={font=\small}]

        % Two vertical "action" lines
        \draw[thick] (0, -0.4) -- (0, 4.2);   % A1
        \draw[thick] (1.8, -0.4) -- (1.8, 4.2);   % A2
        \node[below] at (0, -0.45){$A_1$};
        \node[below] at (1.8, -0.45){$A_2$};

        % Points on the left line (A1)
        \fill[gammaCol] (0, 1) circle (2.pt)
                        node[left] {$\gamma_1$};

        \fill[thetaCol] (0, 3) circle (2.pt)
                        node[left] {$\nu_1$};

        % Points on the right line (A2)
        \fill[gammaCol] (1.8, 3) circle (2.pt)
                        node[right] {$\gamma_2$};

        \fill[thetaCol] (1.8, 1) circle (2.pt)
                        node[right] {$\nu_2$};

        % Horizontal lines connecting points
        \draw[very thick, gammaCol] (0, 1) -- (1.8, 3);   % γ1 → γ2 (positive slope)
        \draw[very thick, thetaCol] (0, 3) -- (1.8, 1);   % θ1 → θ2 (negative slope)

        % Midpoints
        \draw[dashed, expCol] (-0.5, 2) -- (0, 2);  % Left midpoint
        \draw[dashed, expCol] (1.8, 2) -- (2.3, 2); % Right midpoint

        % Region markings for left line
        \draw[thick, <-, expCol] (-0.5, 3.5) -- (-0.5, 2.1);
        \node at (-0.7, 3) {I};

        \node at (-0.7, 2) {II};

        \draw[thick, ->, expCol] (-0.5, 1.9) -- (-0.5, 0.5);
        \node at (-0.7, 1) {III};

        % Region markings for right line
        \draw[thick, <-, expCol] (2.3, 3.5) -- (2.3, 2.1);
        \node at (2.5, 3) {III};

        \node at (2.5, 2) {II};

        \draw[thick, ->, expCol] (2.3, 1.9) -- (2.3, 0.5);
        \node at (2.5, 1) {I};

        \end{tikzpicture}
        \caption{Case 2: Models disagree on optimal arm}
        \label{fig:case2}
    \end{subfigure}
    \caption{Comparison of Case 1 and Case 2 scenarios. In Case 1, both models agree arm $A_1$ is optimal. In Case 2, models disagree: $\nu$ prefers $A_1$ while $\gamma$ prefers $A_2$. The three regions (I, II, III) for the true expected rewards are indicated on each action line. In Region~I the true expected reward is such that model $\nu$ is closer to the true DGP in KL divergence ($\Delta_i > 0$), in Region~III model $\gamma$ is closer ($\Delta_i < 0$), and Region~II is the neutral boundary ($\Delta_i = 0$).}
    \label{fig:comparing_cases}
\end{figure}
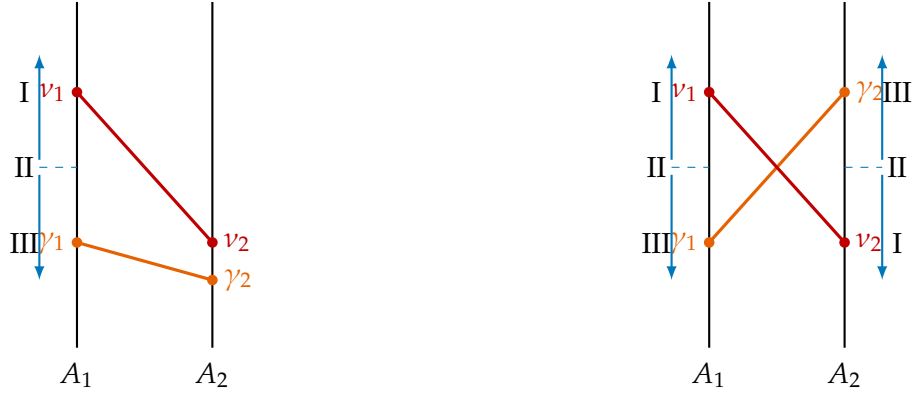
% ---- END comparingCases.tex ----

\paragraph{Agreement Case} (Baseline) Even though actions do not depend on learning in this regime (TS plays the common recommended action forever), posterior learning still selects among models based on their relative fit under the induced stationary data. The agreement case is a relatively straightforward scenario where both models agree on the optimal action: $\phi(\nu) = \phi(\gamma) = 1$. In this case, Thompson Sampling always selects one action, and the posterior evolution reduces to a simple random walk driven by the log-likelihood ratio increments from action 1. The asymptotic behavior of the posterior depends solely on the sign of $\Delta_1$.

\begin{theorem}[Thompson Sampling: Agreement Case]
\label{thm:agreement_case}
Assume both models agree on the optimal action: $\phi(\nu) = \phi(\gamma) = 1$. Then Thompson Sampling always selects action 1, and the posterior $\pi_t$ converges as follows:
\begin{enumerate}
    \item[(i)] If $\Delta_1 > 0$: $\pi_t \xrightarrow{a.s.} 1$
    \item[(ii)] If $\Delta_1 = 0$: $\pi_t \xrightarrow{d} \text{Bernoulli}(1/2)$
    \item[(iii)] If $\Delta_1 < 0$: $\pi_t \xrightarrow{a.s.} 0$
\end{enumerate}
\end{theorem}

In the agreement case, both models recommend the same action, so TS selects that action at every date. Hence the induced reward sequence is i.i.d. from the fixed conditional distribution $g(\cdot\mid 1)$ (and is therefore stationary in the usual time-series sense). Learning then reduces to standard Bayesian updating under misspecification with a fixed design (the action distribution is degenerate). In particular, the posterior concentrates on the pseudo-true model(s) that minimize KL divergence to the induced data-generating distribution, in the sense of \citet{berkLimitingBehaviorPosterior1966}.

\paragraph{Disagreement Cases}

When models disagree on the optimal action ($\phi(\nu) = 1, \phi(\gamma) = 2$), the behavior becomes more complex and depends on the signs of both $\Delta_1$ and $\Delta_2$. Specifically, three distinct regimes arise: the \textit{self-confirming} case ($\Delta_1 > 0 > \Delta_2$), where each model is reinforced by evidence from its preferred action and the posterior converges to one of the two boundary vertices with probabilities determined by the relative magnitudes of $\Delta_1$ and $\Delta_2$; the \textit{uniform-dominance} case ($\Delta_1, \Delta_2 > 0$), where one model consistently receives positive evidence under both actions and the posterior almost surely concentrates on that model; and the \textit{self-defeating} case ($\Delta_1 < 0 < \Delta_2$), where each model is contradicted by evidence from its preferred action, leading to a unique invariant distribution over $(0,1)$ and sustained stochastic fluctuations of beliefs around the interior.

\vspace{1em}
\noindent\textit{1. Self-Confirming ($\Delta_1 > 0 > \Delta_2$)}

\begin{theorem}[Thompson Sampling: Self-Confirming Case]
\label{thm:self_confirming}
Under disagreement ($\phi(\nu)=1,\phi(\gamma)=2$) with $\Delta_1>0>\Delta_2$, the posterior concentrates on the boundary: there exists $p^*\in(0,1)$ such that
\[\mathbb{P}(\pi_t\to 1)=p^*,\qquad \mathbb{P}(\pi_t\to 0)=1-p^*.\]
Moreover, the selection probability is
\[
p^* = u(S_0).
\]
\end{theorem}
\begin{remark}[Attractors and path dependence]
The two limits correspond to two absorbing long-run regimes: (i) $\pi_t\to 1$ implies eventual play of action $1$ forever; (ii) $\pi_t\to 0$ implies eventual play of action $2$ forever. Which regime is selected is path-dependent and is determined by early likelihood shocks through the unstable threshold at $p^*$.
\end{remark}

The self-confirming regime ($\Delta_1>0>\Delta_2$) exhibits the reinforcing feedback loop: beliefs tilt action choice, and the induced data coming from the chosen action further sharpen the belief which favors that action more. Here this feedback is: when $\nu$ is likely, TS plays action $1$ more often, and under action $1$ the likelihood ratio favors $\nu$ on average; when $\gamma$ is likely, TS plays action $2$ more often, and under action $2$ the likelihood ratio favors $\gamma$ on average. As a result, the process has two absorbing attractors, concentration on $\nu$ with perpetual play of action $1$, or concentration on $\gamma$ with perpetual play of action $2$, and which one prevails is determined by early random fluctuations in the posterior (captured by the threshold $p^*$).

These two limit outcomes can be interpreted as \emph{Berk-Nash equilibria}: in each outcome, the long-run action is optimal given the agent's limiting beliefs, and the limiting beliefs are KL-best responses to the endogenously generated stationary data under that action, while off-path beliefs about the unplayed action are unrestricted. In this sense, our self-confirming case provides a sharp two-point illustration of the general message in \citet{fudenbergLimitPointsEndogenous2021a}: under broad misspecification learning dynamics, any stable limit point must be strict Berk-Nash equilibrium. What Thompson Sampling adds here is a concrete implication in equilibrium selection: because TS is intrinsically randomized, it can select between multiple Berk-Nash outcomes with well-defined probabilities (See $p^*$ in Theorem~\ref{thm:self_confirming}).

\begin{corollary}[Self-confirming limits are strict Berk--Nash equilibria]
\label{coro:SC_is_BN}
Under the self-confirming regime of Theorem~\ref{thm:self_confirming} (disagreement with $\Delta_1>0>\Delta_2$), the two boundary limits $(1,\delta_\nu)$ and $(2,\delta_\gamma)$ are \emph{strict Berk--Nash equilibria}: each action $a^*$ is uniquely optimal under the corresponding belief $\mu^*=\delta_{\phi^{-1}(a^*)}$, and $\mu^*$ assigns probability one to the unique KL-minimizing model for the stationary data $g(\cdot\mid a^*)$.
\end{corollary}

\vspace{1em}
\noindent\textit{2. Uniform Dominance ($\Delta_1 > 0, \Delta_2 > 0$)}

\begin{theorem}[Thompson Sampling: Uniform Dominance Case]
\label{thm:uniform_dominance}
Under the disagreement condition with $\Delta_1 > 0$ and $\Delta_2 > 0$, the posterior converges almost surely: $\pi_t \xrightarrow{a.s.} 1$.
\end{theorem}

In the uniform-dominance regime ($\Delta_1>0$ and $\Delta_2>0$), model $\nu$ is favored by the likelihood ratio no matter which action is taken. $\gamma$ is globally inferior in explanatory power for every feasible sequences of observations generated by the agent's actions. As a result, TS will push beliefs away from $\gamma$. Beyond the convergence result, this behavior reveals the mechanism which closely aligns with the ``iterative elimination'' intuition in \citet{frickBeliefConvergenceMisspecified2023}: when one model is uniformly dominated in relative likelihood terms, posterior odds evolve like a supermartingale for the dominated model, leading to its eventual extinction. As we will see in Section~\ref{sec:multi-model-setting}, it plays an important role in the multiple model cases as well in a more general way.

\vspace{1em}
\noindent\textit{3. Self-Defeating ($\Delta_1 < 0 < \Delta_2$)}

\begin{theorem}[Thompson Sampling: Self-Defeating Case]
\label{thm:self_defeating}
Under disagreement condition with $\Delta_1<0<\Delta_2$, the posterior belief process
$\{\pi_t\}_{t\ge 0}$ is positive Harris recurrent on $(0,1)$.
Hence it admits a unique invariant probability measure $\mu$ on $(0,1)$
such that, for any initial condition $\pi_0\in(0,1)$,
\[
\mathcal{L}(\pi_t)\ \Longrightarrow\ \mu,
\qquad\text{as } t\to\infty.
\]
Moreover, $\mu$ has full support on $(0,1)$ and admits a continuous density.
Equivalently, if $\pi_\infty\sim\mu$, then $\pi_t \xrightarrow{d} \pi_\infty$.
\end{theorem}

In the self-defeating regime ($\Delta_1<0<\Delta_2$), the convergence result is qualitatively different. Here the feedback loop is negative: whenever a model becomes more likely, Thompson Sampling plays that model's preferred action more often, but the resulting data systematically undermines that same model in likelihood terms. High belief in $\nu$ induces frequent play of action $1$, yet action $1$ generates evidence against $\nu$ on average; low belief in $\nu$ induces frequent play of action $2$, yet action $2$ generates evidence against $\gamma$ on average. The consequence is persistent oscillation in beliefs and ongoing experimentation generated endogenously by the algorithm: neither model can permanently sustain itself because its own recommended behavior creates refuting data. This intuition is formalized by the existence of a unique invariant distribution with full support and the non-concentration corollary below, which together imply that beliefs (and hence actions) keep fluctuating in the long run.

\begin{corollary}[Non-concentration and persistent experimentation]
Under Theorem~\ref{thm:self_defeating}, the stationary law is non-degenerate:
for any $\varepsilon\in(0,1/2)$, $\mu([\varepsilon,1-\varepsilon])>0.$
In particular, in the stationary regime both actions are chosen with strictly
positive long-run frequency.
\end{corollary}

This phenomenon sits outside the standard Berk--Nash framework, which is fundamentally a \emph{fixed-point} concept: it predicts long-run outcomes in which play settles on a stationary action profile and beliefs are KL-best responses to the stationary data produced by that play. In our self-defeating case, there is no stationary action path and no belief concentration; instead the correct long-run object is an \emph{invariant distribution} over posteriors, reflecting perpetual stochastic cycling driven by endogenous information acquisition. 

This distinction also clarifies why the existing asymptotic results in the misspecification-learning literature do not describe this case. Results such as \citet{fudenbergLimitPointsEndogenous2021a} characterize the set of possible deterministic limit points of learning dynamics. Our self-defeating regime is not a counterexample to that logic. It instead illustrates that once we endogenize exploration through posterior sampling, the relevant asymptotic behavior need not be described by point limits at all. Similarly, elimination arguments in the spirit of \citet{frickBeliefConvergenceMisspecified2023} rely on a model being uniformly or sequentially dominated so that one can be deleted and the belief evolves as a super-martingale on the simplex. Here, neither model is globally dominated: each model is supported by the fact that when it is unlikely, it is sampled less, and thus faces less refuting evidence, preventing monotone extinction.

For economic applications, the self-defeating regime is important because it shows that misspecification can generate not only incorrect long-run certainty, but also persistent policy randomness even with infinite data. In environments where the decision-maker is constrained to a misspecified model class but continues to experiment (as in many online platforms, adaptive pricing, or treatment assignment), long-run predictions should therefore be distributional: welfare, regret, and steady-state behavior depend on the stationary mix of beliefs and actions rather than on a single equilibrium. 

At the same time, the emergence of qualitatively distinct convergence regimes (as summarized in Table~\ref{tab:classification_two_arm}) raises a practical question: given a particular decision problem and a candidate model class, can one determine \emph{ex ante} which regime will govern long-run behavior? Our answer is: partially yes, as we will show in Section~\ref{sec:multi-model-setting}, where we develop diagnostic criteria for regime identification and treat the \emph{ergodicity} or \emph{transience} of the belief--action process as the appropriate notion of a steady state under algorithmic learning.

\begin{table}[H]
\centering
\begin{tabular}{|c|c|c|c|}
\hline
\textbf{Case} & \textbf{Condition} & \textbf{Behavior} & \textbf{Limit Distribution} \\
\hline
Agreement & $\Delta_1 > 0$ & Submartingale & $\delta_1$ \\
Agreement & $\Delta_1 = 0$ & Martingale & $\text{Bernoulli}(1/2)$ \\
Agreement & $\Delta_1 < 0$ & Supermartingale & $\delta_0$ \\
\hline
Self-Confirming & $\Delta_1 > 0 > \Delta_2$ & Two attractors & $\text{Bernoulli}(p^*)$ \\
Uniform Dominance & $\Delta_1, \Delta_2 > 0$ & Global attractor & $\delta_1$ \\
Self-Defeating & $\Delta_1 < 0 < \Delta_2$ & Ergodic & Invariant measure on $(0,1)$ \\
\hline
\end{tabular}
\caption{Complete classification of Thompson Sampling behavior under misspecification in the two-arm case}
\label{tab:classification_two_arm}
\end{table}

\subsection{Regret Analysis for the Two-Arm Case}
\label{sec:two_arm_regret}
Building upon our characterization of posterior dynamics, we now summarize the long-run regret implications of Thompson Sampling in the two-arm environment. It is not surprising that under misspecification, Thompson Sampling can exhibit sharply different long-run welfare properties despite identical short-run incentives. 
Zero regret is obtained precisely when belief updating reinforces payoff comparisons, while persistent regret arises when belief dynamics and payoff ranking are misaligned. All formal statements and proofs are deferred to Appendix~\ref{app:two_arm_regret}. Table~\ref{tab:regret_analysis_two_arm} provides a complete classification of limiting action distributions and average regret across the four cases identified in Section~\ref{sec:two_arm_dynamics}. 
Throughout, we normalize payoffs so that $g(1) \ge g(2)$, and thus the oracle policy always selects action~1.

We define $\alpha^*$ to be the stationary long-run frequency of playing action $1$ in the self-defeating case, i.e., if $\mu$ denotes the unique invariant distribution of the belief process $\pi_t$, then $\alpha^* := \mathbb{E}_{\mu}[\pi]$.

\begin{table}[H]
\centering
\begin{tabular}{|c|c|c|c|}
\hline
\textbf{Case} & \textbf{Prob. of} & \textbf{Prob. of} & \textbf{Average} \\
& \textbf{Action 1} & \textbf{Action 2} & \textbf{Regret} \\
\hline
Agreement ($\phi(\nu)=\phi(\gamma)=1$) & $1$ & $0$ & $0$ \\
\hline
Agreement ($\phi(\nu)=\phi(\gamma)=2$) & $0$ & $1$ & $g(1) - g(2)$ \\
\hline
Self-Confirming & Bernoulli$(p^*)$ & Bernoulli$(1-p^*)$ & $\begin{cases} 0 \text{ w.p. } p^* \\ g(1)-g(2) \text{ w.p. } 1-p^* \end{cases}$ \\
\hline
Uniform Dominance & $1$ & $0$ & $0$ \\
\hline
Self-Defeating & $\alpha^*$ & $1-\alpha^*$ & $(1-\alpha^*)[g(1) - g(2)]$ \\
\hline
\end{tabular}
\caption{Limiting action frequencies and average regret of Thompson Sampling in the two-arm model under misspecification}
\label{tab:regret_analysis_two_arm}
\end{table}

\subsection{Application: Monopoly Pricing under Misspecified Demand}
\label{sec:pricing_application}

We apply the preceding results to a monopoly pricing problem studied by \citet{espondaBerkNashEquilibrium2016}, \citet{heidhuesConvergenceMisspecifiedLearning2021}, and \citet{frickBeliefConvergenceMisspecified2023}. A monopolist learns about his demand function while setting prices. True demand at price $p$ is $z = \omega^* - \beta p + \varepsilon$, where $\varepsilon \sim \mathcal{N}(0,1)$, $\omega^*$ is the unknown intercept, and $\beta > 0$ is the true slope of demand. The monopolist misperceives the slope to be $\hat{\beta} \neq \beta$ while entertaining uncertainty about the intercept.

We keep the same demand environment but change two things: we discretize the model class to two intercepts $\omega_H > \omega_L$, and we replace SEU maximization with Thompson Sampling. Under perceived slope $\hat{\beta}$ and intercept $\omega_\theta$, each model's optimal price is $\phi(\theta) = \omega_\theta/(2\hat{\beta})$. Since $\omega_H > \omega_L$, the high-intercept model $\nu$ recommends the high price $p_H = \omega_H/(2\hat{\beta})$ and the low-intercept model $\gamma$ recommends the low price $p_L = \omega_L/(2\hat{\beta})$. At each period, the TS monopolist draws a model from the posterior and charges that model's optimal price.

The predicted demands at price $p$ are $\nu_p = \omega_H - \hat{\beta}p$ and $\gamma_p = \omega_L - \hat{\beta}p$, and the true demand is $g(p) = \omega^* - \beta p$. Since both models share the slope $\hat{\beta}$, the model disagreement $\nu_p - \gamma_p = \omega_H - \omega_L$ is constant across prices. The expected log-likelihood ratio drift (see Appendix~\ref{app:B:gaussian}) therefore takes the closed form
\begin{equation}
\label{eq:delta_monopoly}
\Delta_p = (\omega_H - \omega_L)\bigl(\omega^* - \bar{\omega} + (\hat{\beta} - \beta)\,p\bigr),
\end{equation}
where $\bar{\omega} := (\omega_H + \omega_L)/2$. The drift at each price is the product of model disagreement and the deviation of true demand from the models' average prediction. The term $(\hat{\beta}-\beta)p$ captures how the slope misspecification amplifies with the price level.

Evaluating~\eqref{eq:delta_monopoly} at the two prices, the difference $\Delta_{p_H} - \Delta_{p_L} = (\omega_H - \omega_L)^2(\hat{\beta}-\beta)/(2\hat{\beta})$ is signed entirely by $\hat{\beta} - \beta$. This yields a sharp regime classification:
\begin{itemize}
\item When $\hat{\beta} > \beta$ (the monopolist overestimates price sensitivity), $\Delta_{p_H} > \Delta_{p_L}$, and only the self-confirming regime or uniform dominance can arise. By Theorem~\ref{thm:self_confirming}, the firm eventually locks into one pricing strategy permanently.
\item When $\hat{\beta} < \beta$ (the monopolist underestimates price sensitivity), $\Delta_{p_H} < \Delta_{p_L}$, enabling the self-defeating regime $\Delta_{p_H} < 0 < \Delta_{p_L}$ whenever the true intercept $\omega^*$ falls in an intermediate range. By Theorem~\ref{thm:self_defeating}, the posterior converges to a nondegenerate stationary distribution and the firm oscillates between prices indefinitely.
\end{itemize}
The economic intuition for the self-defeating case is as follows. Because both models underestimate price sensitivity, the high price $p_H$ generates demand below what both models predict, favoring the low-intercept model $\gamma$ that recommends $p_L$; conversely, the low price $p_L$ generates demand above what both models predict, favoring the high-intercept model $\nu$ that recommends $p_H$. Each model's recommended price produces data supporting its rival, creating an endogenous feedback loop that prevents belief convergence.

This application concretizes the general comparison with \citet{frickBeliefConvergenceMisspecified2023} developed after Theorems~\ref{thm:self_confirming}--\ref{thm:self_defeating}: the same slope misspecification that produces a globally stable Berk--Nash equilibrium under their continuous model class and SEU rule generates persistent ergodic cycling once the model class is discrete and the decision rule is Thompson Sampling.

\begin{remark}[Role of the decision rule]
\label{rem:role_decision_rule}
The regime classification depends on the sign pattern of $(\Delta_{p_H}, \Delta_{p_L})$, which is a property of the misspecification---the true DGP $g$ and the model class $\{\nu, \gamma\}$---not of the decision rule. Under any posterior-based rule that assigns interior action probabilities, the self-defeating configuration produces a restoring drift toward the interior, while the self-confirming configuration produces a repelling drift toward the boundary. What changes across rules is the \emph{link function} $S \mapsto (p_a(S))_a$, which shapes the mean-drift function and hence the stationary distribution or selection probabilities. We discuss specific alternatives---myopic Bayes, top-$k$ Thompson Sampling, and $\epsilon$-greedy---in Section~\ref{sec:multi-model-setting}.
\end{remark}

\subsection{Implications for General Misspecification}

The analysis of the two-arm bandit under misspecification reveals important insights into the mechanisms governing Thompson Sampling's behavior. Our case by case examination identifies two critical components that jointly determine the dynamics of both posterior beliefs and action selection: the structural relationship between models and actions, and the quality of model fit to the true data generating process.

The first component is the \emph{model-action structure}, which is the mapping from each model's parameters to its preferred actions. This structure partitions the action space according to each model's optimal policy, creating distinct regions of action preference. In our two-arm analysis, this manifests as whether models $\nu$ and $\gamma$ agree on the same optimal action ($\phi(\nu) = \phi(\gamma)$) or disagree ($\phi(\nu) \neq \phi(\gamma)$). This structural relationship shapes the algorithm's asymptotic behavior independently of how well either model fits the data. When models agree on actions, it creates redundancy in the action space, as one action will never be selected and the posterior dynamics simplify to a competition between model fit over the other action. When models disagree, much richer dynamics emerge, including the possibility of persistent exploration and mixed strategies.

The second component is the relative quality of each model's fit to the true data generating process, measured by Kullback-Leibler divergence. This component determines which specific dynamics emerge within the constraints imposed by the model-action structure. For instance, in disagreement cases, whether the algorithm exhibits self-confirming or self-defeating behavior depends on whether each model receives positive or negative evidence when its preferred action is selected. The interplay between these components explains why model accuracy alone is insufficient to predict Thompson Sampling's performance under misspecification.

Our case analysis reveals that complex posterior dynamics can be decomposed into combinations of three elementary categories. First, \emph{uniform dominance} lead to model elimination, where the posterior only concentrates on the model that better fits the observed data globally. Second, \emph{self-confirming dynamics} create convergence to boundary limits, where one model dominates because it receives consistent positive evidence. Third, \emph{self-defeating dynamics} generate persistent exploration, where both models receive negative evidence when relied upon, preventing convergence to either boundary. As we will see in the next section, these elementary categories operate as building blocks that can combine in various ways to produce the rich array of behaviors observed in more complex settings.

\section{Multi-Model Setting}
\label{sec:multi-model-setting}

We now extend the two-model analysis to a finite set of candidate models
$\Theta = \{\theta^{(1)}, \ldots, \theta^{(M)}\}$. This section develops a geometric stochastic-stability framework that organizes the asymptotic possibilities and provides testable sufficient conditions for general finite model classes. The section proceeds in three steps: (i) characterize the drift geometry of posterior log-odds, (ii) provide sufficient conditions for interior ergodicity versus boundary transience, and (iii) establish recursive reduction when posterior mass concentrates on lower-dimensional faces.

A key insight is that, although posterior beliefs evolve on the $(M-1)$-dimensional simplex, their dynamics admit a low-dimensional and highly structured representation.
Specifically, by expressing beliefs in log-odds coordinates relative to a reference model, the posterior evolution can be represented as a stochastic process in $\mathbb{R}^{M-1}$ with an additive drift--noise decomposition.
This representation reveals that belief dynamics are governed by a smooth vector field whose geometry fully determines the asymptotic behavior of the learning process.

The resulting dynamics exhibit a sharp qualitative dichotomy.
Either posterior beliefs remain stochastically stable in the interior of the simplex---leading to persistent mixing across models---or beliefs are driven toward a lower-dimensional face, endogenously eliminating some models.
Importantly, when belief mass concentrates on a face, the induced dynamics on that face inherit the same structural form as the original system.
This yields an inductive characterization of learning outcomes in the multi-model setting: model exclusion proceeds sequentially, and long-run behavior can be analyzed by iterating the same classification on successively smaller model sets.

In practice, the framework can be applied as follows. First, compute the vertex drift vectors $d(a_j)$ for each model's optimal action and check whether $\mathbf{0}$ lies in the convex hull of these vectors. If not, the posterior is transient (Proposition~\ref{prop:transience-no-equilibrium}). If so, check the angle condition (Theorem~\ref{thm:ergodicity-angle-condition}) or compute $\Sym(G)$ and apply the spectral condition (Theorem~\ref{thm:sufficient-condition-ergodicity-transience}). If interior stability fails, apply the same logic recursively on each boundary face to which the posterior may converge.

The remainder of this section formalizes these ideas.
We first introduce the log-odds representation and its drift--noise decomposition.
We then characterize two sufficient conditions for interior versus boundary attraction in terms of the induced vector field.
All technical derivations and proofs are deferred to the appendix.

\subsection{Log-Odds Representation and Drift Geometry}

To analyze posterior dynamics in the multi-model setting, it is convenient to work in log-odds coordinates.
Fix a reference model $\theta^{(M)}$ and define the log-odds vector

\[
S_t^{(k)} := \log \frac{\pi_t(\theta^{(k)})}{\pi_t(\theta^{(M)})},
\qquad k = 1,\ldots,M-1.
\]

This mapping provides a smooth global chart from $\mathbb{R}^{M-1}$ to the interior of the simplex $\Delta_M$.
Since the transformation is invertible, the posterior $\pi_t$ can be recovered uniquely from $S_t$ via a softmax map.

\paragraph{General State Markov Process} Under Bayesian updating, the log-odds process admits a simple additive representation (See Appendix~\ref{app:A:logodds-dynamics}):

\begin{equation}S_{t+1} = S_t + Z_t = S_t + \xi(S_t) + \varepsilon_t,\label{eq:logodds-main}
\end{equation}

where $\xi(S_t) := \mathbb{E}[S_{t+1} - S_t \mid S_t]$ is the conditional drift and $\varepsilon_t$ is a mean-zero martingale increment. Unlike unstructured random walks where the noise is often assumed to be Gaussian supported on the entire space, in our setting the increment $Z_t := \log \left( \frac{f_{\theta^{(1)}}}{f_{\theta^{(M)}}}, \ldots, \frac{f_{\theta^{(M-1)}}}{f_{\theta^{(M)}}} \right)$ is supported on a union of 1-dimensional manifolds. Specifically, it is the union of finite curves parameterized by the likelihood ratio of the underlying model since the reward $R_t$ is a scalar and the action $A_t$ possible to be selected is finite. Therefore, the single-step transition kernel is singular with respect to the Lebesgue measure. However, as long as the information geometry of the model is non-degenerate\footnote{For example, if those curves are not parallel straight lines. This means under different actions, the change rate of relative entropy between different models according to the reward is not perfectly proportional.}, the process can still reach any region of the state space via finite multiple steps. In such case, we will have $k$-step irreducibility. We formalize this intuition in Lemma~\ref{lemma:irreducibility-and-aperiodicity-of-the-log-odds-process}. 

\begin{assumption}\label{assum:full-rank-Reversibility}
    There exists a finite sequence of increments $z_1, \dots, z_L \in \mathcal{Z}$, where each $z_i$ lies on a smooth segment of a support curve $\gamma_{k_i}(t)$, such that:
    \begin{enumerate}
        \item \textbf{Rank Condition:} The Jacobian of the summation map at these points is full rank. Specifically, if $z_i = \gamma_{k_i}(t_i^*)$, then the set of tangent vectors $\{\gamma_{k_1}'(t_1^*), \dots, \gamma_{k_L}'(t_L^*)\}$ spans $\mathbb{R}^{M-1}$.
        \item \textbf{Reversibility Condition:} The vectors sum to zero:\[ \sum_{i=1}^L z_i = \mathbf{0}. \]
    \end{enumerate}
\end{assumption}

\begin{lemma}[Irreducibility and Aperiodicity of the Log-Odds Process]\label{lemma:irreducibility-and-aperiodicity-of-the-log-odds-process} The log-odds process $(S_t)_{t\geq 0}$ is a time-homogeneous Markov chain on the general state space $\mathbb{R}^{M-1}$. Let $\mathcal{Z} \subset \mathbb{R}^{M-1}$ denote the support of the increment $Z_t$'s distribution. Suppose the model satisfies the Assumption~\ref{assum:full-rank-Reversibility}, then the chain $(S_t)$ is $\psi$-irreducible and aperiodic with respect to the Lebesgue measure on $\mathbb{R}^{M-1}$. That is, for any $s \in \mathbb{R}^{M-1}$ and any set $B$ with positive Lebesgue measure, there exists $n$ such that $P^n(s, B) > 0$.
\end{lemma}

The drift $\xi(S)$ has a transparent geometric interpretation.
For each model $\theta^{(j)}$, let $d(a_j) \in \mathbb{R}^{M-1}$ denote the expected log-likelihood ratio vector when the action prescribed by $\theta^{(j)}$ is played under the true environment.\footnote{The notation $d(a_j)$ is indexed by model $j$ through the map $\phi$: we write $a_j := \phi(\theta^{(j)})$. When two models prescribe the same optimal action, their drift vectors coincide.}
\begin{equation}
    d(a_j) := \mathbb{E}_{r \sim f_{\theta^*}(\cdot \mid \phi(\theta^{(j)}))} \left[ \log \left( \frac{f_{\theta^{(1)}}}{f_{\theta^{(M)}}}, \ldots, \frac{f_{\theta^{(M-1)}}}{f_{\theta^{(M)}}} \right) \right] \in \mathbb{R}^{M-1}.
    \label{eq:drift_vector}
\end{equation}
Then the drift can be written as
\begin{equation}
\xi(S) = \sum_{j=1}^M \pi_j(S)\, d(a_j),
\label{eq:drift-convex-combination}
\end{equation}
where $\pi(S)$ denotes the posterior induced by $S$.
Thus, the log-odds drift is a convex combination of a finite set of \emph{vertex drift vectors}, with weights given by current beliefs. This mixture structure is a direct consequence of posterior sampling under TS.

The general-state Markov structure established above imposes strong restrictions on the possible long-run behavior of the log-odds process.
In particular, irreducibility and aperiodicity rule out periodic or absorbing cycles and ensure that the asymptotic dynamics must fall into a small number of canonical regimes.
Before turning to geometric intuition, we first state a fundamental probabilistic dichotomy that applies to the log-odds process.

\paragraph{Dichotomy of the Log-Odds Process}
A probabilistic feature of the log-odds dynamics is that the Markov chain $\{S_t\}$ admits a sharp dichotomy between \emph{recurrence} and \emph{transience}.
Informally, recurrence means that the process returns to neighborhoods of its current region infinitely often, while transience means that it eventually drifts away and does not return.
In particular, positive Harris recurrence implies the existence of a unique invariant distribution and ergodicity of time averages.

\begin{lemma}[Dichotomy of the Log-Odds Process]
\label{lemma:dichotomy-of-the-log-odds-process}
Under the assumption~\ref{assum:full-rank-Reversibility}, the Markov chain $\{S_t\}$ is either recurrent or transient.
\end{lemma}

From the perspective of posterior learning, recurrence corresponds to persistent belief fluctuations within the interior of the simplex, whereas transience implies that beliefs eventually drift toward the boundary, leading to the exclusion of some models. When beliefs concentrate on a proper face of the simplex, the induced dynamics inherit the same structural form as the original system, allowing the analysis to proceed inductively on a lower-dimensional model space. We formalize these implications in subsequent sections.

The dichotomy above provides a probabilistic backbone for long-run belief behavior but does not, by itself, determine which regime arises in a given environment.
To gain intuition about how this dichotomy resolves, we next examine the structure of the conditional drift through a mean-field representation.

\paragraph{Mean-Field ODE Intuition.}
To build intuition for how the recurrence--transience dichotomy resolves in different environments, we examine the structure of the drift.
The drift function $\xi(S)$ defines a deterministic vector field on $\mathbb{R}^{M-1}$, which can be interpreted as a mean-field representation of the forces acting on the log-odds process and characterize a corresponding mean-field ODE system.\footnote{Importantly, this vector field is not used to approximate the stochastic dynamics pathwise or to characterize convergence directly. Since the learning step here is not assumed to be asymptotically small, our analysis remains in discrete time. For continuous-time approximations of Thompson Sampling, see \citet{fanDiffusionApproximationsThompson2025}.} Fixed points of the vector field, defined by $\xi(S)=0$, correspond to belief configurations at which the expected log-likelihood forces across models exactly balance. Whether such points exist, and how the drift behaves in their neighborhood, will play a key role in the qualitative classification that follows.

We visualize the mean-field dynamics for $M=3$ using two complementary representations, shown in Figures~\ref{fig:vector_fields}.
The left column depicts the vector field in log-odds space $\mathbb{R}^{M-1}$, where each point $S$ represents a log-odds vector and each arrow indicates the mean drift $\xi(S)$—the expected direction of motion of the stochastic process at that state.
The right column shows the same dynamics projected onto the posterior simplex $\Delta_M$, where each point $\pi$ represents a belief distribution over models.
The two representations are linked by the softmax map: each log-odds state $S$ induces a unique posterior $\pi(S)=\psi(S)$, and the dynamics on the simplex are the pushforward of the log-odds dynamics under this map.

% ---- BEGIN vectorFieldfig.tex ----
\begin{figure}[htbp]
    \centering
    \captionsetup{font=small}
    % Row 1: Self-Confirming (a,b) and Self-Defeating (c,d)
    \begin{subfigure}[b]{0.24\textwidth}
        \centering
        \includegraphics[width=\textwidth]{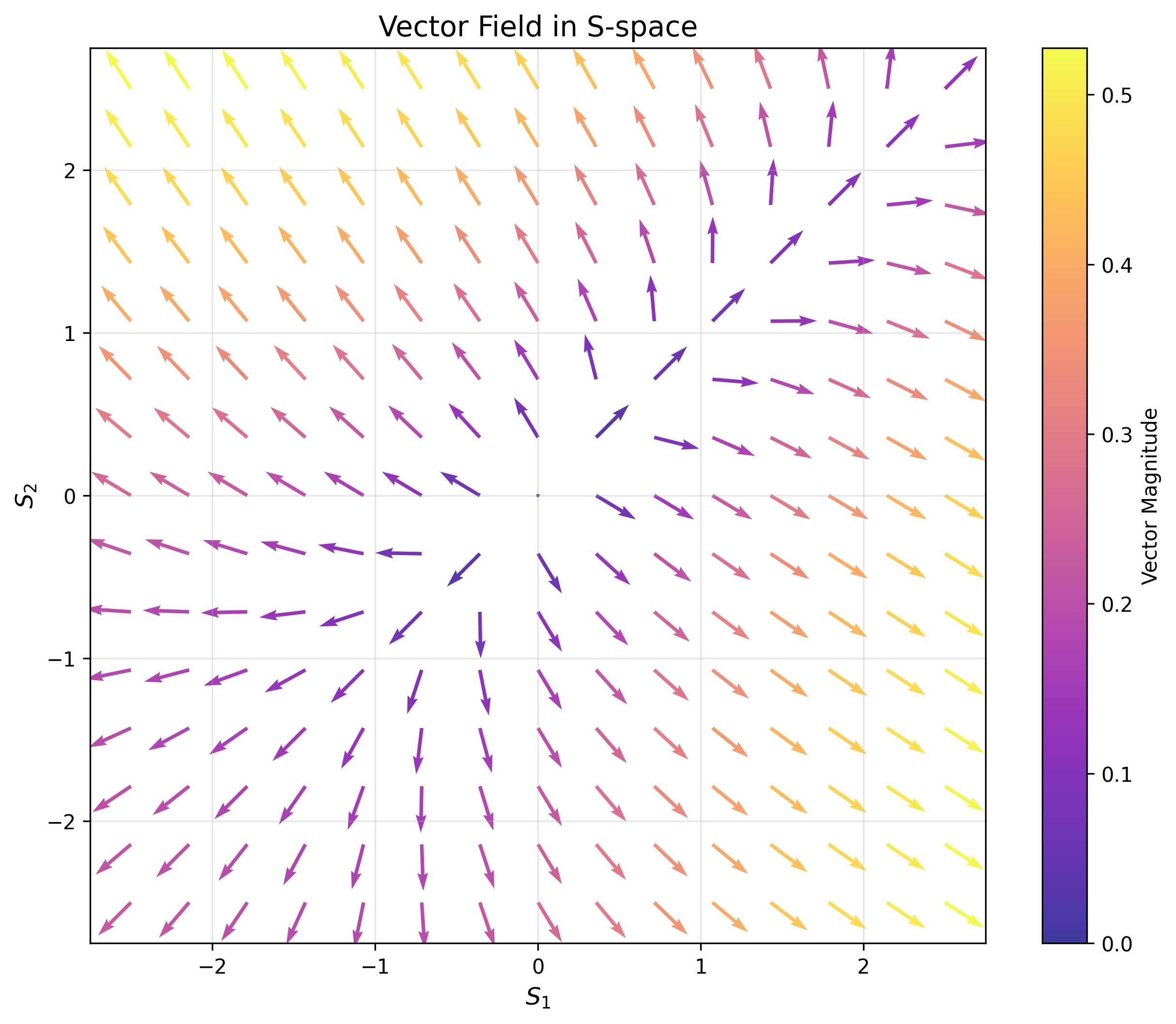}
        \caption{}
        \label{fig:S_selfconfirming}
    \end{subfigure}\hspace{-0.5em}%
    \begin{subfigure}[b]{0.24\textwidth}
        \centering
        \includegraphics[width=\textwidth]{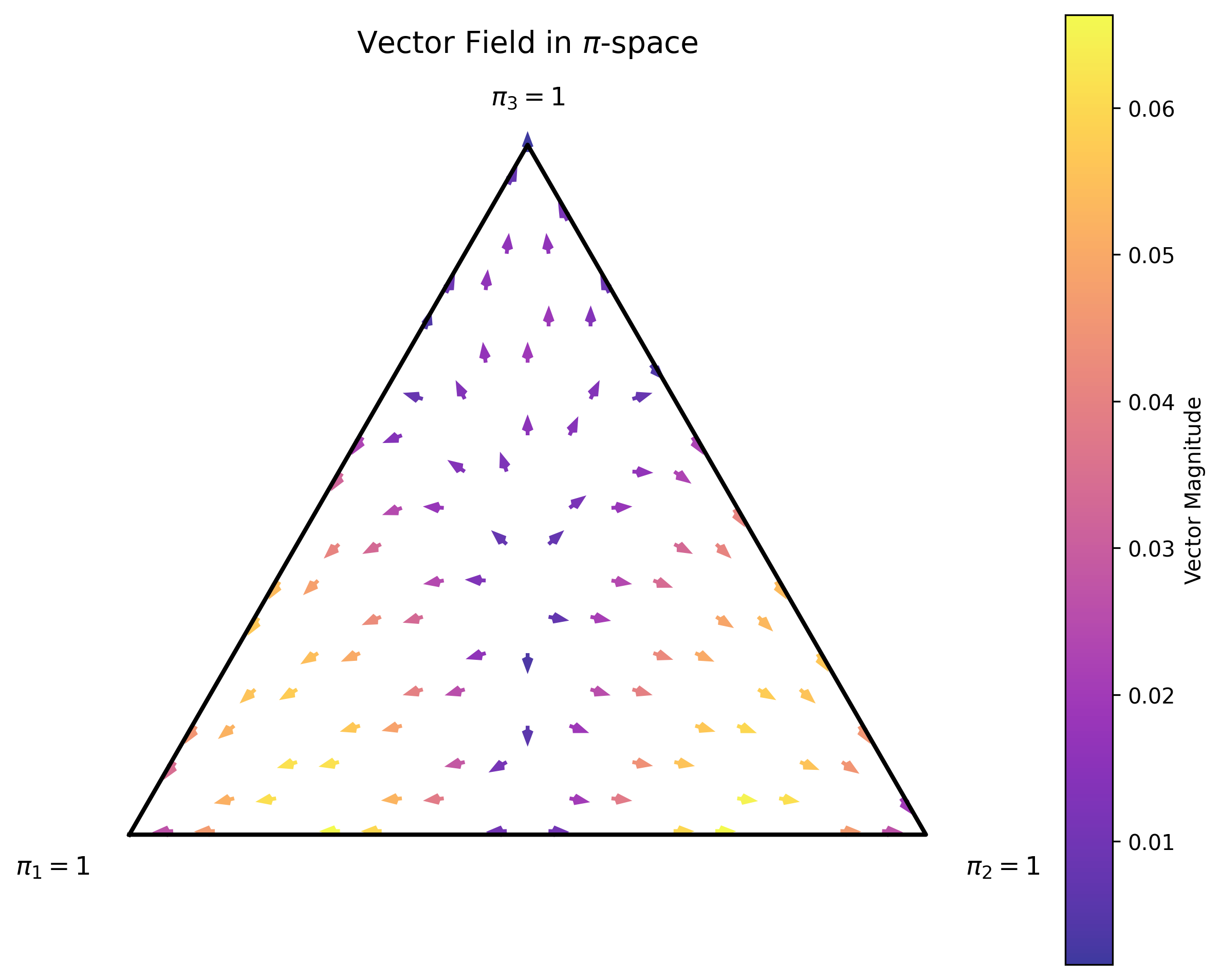}
        \caption{}
        \label{fig:pi_selfconfirming}
    \end{subfigure}\hspace{2em}%
    \begin{subfigure}[b]{0.24\textwidth}
        \centering
        \includegraphics[width=\textwidth]{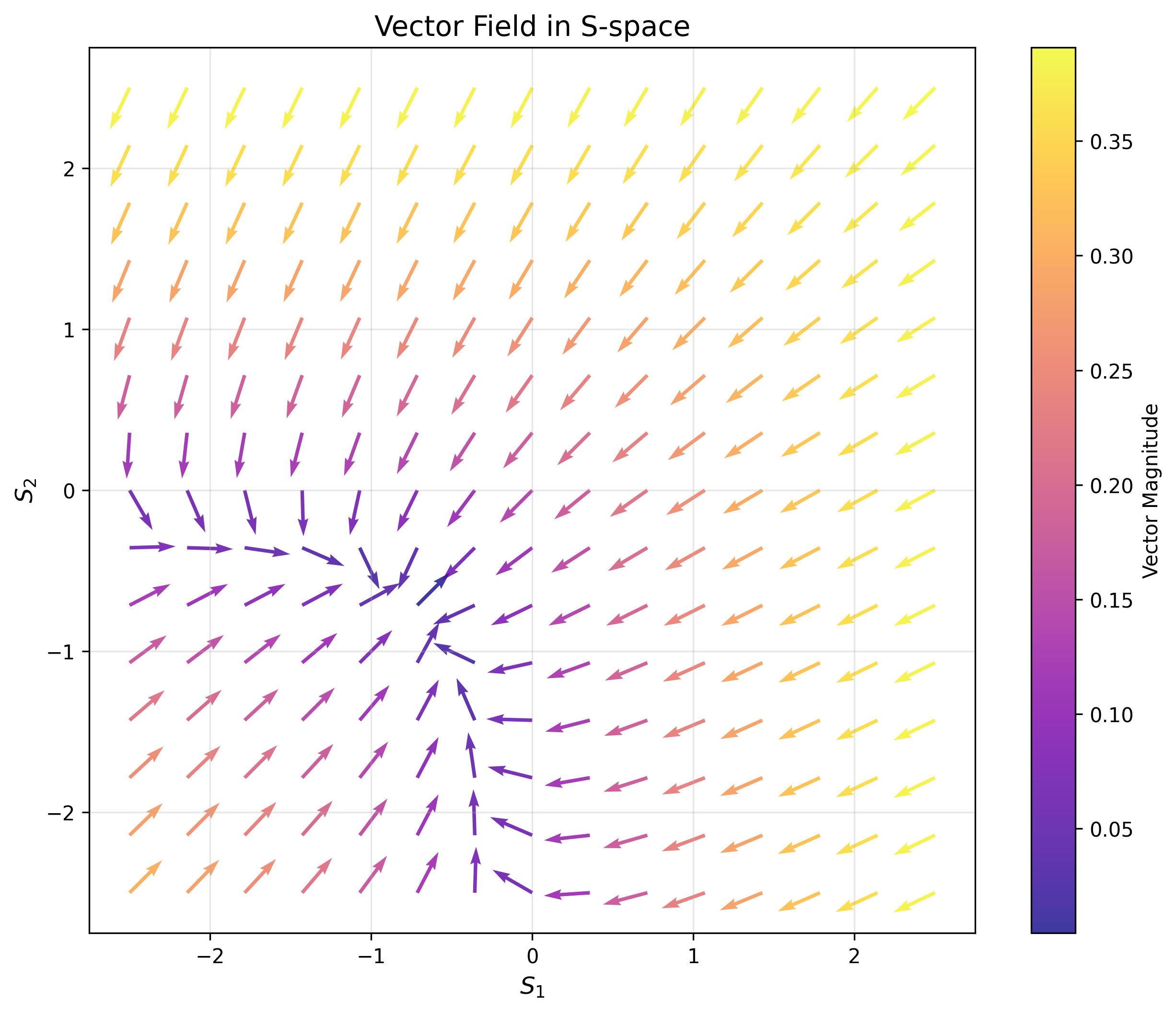}
        \caption{}
        \label{fig:S_selfdefeating}
    \end{subfigure}\hspace{-0.5em}%
    \begin{subfigure}[b]{0.24\textwidth}
        \centering
        \includegraphics[width=\textwidth]{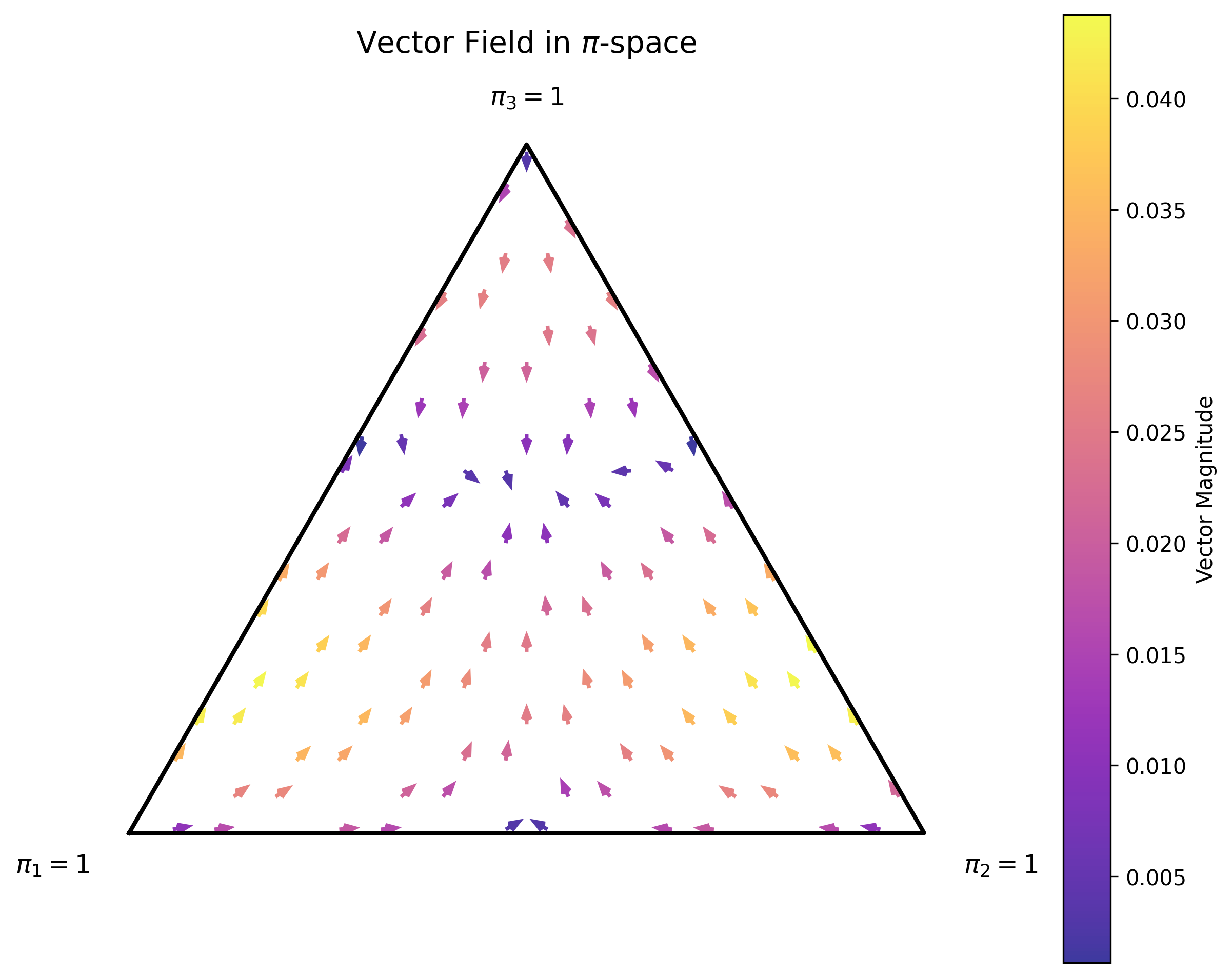}
        \caption{}
        \label{fig:pi_selfdefeating}
    \end{subfigure}

    \vspace{0.5em}

    % Row 2: Drift to Face (e,f) and Drift to Vertex (g,h)
    \begin{subfigure}[b]{0.24\textwidth}
        \centering
        \includegraphics[width=\textwidth]{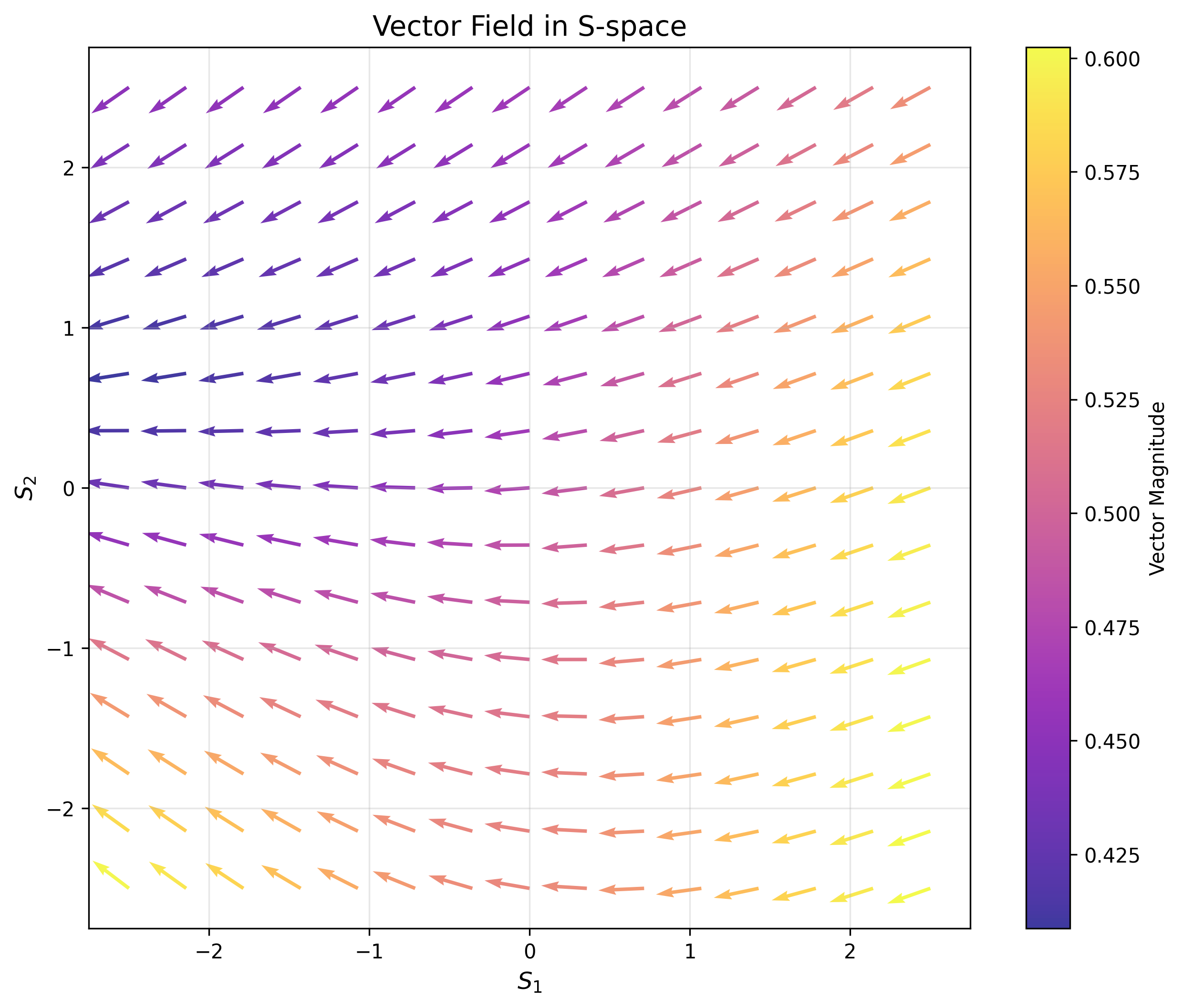}
        \caption{}
        \label{fig:S_drifttoface}
    \end{subfigure}\hspace{-0.5em}%
    \begin{subfigure}[b]{0.24\textwidth}
        \centering
        \includegraphics[width=\textwidth]{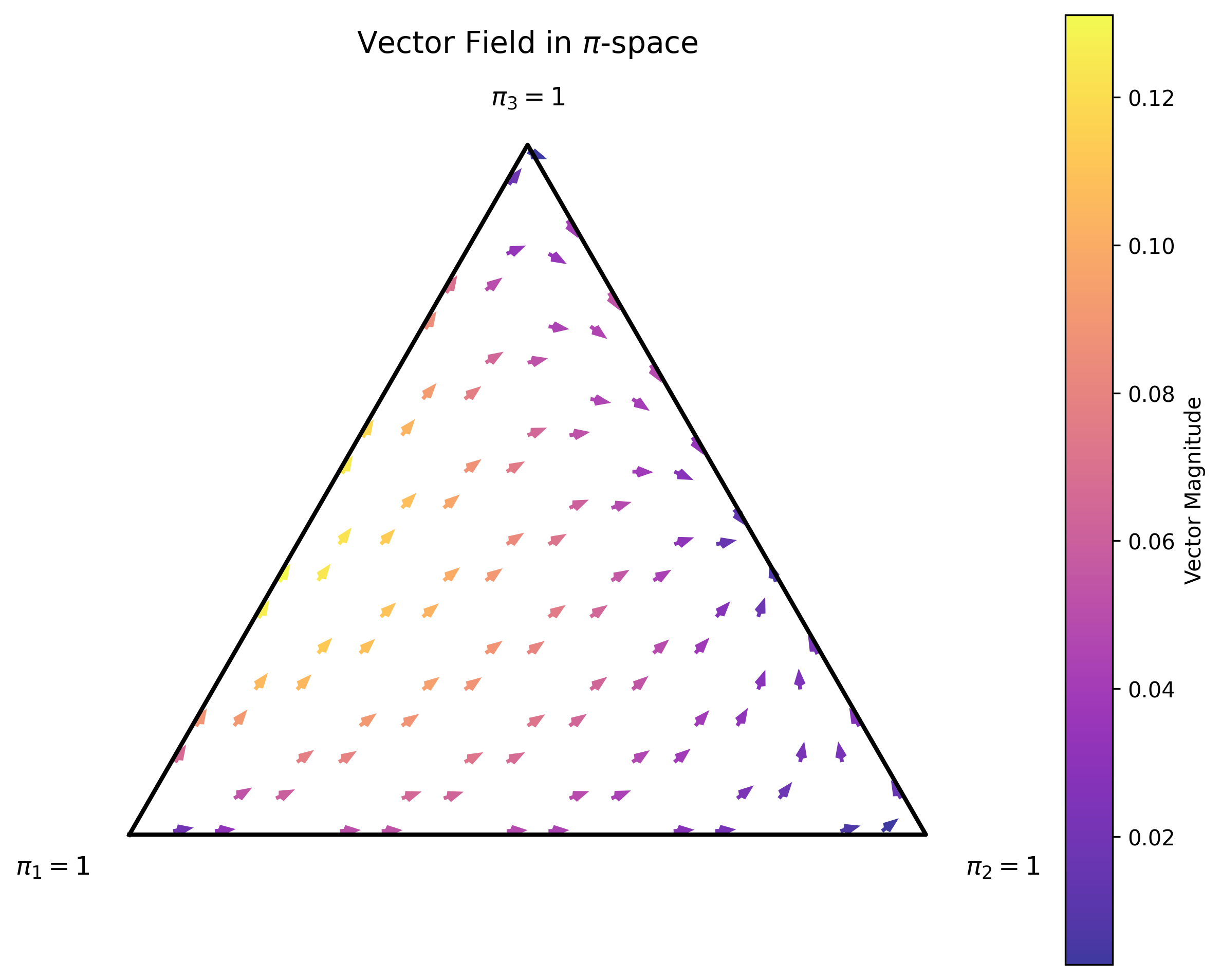}
        \caption{}
        \label{fig:pi_drifttoface}
    \end{subfigure}\hspace{2em}%
    \begin{subfigure}[b]{0.24\textwidth}
        \centering
        \includegraphics[width=\textwidth]{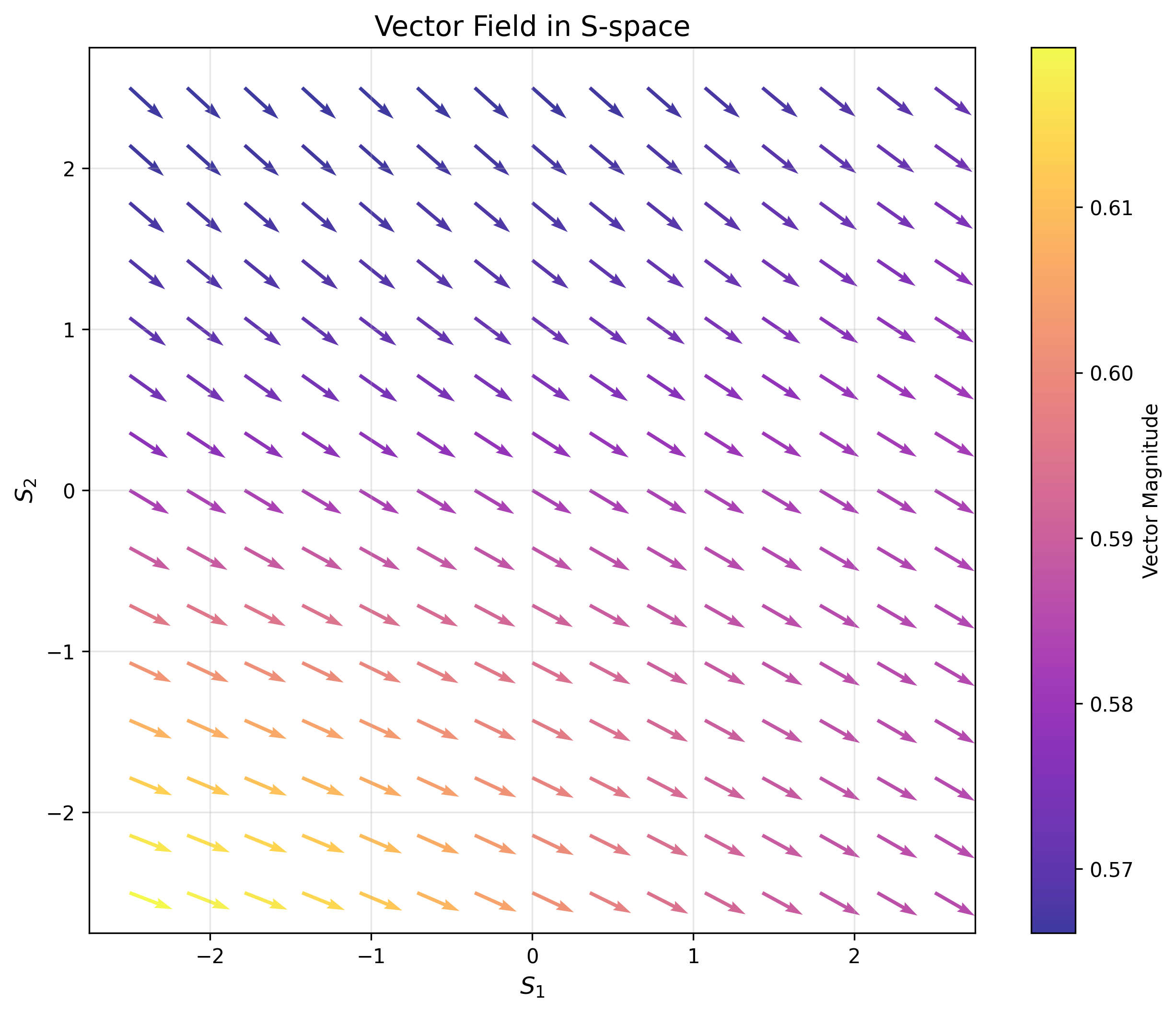}
        \caption{}
        \label{fig:S_drifttovertex}
    \end{subfigure}\hspace{-0.5em}%
    \begin{subfigure}[b]{0.24\textwidth}
        \centering
        \includegraphics[width=\textwidth]{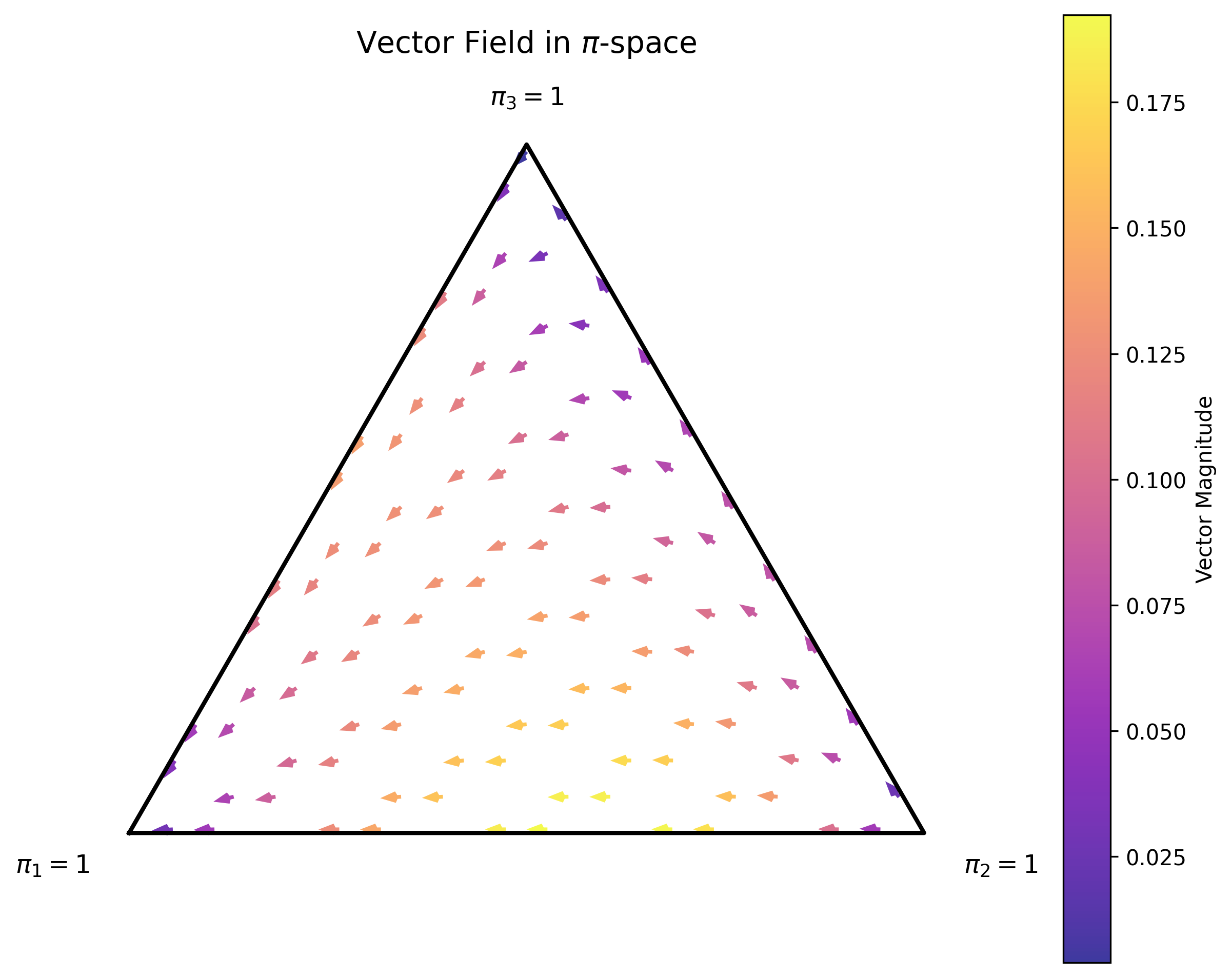}
        \caption{}
        \label{fig:pi_drifttovertex}
    \end{subfigure}

    \caption{Vector field visualizations. Left two columns show S-space; right two columns show simplex space. (a,b)~Self-confirming equilibrium; (c,d)~Self-defeating equilibrium. (e,f)~Drift towards face interior; (g,h)~Drift towards vertex.}
    \label{fig:vector_fields}
\end{figure}
% ---- END vectorFieldfig.tex ----

At any state $S$, the mean drift is given by a convex combination of the model-specific drift vectors according to \eqref{eq:drift-convex-combination}.
Thus, the arrow shown at each point in the vector field figures is a weighted average of the vertex drift vectors $\{d(a_1),\ldots,d(a_M)\}$, with weights given by the current posterior beliefs.
This convex-mixture structure imposes strong global geometric constraints on the drift field and underlies the classification below.

\paragraph{Classification via Drift Geometry}
The vector field of the log-odds drift $\xi(S)$ provides a geometric heuristic for refining the probabilistic dichotomy established in Lemma~\ref{lemma:dichotomy-of-the-log-odds-process}. We illustrate this idea using the case $M=3$.

As shown in Figure~\ref{fig:equilibria_types}, a central object is the relative position of the origin and the convex hull of the vertex drift vectors $\{d(a_1),\ldots,d(a_M)\}$.
Because the mean drift $\xi(S)$ is always a convex combination of these vectors, it can vanish only if the origin lies inside their convex hull.
The position of the origin therefore determines whether an interior fixed point of the mean-field vector field can exist.

% ---- BEGIN fig_equilibria_types.tex ----
\begin{figure}[htbp]
\centering
\captionsetup{font=small}
\begin{subfigure}[b]{0.45\textwidth}
\centering
\begin{tikzpicture}[scale=0.9]
% Define points for self-confirming equilibrium case
\coordinate (d1) at (2, 0.5);
\coordinate (d2) at (-0.5, 2);
\coordinate (d3) at (-1.5, -2);
\coordinate (origin) at (0, 0);

% Draw and fill convex hull
\fill[blue!20, opacity=0.7] (d1) -- (d2) -- (d3) -- cycle;
\draw[thick, blue] (d1) -- (d2) -- (d3) -- cycle;

% Draw axes
\draw[->, thin] (-2.5, 0) -- (2.5, 0) node[right] {$S_1$};
\draw[->, thin] (0, -2.5) -- (0, 2.5) node[above] {$S_2$};

% Draw drift vectors
\draw[->, thick, black] (origin) -- (d1);
\draw[->, thick, black] (origin) -- (d2);
\draw[->, thick, black] (origin) -- (d3);

% Label points
\node[above right] at (d1) {\(d(a_1)\)};
\node[above left] at (d2) {\(d(a_2)\)};
\node[below right] at (d3) {\(d(a_3)\)};

% Mark origin
\fill[red] (origin) circle (3pt);
\node[below left] at (origin) {\(\mathbf{0}\)};

\end{tikzpicture}
\caption{Interior Equilibrium: Self-Confirming}
\label{fig:interior_equilibrium_self_confirming}
\end{subfigure}
\hfill
\begin{subfigure}[b]{0.45\textwidth}
\centering
\begin{tikzpicture}[scale=0.9]
% Define points for self-defeating equilibrium case
\coordinate (d1) at (-2, -0.5);
\coordinate (d2) at (0.5, -2);
\coordinate (d3) at (1.5, 2);
\coordinate (origin) at (0, 0);

% Draw and fill convex hull
\fill[blue!20, opacity=0.7] (d1) -- (d2) -- (d3) -- cycle;
\draw[thick, blue] (d1) -- (d2) -- (d3) -- cycle;

% Draw axes
\draw[->, thin] (-2.5, 0) -- (2.5, 0) node[right] {$S_1$};
\draw[->, thin] (0, -2.5) -- (0, 2.5) node[above] {$S_2$};

% Draw drift vectors
\draw[->, thick, black] (origin) -- (d1);
\draw[->, thick, black] (origin) -- (d2);
\draw[->, thick, black] (origin) -- (d3);

% Label points
\node[above left] at (d1) {\(d(a_1)\)};
\node[below right] at (d2) {\(d(a_2)\)};
\node[above right] at (d3) {\(d(a_3)\)};

% Mark origin (coincides with d1)
\fill[red] (origin) circle (3pt);
\node[below left] at ([xshift=-10pt,yshift=-10pt]origin) {\(\mathbf{0}\)};

\end{tikzpicture}
\caption{Interior Equilibrium: Self-Defeating}
\label{fig:interior_equilibrium_self_defeating}
\end{subfigure}

\begin{subfigure}[b]{0.45\textwidth}
\centering
\begin{tikzpicture}[scale=0.9]
% Define points for no equilibria case
\coordinate (d1) at (-0.5, 1.5);
\coordinate (d2) at (1.5, -0.5);
\coordinate (d3) at (1, 1);
\coordinate (origin) at (0, 0);

% Draw and fill convex hull
\fill[blue!20, opacity=0.7] (d1) -- (d2) -- (d3) -- cycle;
\draw[thick, blue] (d1) -- (d2) -- (d3) -- cycle;

% Draw axes
\draw[->, thin] (-2.5, 0) -- (2.5, 0) node[right] {$S_1$};
\draw[->, thin] (0, -2.5) -- (0, 2.5) node[above] {$S_2$};

% Draw drift vectors
\draw[->, thick, black] (origin) -- (d1);
\draw[->, thick, black] (origin) -- (d2);
\draw[->, thick, black] (origin) -- (d3);

% Label points
\node[above left] at (d1) {\(d(a_1)\)};
\node[below right] at (d2) {\(d(a_2)\)};
\node[above right] at (d3) {\(d(a_3)\)};

% Mark origin
\fill[red] (origin) circle (3pt);
\node[below left] at (origin) {\(\mathbf{0}\)};

\end{tikzpicture}
\caption{No Equilibria: Drift to a Face Interior}
\label{fig:no_equilibrium_face_interior}
\end{subfigure}
\hfill
\begin{subfigure}[b]{0.45\textwidth}
\centering
\begin{tikzpicture}[scale=0.9]
% Define points for no equilibrium case
\coordinate (d1) at (1, 1.5);
\coordinate (d2) at (-0.5, 1);
\coordinate (d3) at (0.3, 0.8);
\coordinate (origin) at (0, 0);

% Draw and fill convex hull
\fill[blue!20, opacity=0.7] (d1) -- (d2) -- (d3) -- cycle;
\draw[thick, blue] (d1) -- (d2) -- (d3) -- cycle;

% Draw axes
\draw[->, thin] (-2.5, 0) -- (2.5, 0) node[right] {$S_1$};
\draw[->, thin] (0, -2.5) -- (0, 2.5) node[above] {$S_2$};

% Draw drift vectors
\draw[->, thick, black] (origin) -- (d1);
\draw[->, thick, black] (origin) -- (d2);
\draw[->, thick, black] (origin) -- (d3);

% Label points
\node[above right] at (d1) {\(d(a_1)\)};
\node[above left] at (d2) {\(d(a_2)\)};
\node[below right] at (d3) {\(d(a_3)\)};

% Mark origin
\fill[red] (origin) circle (3pt);
\node[below left] at (origin) {\(\mathbf{0}\)};

\end{tikzpicture}
\caption{No Equilibria: Drift to a Vertex}
\label{fig:no_equilibrium_vertex}
\end{subfigure}

\caption{Different equilibrium types in multi-model Thompson Sampling. In each panel, the arrows $d(a_1), d(a_2), d(a_3)$ represent the drift vectors associated with each model, which are the expected directions in which log-odds would move if that model were played with probability one. The shaded region is the convex hull $\mathrm{conv}\{d(a_1), d(a_2), d(a_3)\}$, representing all possible mean drift directions achievable by mixing over models. We list out 4 representative cases:  (a)~Interior equilibrium with self-confirming models. (b)~Interior equilibrium with self-defeating models. (c)~No equilibria: drift towards face interior. (d)~No equilibria: drift towards vertex.}
\label{fig:equilibria_types}
\end{figure}
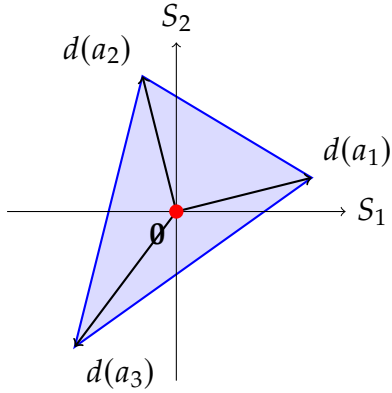
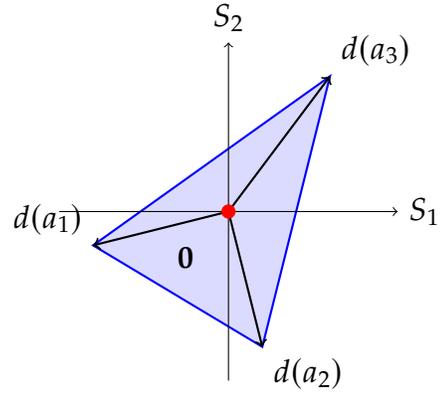
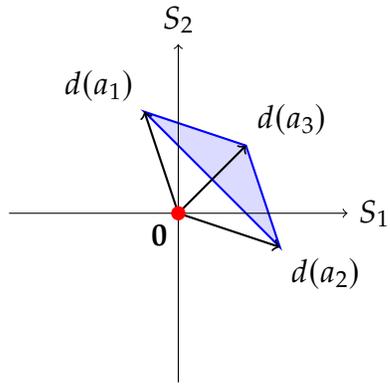
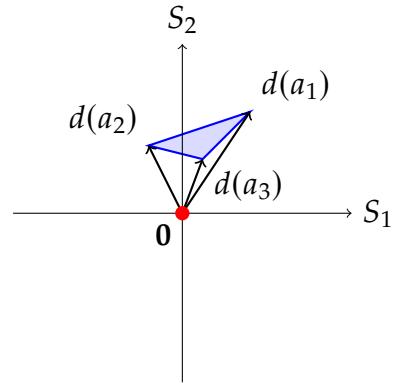
% ---- END fig_equilibria_types.tex ----

This observation suggests the following geometric classification, which serves as an organizing framework for the regimes analyzed formally in later sections:\footnote{This classification generalizes the two-model case, where the sign of the scalar drift determines whether beliefs mix or concentrate.}

\begin{enumerate}
    \item \textbf{Interior Fixed Point:} There exists $\pi\in\Delta_M$ with $\pi_j>0$ for all $j$ such that $\xi(S)=0$, corresponding to an interior balance of mean log-likelihood forces.
    \item \textbf{No Interior Fixed Point:} No $\pi\in\Delta_M$ satisfies $\xi(S)=0$, equivalently the origin $\mathbf{0}$ lies outside the convex hull $\mathrm{conv}\{d(a_1),\ldots,d(a_M)\}$.
\end{enumerate}

We now explain how each panel of Figure~\ref{fig:equilibria_types} corresponds to a distinct qualitative pattern of the mean drift field, as illustrated in Figures~\ref{fig:vector_fields}.

\begin{example}[Drift Geometry and Equilibrium Types in the Case $M=3$]
\label{ex:drift-geometry-and-equilibrium-types}

\textit{Case (a): Interior Equilibrium, Self-Confirming} (Figure~\ref{fig:equilibria_types}a $\leftrightarrow$ Figures~\ref{fig:vector_fields}a,b).
The origin lies inside the convex hull, and the mean drift points inward toward an interior fixed point $S^\star$ from all directions. For example, the $d(a_1)$ vectors are roughly aligned with the $S_1$ axis, which indicates a self-confirming effect for model 1. In the vector field, arrows converge toward $S^\star$, indicating local stability of the mean-field equilibrium.
This corresponds to a \emph{self-confirming} regime, in which beliefs fluctuate persistently around an interior configuration and the stochastic dynamics are ergodic.

\textit{Case (b): Interior Equilibrium, Self-Defeating} (Figure~\ref{fig:equilibria_types}b $\leftrightarrow$ Figures~\ref{fig:vector_fields}c,d).
The origin again lies inside the convex hull, so an interior fixed point exists. 
However, the orientation of the vertex drifts implies that the mean drift points outward from $S^\star$. For example, the $d(a_1)$ vectors are roughly opposite to $S_1$ axis, which indicates a self-defeating effect for model 1. In the vector field, arrows diverge away from the fixed point, indicating instability.
Although a balance exists at the level of expected drift, the stochastic process is pushed toward the boundary, leading to eventual model exclusion.
This is the \emph{self-defeating} regime.

\textit{Case (c): No Interior Equilibrium, Drift to a Face} (Figure~\ref{fig:equilibria_types}c $\leftrightarrow$ Figures~\ref{fig:vector_fields}e,f).
The origin lies outside the convex hull, so no interior fixed point exists.
The geometry of the drift vectors implies that the mean drift points toward a proper face of the simplex rather than toward a vertex. Here all the drift vectors go to the direction favoring model 1 and model 2, which indicates model 3 will be eventually eliminated. However, the $d(a_1)$ and $d(a_2)$ vectors are both self-defeating for their respective models. Therefore, in the vector field, arrows systematically push beliefs toward an edge, where one model is eliminated while the remaining models continue to mix.
The dynamics then reduce to a lower-dimensional system on that face.

\textit{Case (d): No Interior Equilibrium, Drift to a Vertex} (Figure~\ref{fig:equilibria_types}d $\leftrightarrow$ Figures~\ref{fig:vector_fields}g,h).
The origin lies outside the convex hull and the vertex drifts point in roughly the same direction with $S_2$ axis. As a result, the mean drift directs beliefs toward a single vertex of the simplex.
In the vector field, arrows converge toward that vertex, corresponding to eventual concentration on a single model.
    
\end{example}

\begin{remark}[Correct specification vs. misspecification]
When the true model $\theta^\star$ is contained in $\Theta$ and is identifiable, posterior concentration on $\theta^\star$ is the generic outcome. In our drift picture, this corresponds the case (d) in Figure~\ref{fig:equilibria_types}d and Figures~\ref{fig:vector_fields}g,h. The drift vectors systematically point toward the true model, so interior mixing is not expected. In fact, the possibility of an interior invariant distribution (persistent belief mixing) is a phenomenon that is naturally associated with misspecification, where no single model can uniformly fit the endogenously generated data across actions.
\end{remark}

We emphasize that these equilibria and drift patterns are defined at the level of the mean-field vector field and do not, by themselves, characterize stochastic convergence.
In the sections that follow, we analyze each regime formally and establish the corresponding asymptotic behavior of the log-odds and belief processes.

\subsection{Asymptotic Behavior with an Interior Fixed Point}

We begin with the case in which the drift geometry admits an interior fixed point.
An interior fixed point exists if
\[
\mathbf{0} \in \mathrm{int}\!\left(\mathrm{conv}\{d(a_1),\ldots,d(a_M)\}\right),
\]
equivalently, if there exists a belief $\pi^\star \in \Delta_M$ with $\pi^\star_j>0$ for all $j$ such that
\[
\sum_{j=1}^M \pi^\star_j\, d(a_j) = 0.
\]
When the drift vectors are affinely independent, such a belief $\pi^\star$ is unique.
The existence of an interior fixed point is necessary for long-run belief mixing, but, as we show below, it is not sufficient.

\vspace{0.5em}

The existence of an interior fixed point raises the question of whether the log-odds process $\{S_t\}$ remains stochastically stable in its neighborhood or instead drifts toward the boundary despite local balance.
To address this question, we study the ergodicity of the Markov process $\{S_t\}$.
Throughout this subsection, we impose a mild regularity condition ensuring that the conditional second moments of the one-step increments are uniformly bounded.

\begin{assumption}[Uniform second-moment bound]\label{assumption:uniform-second-moment-bound}
There exists $\bar\sigma^2<\infty$ such that
\[
\E\!\big[\|S_{t+1}-S_t\|^2\mid S_t=S\big]\le \bar\sigma^2\qquad\text{for all }S\in\mathbb{R}^{M-1}.
\]
\end{assumption}

Assumption~\ref{assumption:uniform-second-moment-bound} is satisfied in standard settings such as Gaussian bandits (see Appendix~\ref{app:B:multiarm}).

Our main technical framework is the stochastic stability theory for general state space Markov chains \citep{meynMarkovChainsStochastic1993}.
The key idea is to use Lyapunov functions to certify either (i) stability (positive Harris recurrence), which implies existence of an invariant distribution and ergodicity, or (ii) instability (transience), which implies drift toward the boundary.
We present two sufficient conditions based on two different Lyapunov constructions: one “angle-based” and one based on a softmax potential.

\vspace{1em}
\textit{1. Angle Condition.}

Our first sufficient condition for ergodicity is geometric in nature.
It requires that, at every point away from the fixed point, the mean drift points inward on average.
Formally, this condition can be expressed as an angle condition between the drift vector and the displacement from the fixed point.
When satisfied, this condition ensures that the log-odds process admits \textit{a unique invariant distribution} supported on the interior of the simplex. While intuitive, this condition depends on the nonlinear softmax map in $\xi(S)$ and is therefore difficult to verify ex ante.\footnote{However, as we will see in the second condition, this angle condition requires milder assumptions on the noise scale, which can be useful when the noise variance is large. }

\begin{theorem}[The First Ergodic Sufficient Condition: Angle Condition]\label{thm:ergodicity-angle-condition}
Let $\{S_t\}_{t\ge0}\subset \mathbb{R}^{M-1}$ denote the log-odds process and $\{\pi_t\}_{t \geq 0} \in \Delta_M$ denote the posterior beliefs, Suppose the process satisfies the regularity Assumption~\ref{assumption:uniform-second-moment-bound}.
Define $S^\star$ as the fixed point of the mean-field ODE. Then the posterior belief is ergodic if
\begin{equation}\label{eq:angle-condition}
\big\langle \xi(S), \, S - S^\star \big\rangle < 0.
\end{equation}
\end{theorem}

\vspace{0.5em}
\textit{2. Spectral Condition.}

We therefore turn to a more tractable sufficient condition that highlights the interaction between deterministic drift and stochastic fluctuations. Recall that the log-odds drift admits the factorization
\begin{equation}
\xi(S)
= G\big(\pi_{-M}(S)-\pi_{-M}^\star\big),
\label{eq:xi-factor-main}
\end{equation}
where $G\in\mathbb{R}^{(M-1)\times(M-1)}$ is the drift matrix whose columns are the relative log-likelihood drift vectors $d(a_i)-d(a_M)$, and $\pi_{-M}^\star$ denotes the interior fixed-point belief excluding the reference model. This factorization reveals that the mean-field dynamics admit a generalized gradient structure with respect to a canonical softmax potential $V(S)$ defined as below:

\begin{definition}[Softmax potential]
\label{def:softmax-potential-main}
Define a function $V:\mathbb{R}^{M-1}\to\mathbb{R}$ as
\begin{equation}
V(S) := \log\!\Big(1+\sum_{k=1}^{M-1} e^{S^{(k)}}\Big)-\langle S,\pi_{-M}^\star\rangle .
\end{equation}
Then it is convex with softmax as its gradient:
\[
\nabla V(S)=\pi_{-M}(S)-\pi_{-M}^\star,
\qquad
\nabla^2 V(S)
= \operatorname{diag}(\pi_{-M}(S))-\pi_{-M}(S)\pi_{-M}(S)^\top \succeq 0 .
\]
\end{definition}

Consequently, the mean-field ODE associated with Thompson Sampling can be written as
\begin{equation}
\dot S = \xi(S) = G\,\nabla V(S),
\label{eq:generalized-gradient-ode}
\end{equation}
which is a \emph{generalized gradient system}: the vector field is obtained by applying the linear operator $G$ to the gradient of a convex potential $V$. This representation has two immediate implications.

First, the geometry of the deterministic dynamics is governed by the spectral properties of the symmetric part
\[
\Sym(G) := \tfrac12(G+G^\top),
\]
while the antisymmetric part of $G$ generates rotational components along the level sets of $V$.
In particular, when $\Sym(G)\prec 0$, the deterministic flow is globally contracting toward the unique interior fixed point $S^\star$, whereas when $\Sym(G)\succ 0$, the fixed point is dynamically unstable.

Second, the stochastic stability of the log-odds process can be analyzed using $V$ as a Lyapunov function.
This formulation allows us to separate the contribution of deterministic drift forces, captured by $G$, from stochastic fluctuations, controlled by the conditional second moments of the one-step increments.

To formalize this idea, we impose a mild \emph{small-noise condition} that bounds the conditional second moments of the one-step increments relative to the strength of the contracting drift. This condition ensures that, outside a sufficiently large neighborhood of the interior fixed point, the expected decrease in the Lyapunov function induced by the drift dominates the accumulation of second-order noise terms. 

\begin{assumption}
    \label{as:small-noise}
    Let $L:=\sup_S\|\nabla^2 V(S)\|\le \tfrac12$, let $\lambda_{\min}$ be the eigenvalue of smallest magnitude of $\Sym(G)$, and let $c_\infty$ be the lower bound on the gradient norm $\|\nabla V(S)\|$ from Lemma~\ref{lem:grad-away}. We assume that the noise level in Assumption~\ref{assumption:uniform-second-moment-bound} satisfies:
    \[
    \bar\sigma^2 \;\leq\; \frac{|\lambda_{\min}|\, c_\infty^2}{L}.
    \]
\end{assumption}

\begin{theorem}[The Second Ergodic Sufficient Condition: Spectral Condition]
    \label{thm:sufficient-condition-ergodicity-transience}
    Let $G$ be the drift matrix denoted in \eqref{eq:xi-factor}. If the condition in Assumption \ref{as:small-noise} holds, we have the following sufficient condition for $\{S_t\}$ to be ergodic or transient:

    \begin{enumerate}
        \item If $\Sym(G)\prec 0$, then $\{S_t\}$ is positive Harris recurrent and hence ergodic. And the posterior belief $\{\pi_t\}$ admits a unique invariant distribution $\pi_\infty \in \mathcal{P}(\Delta_M)$.
        \item If $\Sym(G)\succ 0$, then $\{S_t\}$ is transient, and the posterior belief $\{\pi_t\}$ converges to a (generally path-dependent) proper face $\Delta_I$ of the simplex, corresponding to a random subset $I\subset\{1,\ldots,M\}$ of surviving models. Models $j\notin I$ are eliminated in the sense that $\pi_t(\theta^{(j)})\to 0$. 
    \end{enumerate}

\end{theorem}

\vspace{0.5em}

The quasi-gradient system approach for ergodicity breaks down the condition into two
parts: the deterministic part of the dynamics and the stochastic part of the dynamics. For
the first part, we not only have the global stability of the deterministic dynamics purely
induced by the drift matrix G, but also get a tractable way to predict the richer behavior
of the dynamics such as cycles and chaos.\footnote{For example, when the drift matrix $G$ has complex eigenvalues with positive real parts, the deterministic dynamics can exhibit oscillatory behavior around the fixed point, leading to cycles or even chaotic trajectories. This richer behavior is not obvious from the angle condition in Theorem~\ref{thm:ergodicity-angle-condition}.} And for the second part, we have the regularity
condition for the noise to be comparable to the drift strength. This gives us a clear underlying
structure of the asymptotic behavior of the posterior belief. Before we move on to the
next subsection, we will end this subsection with a quick characterization of the asymptotic
behavior of the posterior belief when the mean field ODE has no equilibria.

\subsection{Asymptotic Behavior without an Interior Fixed Point}

In this subsection we treat the case in which the mean-field ODE \eqref{eq:ode-log-odds} has no equilibria, i.e.,
\[
\xi(S)\;=\;\sum_{j=1}^M \pi_j(S)\,d(a_j)\;\neq\;0\qquad\text{for all }S\in\R^{M-1}.
\]
Equivalently, $\,0\notin \mathrm{conv}\{d(a_1),\ldots,d(a_M)\}$. In this case, one can always find a direction which gives a persistent drift, and thus the log-odds process can be guaranteed to be transient.

\begin{proposition}[Guaranteed transience without interior equilibrium]\label{prop:transience-no-equilibrium}
Suppose that \(0\notin \mathrm{conv}\{d(a_1),\ldots,d(a_M)\}\).
Then the following holds:
\begin{enumerate}
\item $\|S_t\|\to\infty$ a.s., hence the chain $\{S_t\}$ and $\{\pi_t\}$ are transient.
\item $\mathrm{dist}(\pi_t,\partial\Delta_M)\to 0$ a.s.
Equivalently, every limit point of $\{\pi_t\}$ lies on a proper face $\mathrm{face}(I)\subset\Delta_M$.
\end{enumerate}
\end{proposition}

\subsection{Dimensionality Reduction and Recursive Dynamics}

The results above characterize the asymptotic behavior of the log-odds process for a fixed-dimensional system.
In the transient case, however, beliefs do not diverge arbitrarily: instead, posterior mass is driven toward a proper face of the simplex, endogenously eliminating some models.
A key structural property of Thompson Sampling is that the belief dynamics on such a face inherit the same form as the original system.

\begin{proposition}[Dynamics Succession to a Face]
\label{prop:dynamics_succession}
For the log-odds process above, if \(\mathrm{dist}(\pi_t,\Delta_I):=\min_{q\in\Delta_I}|\pi_t-q|\to 0\), then the full dynamics converge to the face dynamics:
\begin{equation}
|( \Delta\pi_t )_{I}-\Delta\pi^{\text{face}}_t| \longrightarrow 0 \quad \text{and}\quad |( \Delta\pi_t )_{I^c}| \longrightarrow 0.
\end{equation}

\end{proposition}

Specifically, when posterior beliefs converge toward a face $\Delta_I$ of the simplex, models outside $I$ receive vanishing probability and no longer affect the drift.
Consequently, the effective learning dynamics are well-approximated by the induced lower-dimensional Thompson Sampling problem restricted to the surviving models in $I$.\footnote{The posterior converges to the face $\Delta_I$ asymptotically but reaches it only in the limit, as face boundaries constitute singular points of the dynamics in finite time.}
Importantly, the reduced system has the same log-odds representation, drift geometry, and stability structure as the original model space.

This dimensionality reduction property allows us to analyze multi-model learning inductively: whenever the dynamics are transient, the problem collapses to a smaller instance of itself, and the same classification logic can be applied recursively.

\paragraph{Unified Classification of Asymptotic Regimes}

The results above yield a unified classification of several economically relevant asymptotic regimes for Thompson Sampling under model misspecification.
It is important to emphasize that this classification is not exhaustive.
All conditions developed in this section are sufficient rather than necessary, and therefore the regimes identified below represent a subset of the possible long-run behaviors of posterior beliefs.

Within this scope, our analysis delivers a tractable and testable characterization of belief dynamics based on drift geometry, stability properties, and recursive dimensionality reduction.
Table~\ref{tab:multi_model_classification} summarizes the asymptotic regimes that arise when the sufficient conditions are satisfied, together with their relationship to the two-arm case.

When the sufficient conditions for interior ergodicity are met, posterior beliefs admit a unique invariant distribution supported on the interior of the simplex, corresponding to persistent belief mixing.
When these conditions fail but an interior fixed point exists, beliefs may instead be driven toward the boundary, leading to stochastic selection among models or convergence to lower-dimensional faces.
In the absence of any interior fixed point—possibly after recursive reduction—posterior beliefs converge almost surely to a single vertex, corresponding to uniform dominance.

The dimensionality reduction property established above allows these regimes to be combined recursively.
As a result, long-run outcomes may involve ergodic behavior on a face, stochastic mixtures over vertices, or mixtures of invariant distributions on faces and point masses at vertices.
Which outcome is realized depends on the realized learning path and on the stability properties of the induced dynamics at each stage.

When the sufficient conditions do not apply, more complex behavior may arise.
In particular, the generalized gradient structure of the mean-field dynamics permits the possibility of persistent cycles or more irregular trajectories under certain misspecification patterns.
While a full characterization of such dynamics is beyond the scope of this paper, exploring richer dynamical behaviors represents an important direction for future research. \footnote{One may draw insights from the evolutionary game theory literature, as the Bayesian learning process can be viewed as a replicator dynamics in stochastic settings \citep{weibullEvolutionaryGameTheory1997}.}

\begin{table}[H]
        \centering
        \renewcommand{\arraystretch}{1.2}
        \footnotesize
        \begin{tabular}{|p{2.4cm}|p{5.8cm}|p{3.5cm}|p{2.6cm}|}
        \hline
        \textbf{Category} 
        & \textbf{Defining condition} 
        & \textbf{Limiting law} 
        & \textbf{2-arm case} \\
        \hline
        \textbf{Interior \newline Ergodic} 
        & {\scriptsize $0\in\mathrm{int}\,\mathrm{conv}\{d(a_j)\}$} + \newline First/Second Ergodic Condition
        & {\scriptsize Unique invariant $\mu_{\mathrm{int}}$ on $\mathrm{int}(\Delta_M)$} 
        & \emph{Self-defeating} \\
        \hline
        \textbf{Vertex \newline Selection} 
        & {\scriptsize $0\in \mathrm{conv}\{d(a_j)\}$} + \newline First/Second Transient Condition 
        & {\scriptsize Mixture $\sum_{v\in V_\star} W_v\,\delta_{e_v}$} 
        & \emph{Self-confirming} \\
        \hline
        \textbf{Uniform \newline Dominance} 
        & {\scriptsize $0\notin \mathrm{conv}\{d(a_j)\}$} for all subfaces, recursion ends at a single vertex {\scriptsize $e_{m^\star}$} 
        & {\scriptsize Point mass $\delta_{e_{m^\star}}$} 
        & \emph{Uniform \newline dominance} \\
        \hline
        \textbf{Face-Ergodic} 
        & Push to some face {\scriptsize $\Delta_I$} (either {\scriptsize $0\notin\mathrm{conv}$} or interior unstable); on {\scriptsize $\Delta_I$}: First/Second Ergodic Condition. Models {\scriptsize $j\notin I$} are \emph{uniformly dominated}. 
        & {\scriptsize Unique invariant $\mu_I$ on $\Delta_I$} 
        & \emph{Self-defeating} (on the face) $+$ \newline \emph{uniform \newline dominance} \\
        \hline
        \textbf{Nested/Mixed \newline Outcomes} 
        & General recursive reduction: irrespective of interior fixed point existence, repeated tests on faces either (i) yield ergodic faces or (ii) descend to vertices. 
        & {\scriptsize Mixture $\displaystyle \sum_{v\in V_\star} W_v\,\delta_{e_v}+\sum_{I\in\mathcal{F}_\star} W_I\,\mu_I$} 
        & \emph{Self-confirming} $+$ \newline \emph{self-defeating} $+$ \newline \emph{uniform \newline dominance} \\
        \hline
        \end{tabular}
        \caption{Multi-model classification of TS under misspecification.}
        \label{tab:multi_model_classification}
\end{table}

\subsection{Discussion and Extensions}

We discuss two extensions of the framework: a probabilistic characterization of which asymptotic regimes arise generically in high-dimensional settings, and the role of the link function in extending the analysis to other posterior-based decision rules.

\paragraph{Generic Outcomes in High-Dimensional Learning}

While the previous sections provide sufficient conditions for ergodicity and transience, deriving necessary and sufficient conditions under arbitrary misspecification is challenging.
A complementary approach is to study which asymptotic regimes arise with high probability when we add certain structure to the model misspecification.

Proposition~\ref{prop:kernel_simplex} shows that when the drift vectors are drawn from a symmetric continuous distribution \footnote{This would arise when candidate models are sampled from a parameter space according to the same distribution.}, the probability that the induced drift geometry admits an interior fixed point declines exponentially with the dimension.
As a consequence, persistent ergodicity on the full simplex is unlikely in high-dimensional environments.
Instead, posterior learning generically exhibits recursive dimensionality reduction and convergence to some low-dimensional face or vertex of the simplex.

\begin{proposition}[Random $(m-1) \times m$ drift, kernel hits the simplex]
\label{prop:kernel_simplex}
Let $m \geq 2$. Let $A \in \mathbb{R}^{(m-1) \times m}$ have i.i.d.\ entries from a continuous distribution symmetric about $0$ (e.g., $\mathcal{N}(0,1)$). Then $\operatorname{rank}(A) = m-1$ almost surely, so $\ker(A)$ is one-dimensional. Let
\[
\Delta^m = \left\{x \in \mathbb{R}^m : x_i \geq 0, \sum_{i=1}^m x_i = 1\right\}.
\]
Then
\[
\mathbb{P}\big(\ker(A) \cap \Delta^m \neq \varnothing\big) = 2^{1-m}.
\]
\end{proposition}

\begin{remark}
    The argument only uses i.i.d.\ continuity, and symmetry about $0$; the exact distribution of the entries does not matter.
\end{remark}

\paragraph{Other Posterior-Based Decision Rules}

The geometric perspective extends to other posterior-based decision rules. The mean drift of the log-odds process takes the form $\xi(S) = \sum_a p_a(S) \cdot d(a)$, where the drift vectors $d(a)$ are fixed by the misspecification and the \emph{link function} $S \mapsto (p_a(S))_a$ is determined by the decision rule. Under Thompson Sampling, $p_a(\pi) = \sum_j \pi_j \mathbf{1}\{\phi(\theta^{(j)})=a\}$---a linear function of the posterior that, in the two-model case, reduces to the logistic $p_1(S) = \sigma(S) = e^S/(1+e^S)$.

To illustrate how different decision rules alter the link, consider the two-model case where the models disagree on the optimal action:
\begin{itemize}
    \item \emph{Myopic Bayes (greedy):} play $\argmax_a \sum_j \pi_j \, r_{\theta^{(j)}}(a)$. The action switches deterministically at a threshold $\pi^*$ where the posterior-expected rewards of the two actions are equal. In log-odds, $p_1(S) = \mathbf{1}\{S > S^*\}$ where $S^* = \log(\pi^*/(1-\pi^*))$ depends on the reward parameters.
    \item \emph{Top-$k$ Thompson Sampling:} sample $k$ models independently from $\pi$ and play the majority action. This gives $p_1(S) = \P(\mathrm{Bin}(k, \sigma(S)) > k/2)$, a steepened sigmoid that approaches the step function as $k \to \infty$.
    \item \emph{$\epsilon$-greedy:} with probability $\epsilon$ explore uniformly, otherwise play greedy. This gives $p_1(S) = (1-\epsilon)\mathbf{1}\{S > S^*\} + \epsilon/2$, bounding the action probabilities away from $0$ and $1$.
\end{itemize}
Each variant preserves the drift--noise decomposition of the log-odds process but changes the weight placed on each drift vector at a given belief state, altering the exploration intensity and potentially the qualitative asymptotic behavior.

\section{Conclusion}

This paper studies the long-run consequences of misspecified dynamic learning under posterior sampling. Our central finding is that endogenous experimentation changes the asymptotic objects of learning. Under misspecification, posterior beliefs may concentrate on a correct model, concentrate on an incorrect one, select randomly among multiple self-confirming outcomes, or fail to concentrate altogether, converging instead to a nondegenerate stationary distribution. These regimes imply sharply different long-run action frequencies and welfare. In particular, misspecification can generate persistent policy distortion through ongoing stochastic mixing that survives even with infinite data.

These results have practical implications for settings where decision makers rely on parametric models they know to be approximate. A firm using posterior-sampling-based pricing can diagnose which regime applies by checking whether its candidate demand models disagree on optimal pricing and computing the evidence direction at each price. If the configuration is self-defeating, the firm should expect persistent price variation that does not diminish with more data---a form of policy instability that is a consequence of the learning algorithm, not of environmental change. More broadly, our classification provides a diagnostic tool: for any misspecified bandit environment, the sign pattern of the expected log-likelihood ratio drifts determines whether learning will converge, and if so, to what.

Our analysis has several limitations that point to important open questions. First, we restrict attention to finite model classes; extending the stochastic-stability framework to continuous parameter spaces would significantly broaden applicability. Second, the multi-model results provide sufficient conditions for ergodicity and transience but not sharp necessary-and-sufficient conditions; closing this gap would strengthen the classification. Third, our detailed illustrations use Gaussian reward distributions; while the framework applies to general parametric families, verifying the conditions in specific non-Gaussian applications requires case-by-case analysis.

The broader lesson is that under misspecification, the standard intuition that Bayesian learning eventually ``gets it right'' or at least settles down can fail in a fundamental way. When posterior-based decision rules generate the data from which they learn, the feedback between beliefs, actions, and evidence can sustain persistent uncertainty and policy variation indefinitely. Understanding when and why this occurs is essential for the design of robust learning algorithms in economic environments.

\subsection*{Data Availability Statement}
This paper is primarily theoretical. The simulation code used to generate the numerical illustrations in Appendix~\ref{app:simulations} is available from the authors upon request.

\newpage
\setlength{\bibsep}{0.5ex}
\bibliographystyle{plainnat}
\bibliography{TSmis_ref_chicago}

\newpage

\appendix
\renewcommand{\thesection}{\Alph{section}}

\addcontentsline{toc}{section}{Appendices}
\section*{Appendices}

\section{Log-Odds Dynamics and Markov Properties}
\label{app:A}

In this appendix, we provide a complete analysis of the log-odds representation, establish the Markov property of the belief process under Thompson Sampling, and characterize the dynamical system structure that underlies our main results.

\subsection{Log-Odds Dynamics}
\label{app:A:logodds-dynamics}
We first define the log-odds vector $S_t$, which provides a global chart for the belief simplex via the softmax inverse map. Since this mapping is smooth and invertible, the dynamics of $\pi_t$ can be easily recovered from $S_t$. However, the dynamics of $S_t$ exhibit a simpler linear additive recurrence structure, which makes our analysis more tractable. It is convenient to view the posterior dynamics as a process on the simplex manifold, with the log-odds vector serving as global chart coordinates.

\begin{definition}[Log-Odds Vector]
Define the log-odds vector:
\[
S_t^{(k)} := \log\frac{\pi_t(\theta^{(k)})}{\pi_t(\theta^{(M)})},
\quad k=1,\ldots,M-1.
\]
\end{definition}

Notice that the vector $S_t\in \mathbb{R}^{M-1}$ uniquely determines $\pi_t$:
\begin{equation}
\pi_t(\theta^{(k)}) = \frac{\exp(S_t^{(k)})}{1 + \sum_{j=1}^{M-1}\exp(S_t^{(j)})}, \quad k=1,\ldots,M-1,
\quad
\pi_t(\theta^{(M)}) = \frac{1}{1 + \sum_{j=1}^{M-1}\exp(S_t^{(j)})},
\label{eq:pi_S}
\end{equation}
We will simply denote $\pi_t = \psi(S_t)$, so the $\psi : \mathbb{R}^{M-1} \mapsto \Delta_M$ is the softmax function.

Recall Bayes' rule:
\[
S_{t+1}^{(k)} = \log\frac{\pi_{t+1}(\theta^{(k)})}{\pi_{t+1}(\theta^{(M)})}
= \log\frac{\pi_t(\theta^{(k)}) f_{\theta^{(k)}}(r_t\mid A_t)}{\pi_t(\theta^{(M)}) f_{\theta^{(M)}}(r_t\mid A_t)}
= S_t^{(k)} + \underbrace{\log\frac{f_{\theta^{(k)}}(r_t\mid A_t)}{f_{\theta^{(M)}}(r_t\mid A_t)}}_{=:Z_t^{(k)}}.
\]

Thus, we obtain the stochastic recurrence $\mathbf{S}_{t+1} = \mathbf{S}_t + \mathbf{Z}_t,$
where $\mathbf{Z}_t\in \mathbb{R}^{M-1}$ is the random increment vector. Given $\mathbf{S}_t$, the law of $\mathbf{Z}_t$ is
\begin{equation}
P\bigl(\mathbf{Z}_t\in B | \mathbf{S}_t\bigr)
= \sum_{m=1}^M \pi_t(\theta^{(m)})
\sum_{a\in \mathcal{A}} \mathbb{I}\{a=\phi(\theta^{(m)})\}
\int
\mathbb{I}\Bigl\{
\log\frac{f_{\theta^{(1)}}}{f_{\theta^{(M)}}},\ldots
\in B
\Bigr\}
f_{\theta^*}(r|a)\,dr.
\label{eq:law-of-increment}
\end{equation}

where $B$ is a Borel set in $\mathbb{R}^{M-1}$. Now let us further decompose the increment $\mathbf{Z}_t$ into a drift, $\xi(\mathbf{S}_t) = \mathbb{E}[\mathbf{Z}_t | \mathbf{S}_t]$, plus a mean zero martingale increment, $\varepsilon_t$. Then, the log-odds process $(\mathbf{S}_t)_{t\geq 0}$ can be written as:
\begin{equation}
\mathbf{S}_{t+1} = \mathbf{S}_t + \xi(\mathbf{S}_t) + \varepsilon_t,
\label{eq:log-odds-process}
\end{equation}

The log-odds drift has a very intuitive interpretation. If we denote the vertex drift vectors $\{ d(a_j) \}_{j=1}^M$ as the expected log-likelihood ratio vector when always playing action $\phi(\theta^{(j)})$ under the true environment $f_{\theta^*}$, then the log-odds drift is a weighted sum of the vertex drift vectors, with the weights being the posterior probabilities. \footnote{Another way to see this is that the \eqref{eq:drift_vector} is the KL divergence difference between the models (See more details in \citet{frickBeliefConvergenceMisspecified2023})}

\begin{equation}
\xi(S) = \sum_{j=1}^M \pi_j \cdot d(a_j),
\quad \text{with } \pi = \psi(S).
\label{eq:drift_vector_pi}
\end{equation}

where

\begin{equation}
    d(a_j) := \mathbb{E}_{r \sim f_{\theta^*}(\cdot \mid \phi(\theta^{(j)}))} \left[ \log \left( \frac{f_{\theta^{(1)}}}{f_{\theta^{(M)}}}, \ldots, \frac{f_{\theta^{(M-1)}}}{f_{\theta^{(M)}}} \right) \right] \in \mathbb{R}^{M-1}.
    \label{eq:drift_vector_app}
\end{equation}

\subsection{Mean-Field ODE Approximation}
\label{app:A:dynamics}

The belief evolution can be represented as a discrete-time stochastic process in $\mathbb{R}^{M-1}$ with additive drift-noise decomposition. This section provides the formal characterization of the dynamics as a general-state-space Markov chain.

The mean-field limit of the log-odds process is described by a system of ordinary differential equations (ODEs), which provides valuable intuition for the qualitative behavior of the discrete-time dynamics. Although our primary analysis focuses on the discrete-time setting, the continuous-time approximation reveals the geometric structure underlying belief evolution and facilitates the identification of stable and unstable equilibria.\footnote{Related continuous-time approximations have been developed in the literature on Thompson Sampling. See, for example, \citet{fanDiffusionApproximationsThompson2025} for a diffusion-based analysis in the correctly specified setting.}

The mean field ODE for the log-odds process is given by:
\begin{equation}
    \dot{S}(t) = \xi(S(t)) = D\psi(S(t)).
    \label{eq:ode-log-odds}
\end{equation}  
where $D := [d(a_1), \ldots, d(a_{M})] \in \mathbb{R}^{(M-1) \times M}$ is the full drift matrix collecting all $M$ vertex drift vectors as columns. In the main text, we also use the reduced matrix $G := [d(a_1)-d(a_M), \ldots, d(a_{M-1})-d(a_M)] \in \mathbb{R}^{(M-1) \times (M-1)}$, which arises from the factorization $\xi(S) = G(\pi_{-M}(S) - \pi_{-M}^\star)$ when an interior fixed point exists.

To derive the induced dynamics of the posterior \( \pi_t = \psi(S_t) \in \Delta_M \), we apply the chain rule:
\begin{equation}
\frac{d\pi_t}{dt} = J_{\psi}(S_t) \cdot \frac{dS_t}{dt} 
= J_{\psi}(S_t) \cdot \xi(S_t).
\label{eq:ode_pi}
\end{equation}
This yields a deterministic dynamical system over the simplex \( \Delta_M \).\footnote{Such ODE system for the posterior belief can also be treated as a type of replicator game dynamical system. (See \citet{shaliziDynamicsBayesianUpdating2009} for more details)  However, our framework focusing on the log-odds space with a reference model gives a simpler and more transparent analysis of the asymptotic behavior of the posterior belief. }
The Jacobian \( J_{\psi}(\mathbf{S}_t) \in \mathbb{R}^{M \times (M-1)} \) has entries:
\begin{equation}
\frac{\partial \pi^{(i)}}{\partial S^{(k)}} = 
\begin{cases}
\pi^{(i)} (1 - \pi^{(i)}), & i = k, \\
- \pi^{(i)} \pi^{(k)}, & i \neq k, \\
- \pi^{(M)} \pi^{(k)}, & i = M,
\end{cases}
\label{eq:jacobian_psi}
\end{equation}
where \( \pi^{(M)} = 1 - \sum_{j=1}^{M-1} \pi^{(j)} \).

\section{Gaussian Bandit Examples}
\label{app:B}

\paragraph{Binary Gaussian Bandit Example}
\label{app:B:gaussian}

\begin{example}[Posterior Dynamics in Binary Gaussian Bandits]
\label{example:binary-gaussian}
Consider the two-arm Gaussian bandit with models defined in $\Theta = \{\nu, \gamma\}$ where $\nu = (\nu_1, \nu_2)$ and $\gamma = (\gamma_1, \gamma_2)$. Let $S_t = \log[\pi_t/(1-\pi_t)]$ denote the log-odds of model $\nu$ at time $t$, and define the optimal action mapping $\phi(\theta) = \argmax_{i \in \{1,2\}} \theta_i$. Then:

Given the filtration $\mathcal{F}_{t-1}$, the log-likelihood ratio follows a mixture distribution:
\begin{equation}
        Z_t \mid \mathcal{F}_{t-1} \sim
        \begin{cases}
            \mathcal{N}(\Delta_{\phi(\nu)}, (\nu_{\phi(\nu)}-\gamma_{\phi(\nu)})^2) & \text{with probability } \pi_{t-1}, \\
            \mathcal{N}(\Delta_{\phi(\gamma)}, (\nu_{\phi(\gamma)}-\gamma_{\phi(\gamma)})^2) & \text{with probability } 1-\pi_{t-1}.
        \end{cases}
        \label{eq:log_likelihood_ratio}
\end{equation}
where $\Delta_i = \frac{(\nu_i - \gamma_i)[2g(i) - \nu_i - \gamma_i]}{2}$ for each action $i \in \{1,2\}$.
\end{example}

\begin{proof}
    The log-odds update follows directly from Bayesian updating and the definition $S_t = \log[\pi_t/(1-\pi_t)]$. For the expected log-likelihood ratios, when $A_t = i$, we have:
    \begin{align*}
    \Delta_i &= \mathbb{E}[\log f_{\nu_i}(R_t)/f_{\gamma_i}(R_t) \mid A_t = i] \\
    &= \mathbb{E}\left[\frac{1}{2}[(R_t-\gamma_i)^2 - (R_t-\nu_i)^2] \mid A_t = i\right] \\
    &= \frac{1}{2}\mathbb{E}[(\nu_i - \gamma_i)(2R_t - \nu_i - \gamma_i) \mid A_t = i] \\
    &= \frac{(\nu_i - \gamma_i)[2g(i) - \nu_i - \gamma_i]}{2}
    \end{align*}
    where we used the fact that $\mathbb{E}[R_t \mid A_t = i] = g(i)$. The conditional distribution follows from the Gaussian assumption and Thompson Sampling's action selection rule, and the expected drift is immediate from the law of total expectation.
\end{proof}

\paragraph{Multi-Arm Gaussian Bandit Example}
\label{app:B:multiarm}
\begin{example}[Gaussian Bandit Case]
    Consider the special case where the reward model under each candidate parameter is Gaussian with known variance: $r_t \mid A_t = a,\, \theta = \theta^{(k)} \sim \mathcal{N}(\mu^{(k)}_a, \sigma^2),$ for all $k \in \{1, \ldots, M\}$. Then, conditional on $A_t=a$, the distribution of $\mathbf{Z}_t$ is Gaussian:
    \[
    \mathbf Z_t \mid A_t=a \;\sim\; \mathcal{N}\!\Big(\alpha_a+\frac{\mu_a^\ast}{\sigma^2}v_a,\; \frac{1}{\sigma^2}v_a v_a^\top\Big).
    \]
    where $v_a=(\mu_a^{(1)}-\mu_a^{(M)},\ldots,\mu_a^{(M-1)}-\mu_a^{(M)})^\top$ and $(\alpha_a)_k=\frac{1}{2\sigma^2} (\mu_a^{(M)})^2-(\mu_a^{(k)})^2$.

    Hence, the joint increment vector $\mathbf{Z}_t = (Z_t^{(1)}, \ldots, Z_t^{(M-1)})$ is a mixture of multivariate Gaussians, where each component corresponds to a choice of model $m \in \{1, \ldots, M\}$, action $a \in \mathcal{A}$, and the resulting conditional distribution under $r_t \sim \mathcal{N}(\mu^*_a, \sigma^2)$. The mixture weights are given by:
    \[
    w_{m,a}(\mathbf{S}_t) = \pi_t(\theta^{(m)}) \cdot \mathbb{I}\{a = \phi(\theta^{(m)})\}.
    \]
    \label{example:gaussian-bandit-case}
\end{example}

\begin{proof}
    Let $\mu_a^{(k)}$ denote the mean under model $k$, $\mu_a^{(M)}$ the mean of the reference model, and $\mu_a^\ast$ the true mean (under $\theta^\ast$). For equal variances $\sigma^2$, the log–likelihood ratio simplifies to an affine form in $r_t$:
    \[
    Z_t^{(k)} \;=\; \log\frac{f_{\theta^{(k)}}(r_t\mid a)}{f_{\theta^{(M)}}(r_t\mid a)}
    \;=\; \frac{\mu_a^{(k)}-\mu_a^{(M)}}{\sigma^2}\, r_t
    \;+\; \frac{(\mu_a^{(M)})^2-(\mu_a^{(k)})^2}{2\sigma^2}\,.
    \]
    Define the $(M-1)$-vector and offset
    \[
    v_a \in \mathbb{R}^{M-1},\qquad (v_a)_k := \mu_a^{(k)}-\mu_a^{(M)},\qquad
    \alpha_a \in \mathbb{R}^{M-1},\qquad (\alpha_a)_k := \frac{(\mu_a^{(M)})^2-(\mu_a^{(k)})^2}{2\sigma^2}.
    \]
    Then, conditional on $A_t=a$, we have
    \[
    \mathbf Z_t \;=\; \alpha_a \;+\; \frac{1}{\sigma^2}\, v_a\, r_t,
    \qquad r_t \mid A_t=a \sim \mathcal N(\mu_a^\ast,\sigma^2).
    \]
    Hence $\mathbf Z_t \mid A_t=a$ is multivariate Gaussian:
    \[
    \mathbf Z_t \mid A_t=a \;\sim\; \mathcal N\!\Big(\;
    m_a,\;\Sigma_a\;\Big),
    \]
    with mean and covariance
    \[
    m_a \;=\; \alpha_a \;+\; \frac{\mu_a^\ast}{\sigma^2}\, v_a,
    \qquad
    \Sigma_a \;=\; \frac{1}{\sigma^2}\, v_a v_a^\top.
    \]
    Entrywise, for $k,\ell\in\{1,\ldots,M-1\}$:
    \[
    \mathbb E[Z_t^{(k)}\mid A_t=a] \;=\; \frac{1}{2\sigma^2}\Big[(\mu_a^{(M)})^2-(\mu_a^{(k)})^2 - 2\mu_a^\ast(\mu_a^{(M)}-\mu_a^{(k)})\Big],
    \]
    \[
    \mathrm{Var}(Z_t^{(k)}\mid A_t=a) \;=\; \frac{(\mu_a^{(k)}-\mu_a^{(M)})^2}{\sigma^2},\qquad
    \mathrm{Cov}(Z_t^{(k)},Z_t^{(\ell)}\mid A_t=a) \;=\; \frac{(\mu_a^{(k)}-\mu_a^{(M)})(\mu_a^{(\ell)}-\mu_a^{(M)})}{\sigma^2}.
    \]

\end{proof}

\section{Regret Analysis}
\label{app:F}

This appendix provides formal regret bounds for both two-arm and multi-model settings, complementing the asymptotic behavior analysis in the main text.

\paragraph{Two-Arm Regret Analysis}
\label{app:two_arm_regret}

This appendix collects the formal results underlying the regret analysis in Section~\ref{sec:two_arm_regret}. 
We first characterize the limiting action distributions induced by posterior convergence, and then derive the corresponding limiting average regret in each case.

\begin{proposition}[Limiting Action Distributions]
\label{prop:limiting_action_distributions}
The limiting action distribution under Thompson Sampling depends on the posterior convergence behavior as follows:

\begin{enumerate}
    \item[(i)] \textbf{Agreement Case} ($\phi(\nu) = \phi(\gamma) = 1$): 
    \[\lim_{t \to \infty} \mathbb{P}(A_t = 1) = 1, \quad \lim_{t \to \infty} \mathbb{P}(A_t = 2) = 0\]
    regardless of the sign of $\Delta_1$.

    \item[(ii)] \textbf{Self-Confirming Case} ($\Delta_1 > 0 > \Delta_2$):
    \[\lim_{t \to \infty} \mathbb{P}(A_t = 1) = p^*, \quad \lim_{t \to \infty} \mathbb{P}(A_t = 2) = 1 - p^*\]
    where $p^* = \frac{-\Delta_2}{\Delta_1 - \Delta_2}$.

    \item[(iii)] \textbf{Uniform Dominance Case} ($\Delta_1, \Delta_2 > 0$):
    \[\lim_{t \to \infty} \mathbb{P}(A_t = 1) = 1, \quad \lim_{t \to \infty} \mathbb{P}(A_t = 2) = 0\]

    \item[(iv)] \textbf{Self-Defeating Case} ($\Delta_1 < 0 < \Delta_2$):
    \[\lim_{t \to \infty} \mathbb{P}(A_t = 1) = \alpha^*, \quad \lim_{t \to \infty} \mathbb{P}(A_t = 2) = 1 - \alpha^*\]
    where $\alpha^* = \int_0^1 \pi \, d\pi_{\infty}(\pi) \in (0,1)$ is the expected posterior probability under the stationary distribution.
\end{enumerate}
\end{proposition}

\begin{proposition}[Limiting Average Regret]
\label{prop:limiting_average_regret}
Without loss of generality, assume $g(1) \geq g(2)$, so the oracle policy always selects action 1. The limiting average regret of Thompson Sampling is:

\begin{enumerate}
    \item[(i)] \textbf{Agreement Case} ($\phi(\nu) = \phi(\gamma) = 1$): 
    \[\lim_{T \to \infty} \frac{1}{T} \text{Regret}_T = 0\]
    Thompson Sampling learns to select the truly optimal action.

    \item[(ii)] \textbf{Agreement Case} ($\phi(\nu) = \phi(\gamma) = 2$): 
    \[\lim_{T \to \infty} \frac{1}{T} \text{Regret}_T = g(1) - g(2) > 0\]
    Thompson Sampling consistently selects the suboptimal action.

    \item[(iii)] \textbf{Self-Confirming Case} ($\Delta_1 > 0 > \Delta_2$):
    \[\lim_{T \to \infty} \frac{1}{T} \text{Regret}_T = 
    \begin{cases}
    0 & \text{with probability } p^* \\
    g(1) - g(2) & \text{with probability } 1-p^*
    \end{cases}\]
    where $p^* = \frac{-\Delta_2}{\Delta_1 - \Delta_2}$. The algorithm converges to one of two deterministic policies.

    \item[(iv)] \textbf{Uniform Dominance Case} ($\Delta_1, \Delta_2 > 0$):
    \[\lim_{T \to \infty} \frac{1}{T} \text{Regret}_T = 0\]
    Thompson Sampling learns to select the truly optimal action.

    \item[(v)] \textbf{Self-Defeating Case} ($\Delta_1 < 0 < \Delta_2$):
    \[\lim_{T \to \infty} \frac{1}{T} \text{Regret}_T = (1-\alpha^*)[g(1) - g(2)] > 0\]
    where $\alpha^* = \int_0^1 \pi \, d\pi_{\infty}(\pi)$ is the probability of selecting action 1 in the long run.
\end{enumerate}
\end{proposition}

\paragraph{Regret Analysis in the Multi-Model Setting}

The regret analysis for the multi-model setting is similar to the two-arm case. We will first give the limiting action distributions and then the limiting average expected regret. Let $a^\star \in \arg\max_{a\in\mathcal A} g(a)$ denote the oracle action under the true environment $g$. 

For any action $a \in \mathcal{A}$, the expected per-period regret is:
\[
r(a) := \mathbb{E}[R_t^{\star} - R_t \mid A_t = a] = g(a^\star) - g(a) \geq 0.
\]

Given a posterior vector $\pi\in\Delta_M$, Thompson Sampling selects action $A(\pi)$ by first sampling a model index $J\sim \pi$ and then playing $A(\pi)=\phi(\theta^{(J)})$. 
Equivalently, the induced action mixture at belief $\pi$ is
\[
p_a(\pi)\ :=\ \sum_{j=1}^M \pi_j\,\mathbf{1}\{\phi(\theta^{(j)})=a\}\,,\qquad a\in\mathcal A,
\]
and the one–step expected regret at $\pi$ is $\sum_{a} p_a(\pi)\,r(a)$.

\begin{proposition}[Limiting action distributions]\label{prop:MM_limiting_action_dist}
Under Theorem (defined above), let the terminal mixture be
\[
\Law(\pi_t)\ \Rightarrow\ \sum_{v\in V_\star} W_v\,\delta_{e_v}\ +\ \sum_{I\in\mathcal F_\star} W_I\,\mu_I,
\]
where $e_v$ is the $v$-th vertex, $\mu_I$ is the unique invariant law on the ergodic face $\Delta_I$, and $W_\cdot$ are basin-entry probabilities. 
Then the limiting action probabilities exist and equal
\[
\alpha(a)\ :=\ \lim_{t\to\infty}\Pr(A_t=a)\ 
=\ \underbrace{\sum_{v\in V_\star} W_v\,\mathbf{1}\{\phi(\theta^{(v)})=a\}}_{\text{terminal vertices}}
\ +\ \underbrace{\sum_{I\in\mathcal F_\star} W_I\,\int_{\Delta_I} p_a(\pi)\, \mu_I(d\pi)}_{\text{ergodic faces}}.
\]
\end{proposition}

\begin{proposition}[Limiting average expected regret by category]\label{prop:MM_limiting_regret}
Let $\overline R:=\lim_{T\to\infty} \frac{1}{T}\sum_{t=1}^T \mathbb{E}[r(A_t)]$ denote the limiting average expected regret whenever the limit exists. Then:
\begin{enumerate}
\item \textbf{Interior Ergodic} (global ergodicity on $\mathrm{int}(\Delta_M)$): 
\[
\alpha_{\mathrm{int}}(a)=\int_{\Delta_M} p_a(\pi)\,\mu_{\mathrm{int}}(d\pi),\qquad 
\overline R\ =\ \sum_{a\in\mathcal A}\alpha_{\mathrm{int}}(a)\,r(a).
\]
\item \textbf{Vertex Selection} (multi self-confirming): 
\[
\alpha(a)\ =\ \sum_{v\in V_\star} W_v\,\mathbf{1}\{\phi(\theta^{(v)})=a\},\qquad 
\overline R\ =\ \sum_{v\in V_\star} W_v\, r\!\big(\phi(\theta^{(v)})\big).
\]
\item \textbf{Uniform Dominance} (global vertex attractor $e_{m^\star}$): 
\[
\alpha(a)\ =\ \mathbf{1}\{\phi(\theta^{(m^\star)})=a\},\qquad 
\overline R\ =\ r\!\big(\phi(\theta^{(m^\star)})\big).
\]
\item \textbf{Face-Ergodic} (restricted self-defeating on $\Delta_I$ with off-face elimination): 
\[
\alpha_I(a)\ =\ \int_{\Delta_I} p_a(\pi)\,\mu_I(d\pi),\qquad 
\overline R\ =\ \sum_{a\in\mathcal A}\alpha_I(a)\,r(a).
\]
\item \textbf{Nested/Mixed Terminal Outcomes} (recursive mixture over faces and vertices):
\[
\alpha(a)\ =\ \sum_{v\in V_\star} W_v\,\mathbf{1}\{\phi(\theta^{(v)})=a\}
\ +\ \sum_{I\in\mathcal F_\star} W_I\,\int_{\Delta_I} p_a(\pi)\,\mu_I(d\pi),
\]
\[ 
\overline R\ =\ \sum_{v\in V_\star} W_v\, r\!\big(\phi(\theta^{(v)})\big)
\ +\ \sum_{I\in\mathcal F_\star} W_I\,\sum_{a\in\mathcal A}\Big(\int_{\Delta_I} p_a(\pi)\,\mu_I(d\pi)\Big)\,r(a).
\]
\end{enumerate}
Moreover, in all cases the limit exists and equals the expectation of the expected per-period regret $r(A)$ with $A\sim \alpha(\cdot)$, i.e., $\overline R = \sum_{a\in\mathcal A} \alpha(a) \cdot r(a)$.
\end{proposition}

From the above proposition, we can see that the Interior Ergodic and Face-Ergodic cases produce stationary randomized policies induced by the invariant law(s), however the Vertex Selection and Uniform Dominance yield degenerate policies. Since our oracle policy is always selecting the optimal action, when the policy is random, the expected regret will never converge to zero. The only possible zero expected regret is when the policy is degenerate under the uniform dominance case, where the belief converge to the model whose optimal action is the optimal action under the true data generating process. Unless we add extra structure to the model misspecification, there is no guarantee (even in probabilistic sense) to achieve a decreasing average expected regret in the long run.

\begin{table}[H]
\centering
\renewcommand{\arraystretch}{1.2}
\footnotesize
\begin{tabular}{|p{3cm}|p{4.5cm}|p{7.5cm}|}
\hline
\textbf{Category} 
& \textbf{Limiting action probabilities} 
& \textbf{Limiting average expected regret} \\
\hline
\textbf{Interior Ergodic} 
& {\scriptsize $\alpha_{\mathrm{int}}(a)=\int p_a(\pi)\, \mu_{\mathrm{int}}(d\pi)$} 
& {\scriptsize $\displaystyle \sum_a \alpha_{\mathrm{int}}(a)\, [g(a^\star)-g(a)]$} \\
\hline
\textbf{Vertex Selection} 
& {\scriptsize $\alpha(a)=\sum_{v\in V_\star} W_v\,\mathbf{1}\{\phi(\theta^{(v)})=a\}$} 
& {\scriptsize $\displaystyle \sum_{v\in V_\star} W_v\, [g(a^\star)-g(\phi(\theta^{(v)}))]$} \\
\hline
\textbf{Uniform \newline Dominance} 
& {\scriptsize $\alpha(a)=\mathbf{1}\{\phi(\theta^{(m^\star)})=a\}$} 
& {\scriptsize $g(a^\star)-g(\phi(\theta^{(m^\star)}))$} \\
\hline
\textbf{Face-Ergodic} 
& {\scriptsize $\alpha_I(a)=\int_{\Delta_I} p_a(\pi)\, \mu_I(d\pi)$} 
& {\scriptsize $\displaystyle \sum_a \alpha_I(a)\,[g(a^\star)-g(a)]$} \\
\hline
\textbf{Nested/Mixed Outcomes} 
& {\scriptsize $\alpha(a)=\sum_{v} W_v \mathbf{1}\{\phi(\theta^{(v)})=a\}+\sum_I W_I \int_{\Delta_I} p_a(\pi)\,\mu_I(d\pi)$} 
& {\scriptsize $\displaystyle \sum_{v} W_v\, r(\phi(\theta^{(v)}))+\sum_I W_I \sum_a \Big(\int_{\Delta_I} p_a(\pi)\,\mu_I(d\pi)\Big) r(a)$} \\
\hline
\end{tabular}
\caption{Limiting action distributions and average expected regret across multi-model categories. 
Here $p_a(\pi)=\sum_j \pi_j\,\mathbf{1}\{\phi(\theta^{(j)})=a\}$, $a^\star=\arg\max_a g(a)$, $r(a)=g(a^\star)-g(a)$ is the expected per-period regret, $V_\star$ are terminal vertices, and $\mathcal F_\star$ are ergodic terminal faces.}
\label{tab:mm_regret_table}
\end{table}

\section{Proofs of Main Classification Theorems}
\label{app:D}

This appendix contains the proofs of the four main classification theorems characterizing Thompson Sampling behavior under misspecification.

\subsection{Proof of Theorem \ref{thm:agreement_case}}
\label{app:D:agreement_proof}

\begin{assumption}
We assume the following: 
\begin{enumerate}
    \item initial log-likelihood does not prefer any arm, i.e., $S_0 = 0$;
    \item Under the true reward distribution given action $i$ $g_i$, the log-likelihood
    \begin{align*}
        \E_{g_i} \left[ \left( \log \frac{f_{\nu_i}(r)} {f_{\gamma_i}(r)} \right)^2 \right] < \infty;
    \end{align*}
\end{enumerate}\label{as:two-arm}
\end{assumption}

\begin{proof}
Assume both models agree on the optimal action:
\[
\phi(\nu) = \phi(\gamma) = 1.
\]
Thus, Thompson Sampling always selects action 1: \( A_t = 1 \) for all \( t \). 
Let \( R_t \) denote the reward observed at time \( t \) upon performing action 1, and define the log-likelihood ratio process as
\[
S_t := \sum_{s=1}^t \log \frac{f_{\nu_1}(R_s)}{f_{\gamma_1}(R_s)}.
\]
Since rewards are assumed to be independent and only arm 1 is being played, the collection of random variables $(Z_s)_{s \in \mathbb N}$ defined as
\[
Z_s := \log \frac{f_{\nu_1}(R_s)}{f_{\gamma_1}(R_s)}
\]
are i.i.d.~and integrable by assumption.
Denote their expectations by
\[
\Delta_1 := \mathbb{E}_{g_1} \left[ \log \frac{f_{\nu_1}(R)}{f_{\gamma_1}(R)} \right].
\]

Recall the filtration $\mathbb F = (\mathcal F_t)_{t \in \mathbb N}$ generated by previous actions and rewards in \eqref{eq:filtration}.
Then \( \{S_t\} \) is $\mathbb F$-adapted and satisfies
\[
\mathbb{E}[S_{t+1} \mid \mathcal{F}_t] = S_t + \Delta_1.
\]
We analyze three cases based on the sign of \( \Delta_1 \):

\paragraph{\underline{Case (i): $\Delta_1 > 0$}}

The process \( S_t \) is a \emph{submartingale} under \( \mathbb{P} \) since
\[
\mathbb{E}[S_{t+1} \mid \mathcal{F}_t] = S_t + \Delta_1 \geq S_t.
\]
Since \( S_t \) is a sum of i.i.d.~variables with positive mean, the strong law of large numbers implies
\[
\lim_{t \to \infty} \frac 1 t S_t = \Delta \quad \text{$\mathbb P$-a.s.},
\]
and consequently, $S_t$ diverges a.s.
The posterior belief evolves as:
\[
\pi_t = \frac{1}{1 + \frac{1 - \pi_0}{\pi_0} e^{-S_t}} \to 1 \quad \text{a.s. as $t \to \infty$.}
\]

\paragraph{\underline{Case (ii): $\Delta_1 = 0$}}

In this case, the log-odds process $S$ is a symmetric random walk.
So, for any $z \in (0,1)$, we have
\begin{align*}
    \Pr(\pi_t \leq z) = \Pr \left( S_t \leq - \log \left( \frac 1 z - 1 \right) \right) = \Pr \left( \frac{S_t}{\sqrt t} \leq \frac {- \log \left( \frac 1 z - 1 \right)} {\sqrt t} \right).
\end{align*}
Now, by the independence of the reward and the assumed second-moment bound, we know that a central limit theorem holds.
Furthermore, by Slutsky's theorem, we have
\begin{align*}
    \frac{S_t}{\sqrt t} + \frac {- \log \left( \frac 1 z - 1 \right)} {\sqrt t} \stackrel{\text{(d)}}{\longrightarrow} \mathcal N(0, \sigma^2).
\end{align*}
Since $0$ is a point of continuity for the normal distribution, we take the limit on both sides and conclude that
\begin{align*}
    \lim_{t \to \infty} \Pr(\pi_t \leq z) = \Pr(\mathcal N(0, \sigma^2) \leq 0) = \frac 1 2.
\end{align*}
By the fact that $z$ is chosen arbitrarily and $0 \leq \pi_t \leq 1$, this shows that $\pi_t$ converges weakly to a Bernoulli with parameter $1/2$.

\paragraph{\underline{Case (iii): $\Delta_1 < 0$}}

This case follows from Case (i) by symmetry, and the proof is complete. 

\end{proof}

\subsection{Proof of Theorem \ref{thm:self_confirming}}
\label{app:D:self_confirming_proof}

\begin{proof}
Writing $g(s):=\E[Z_t\mid S_{t-1}=s]=\Delta_{2}+(\Delta_{1}-\Delta_{2})\sigma(s)$, we have $g(s)\,{\rm sign}(s-s^\star) > 0$ for all $s \neq s^*$ and $s^\star:=\log\!\tfrac{p^{\star}}{1-p^{\star}}$.
Note that the conditional mean increment of $S_t$ points away from $s^\star$.
We want to invoke \citet[Theorem 8.4.2]{meynMarkovChainsStochastic1993} and conclude that the Markov chain $S$ is transient, i.e., there exists a countable covering of $\mathbb R$ such that the expected occupation time in each set is uniformly bounded.
In particular, we can conclude that, for any compact set $K \subset \mathbb R$, $S_t$ will leave $K$ almost surely in finite time and $|S_t| \to \infty$ as $t \to \infty$.

Now, we check the conditions for \citet[Theorem 8.4.2]{meynMarkovChainsStochastic1993} by constructing the appropriate bounded, Borel-measurable function $V: \mathbb R \to \mathbb R$.
For $\alpha > 0$ to be determined, define
\begin{align*}
    V(s) = 1 - e^{-\alpha |s - s^\star|}.
\end{align*}
This function is bounded between 0 and 1, and increases monotonically as $|s - s^\star|$ increases.

We now verify that for sufficiently small $\alpha$ and large $R$, we have strictly positive drift outside the strip $\{|s - s^\star| > R\}$.
Consider the case $s > s^\star + R$ (the case $s < s^\star - R$ is symmetric).
Let $Z_t := S_t - S_{t-1}$ denote the one-step increment.
By Taylor expansion of $V(s + Z_t)$ around $s$ with Lagrange remainder:
\[
V(s + Z_t) - V(s) = V'(s) Z_t + \frac{1}{2} V''(s) Z_t^2 + \frac{1}{6} V'''(\tilde{s}) Z_t^3,
\]
for some $\tilde{s}$ between $s$ and $s + Z_t$.
For $s > s^\star$, we have $V'(s) = \alpha e^{-\alpha(s - s^\star)}$, $V''(s) = -\alpha^2 e^{-\alpha(s - s^\star)}$, and $|V'''| \leq \alpha^3$ everywhere (since $V'''(s) = \alpha^3 e^{-\alpha(s-s^\star)}$ for $s > s^\star$). Under Assumption~\ref{as:two-arm}, the log-likelihood ratio increment $Z_t$ is a Gaussian mixture with uniformly bounded moments of all orders; in particular, $\E[|Z_t|^3 \mid S_{t-1} = s] \leq C_3$ for a constant $C_3$ depending only on the model parameters. Thus the remainder satisfies $\E[|V'''(\tilde{s}) Z_t^3|/6 \mid S_{t-1} = s] \leq \alpha^3 C_3/6$.
Taking conditional expectations:
\begin{align*}
\E[V(S_t) - V(S_{t-1}) \mid S_{t-1} = s]
&= V'(s) \, g(s) + \frac{1}{2} V''(s) \, \E[Z_t^2 \mid S_{t-1} = s] + O(\alpha^3) \\
&= \alpha e^{-\alpha(s - s^\star)} \left( g(s) - \frac{\alpha}{2} \E[Z_t^2 \mid S_{t-1} = s] \right) + O(\alpha^3).
\end{align*}
Since $g(s) > 0$ for $s > s^\star$ (the drift points away from $s^\star$) and $\E[Z_t^2 \mid S_{t-1} = s]$ is bounded by some constant $M < \infty$ (by Assumption~\ref{as:two-arm}), we can choose $\alpha$ small enough such that
\[
g(s) - \frac{\alpha M}{2} > 0.
\]
Moreover, for $s > s^\star + R$ with $R$ sufficiently large, $g(s) \ge c > 0$ for some constant $c$ (since $g(s) \to \Delta_1 > 0$ as $s \to \infty$ in the self-confirming case, but here $g(s) > 0$ for all $s > s^\star$).
Thus, there exists $\varepsilon > 0$ such that
\[
\E\!\bigl[V(S_{t})-V(S_{t-1})\mid S_{t-1}=s\bigr]\;\ge\;\varepsilon
\quad\text{for }|s-s^\star|>R.
\]

The chain $(S_t)$ is $\psi$-irreducible since the transition kernel has Gaussian mixture densities with full support. Combined with the positive drift of the bounded function $V$ outside a compact set, the bounded-drift criterion
\citep[Thm.~8.4.2]{meynMarkovChainsStochastic1993} implies that $(S_t)$ is transient, which means it eventually leaves the strip $\{|s-s^\star|\le R\}$ and never returns.
Consequently,
\[
S_t\;\longrightarrow\;+\infty\quad\text{or}\quad-\infty
\qquad\text{a.s.}
\]

Because $\sigma(s)\to1$ as $s\to+\infty$ and $\sigma(s)\to0$ as $s\to-\infty$,
\[
\pi_t=\sigma(S_t)\;\xrightarrow{\text{a.s.}}\;
\mathbf 1_{\{S_t\to+\infty\}}.
\]

Then\/ $\pi_t\stackrel{d}{\longrightarrow}\mathrm{Bernoulli}\!\bigl(p\bigr)$, now we want to show $p = u(S_0)$ where
$u:\Bbb R\to[0,1]$ is the unique bounded solution of the harmonic equation
\begin{equation}\label{eq:harmonic}
u(s)=\E\bigl[u\!\bigl(s+Z_1\bigr)\mid S_0=s\bigr],\qquad
\lim_{s\to-\infty}u(s)=0,\;\;
\lim_{s\to+\infty}u(s)=1.
\end{equation}
In particular,
\[
\Pr\{\pi_t\to1\mid S_0=s\}=u(s),\qquad
\Pr\{\pi_t\to0\mid S_0=s\}=1-u(s).
\]

Set $u(s):=\Pr_s\{S_t\to+\infty\}$.
For each $L>0$ define the finite
target $A_L=[L,\infty)$ and $u_L(s):=\Pr_s\{\tau_{A_L}<\infty\}$.
The first-step equation for hitting probabilities
\citep{levin2017markov} gives
\[u_L(s)=\E\bigl[u_L\!\bigl(s+Z_1\bigr)\mid S_0=s\bigr],\]
and $u_L\uparrow u$ as $L\to\infty$.
Monotone convergence passes the averaging property to the limit,
yielding~\eqref{eq:harmonic}.
Uniqueness follows because any other bounded solution of~\eqref{eq:harmonic}
would contradict the minimality of the hitting-probability function.
Thus $u(S_0)$ is the desired Bernoulli success probability, completing the proof.
\end{proof}

\subsection{Proof of Theorem \ref{thm:uniform_dominance}}
\label{app:D:uniform_dominance_proof}

Let $A_t := \sum_{s=1}^{t} \mathbb{E}[Z_s \mid \mathcal{F}_{s-1}]$ represents the cumulative drift, and $M_t := \sum_{s=1}^t (Z_s - \mathbb{E}[Z_s \mid \mathcal{F}_{s-1}])$ is a zero-mean martingale.

Under the assumption~\ref{as:two-arm} and since the parameter space is finite, the log-likelihood increments have uniformly bounded second moments, i.e., $\sup_s \mathbb{E}[(\Delta M_s)^2 \mid \mathcal{F}_{s-1}] \leq \sigma^2 < \infty$. By the Strong Law of Large Numbers for Martingales \citep{chow1967strong}, the condition $\sum_{s=1}^\infty \frac{\text{Var}(\Delta M_s)}{s^2} \leq \sum_{s=1}^\infty \frac{\sigma^2}{s^2} < \infty$ is satisfied. Consequently:
\[
\lim_{t \to \infty} \frac{M_t}{t} = 0 \quad \text{a.s.}
\]

The expected log-increment is defined by the mixture $\mathbb{E}[Z_s \mid \mathcal{F}_{s-1}] = \sigma(s)\Delta_1 + (1-\sigma(s))\Delta_2$. Given the uniform dominance condition $\min(\Delta_1, \Delta_2) > c > 0$, the average drift is bounded away from zero:
\[
\liminf_{t \to \infty} \frac{A_t}{t} = \liminf_{t \to \infty} \frac{1}{t} \sum_{s=1}^t \mathbb{E}[Z_s \mid \mathcal{F}_{s-1}] \geq c > 0
\]

Combining these limits, we obtain the asymptotic growth rate of the log-evidence ratio:
\[
\liminf_{t \to \infty} \frac{S_{t+1}}{t} = \liminf_{t \to \infty} \left( \frac{S_0}{t} + \frac{A_t}{t} + \frac{M_t}{t} \right) \geq 0 + c + 0 = c > 0 \quad \text{a.s.}
\]
This implies $\lim_{t \to \infty} S_t = +\infty$ almost surely. By the Continuous Mapping Theorem, since the posterior probability $\pi_t = (1 + e^{-S_t})^{-1}$ is a continuous function of $S_t$:
\[
\lim_{t \to \infty} \pi_t = \frac{1}{1 + \lim_{S_t \to \infty} e^{-S_t}} = 1 \quad \text{a.s.}
\]
Thus, the posterior probability of model $\nu$ converges to $1$ almost surely.

\subsection{Proof of Theorem \ref{thm:self_defeating}}
\label{app:D:self_defeating_proof}

\begin{proof}
    We prove this theorem in two steps. First, we establish the Harris ergodicity of the log-odds process via a Lyapunov function approach. Second, we derive the existence and characterization of the invariant distribution for the posterior probabilities.
    
    According to equation \eqref{eq:log-odds-two-arm}, we have the stochastic process:
    \[
    S_{t} = S_{t-1} + Z_t, \quad Z_t \sim
    \begin{cases}
    X_t & \text{with probability } \pi_{t-1} = \sigma(S_{t-1}), \\
    Y_t & \text{with probability } 1 - \pi_{t-1},
    \end{cases} \quad \sigma(s) = \frac{1}{1 + e^{-s}},
    \]
    where $\mathbb{E}[X_t] = \Delta_1 < 0$ and $\mathbb{E}[Y_t] = \Delta_2 > 0$. Under assumption~\ref{as:two-arm}, both $X_t$ and $Y_t$ have bounded second moments $\sigma_X^2$ and $\sigma_Y^2$, respectively.
    
    Define the Lyapunov function:
    \[
    V(s) = (s - s^\star)^2,
    \]
    where $s^\star$ is the unique solution to $g(s^\star)=0$ with:
    \[
    g(s) = \mathbb{E}[S_t - S_{t-1} \mid S_{t-1}=s] = \sigma(s) \Delta_1 + (1 - \sigma(s)) \Delta_2.
    \]
    Observe that $g(s)$ is continuous, strictly decreasing, and satisfies $g(-\infty)=\Delta_2>0$, $g(\infty)=\Delta_1<0$.
    
    Compute the drift:
    \[
    \begin{aligned}
    \Delta V(s) &= \mathbb{E}[V(S_{t+1}) - V(S_t) \mid S_t=s] \\
    &= \mathbb{E}[(S_t + Z_t - s^\star)^2 - (s - s^\star)^2 \mid S_t=s] \\
    &= 2(s - s^\star)g(s) + \mathbb{E}(Z^2_t \mid S_t=s).
    \end{aligned}
    \]
    
    The term $2(s-s^\star)g(s)$ is negative outside $s=s^\star$ because $g(s)(s-s^\star)<0$. The second term is bounded by the assumption~\ref{as:two-arm} and the structure of the mixture distribution.
    
    Thus, there exist constants $c>0$ and $B<\infty$ such that:
    \[
    \mathbb{E}[V(S_{t+1}) \mid S_t=s] \le V(s) - c V(s)^{\frac{1}{2}} + B, \quad \forall s \in \mathbb{R}.
    \]
    To bring the drift inequality into the $f$-ergodicity condition in \citep[Ch.~14]{meynMarkovChainsStochastic1993}, define the shifted Lyapunov function
    \[
    \widetilde{V}(s) := 1 + V(s) = 1 + (s - s^\star)^2,
    \]
    which satisfies \( \widetilde{V}(s) \in [1, \infty) \) for all \( s \in \mathbb{R} \). Then the drift becomes:
    \[
    \begin{aligned}
    \mathbb{E}[\widetilde{V}(S_{t+1}) - \widetilde{V}(S_t) \mid S_t = s]
    &= \mathbb{E}[V(S_{t+1}) \mid S_t = s] - V(s) \\
    &\leq - c V(s)^{\frac{1}{2}} + B.
    \end{aligned}
    \]
    Let \( 0 < \beta < c \), then there exists a constant \( R_1 > 0 \) such that for all \( s \) satisfying \( |s - s^\star| > R_1 \), we have 
    \[
    \mathbb{E}[\widetilde{V}(S_{t+1}) - \widetilde{V}(S_t) \mid S_t = s] \leq -\beta \widetilde{V}(s)^{\frac{1}{2}} + B.
    \]
    
    Moreover, because \( |s - s^\star| \to \infty \) as \( |s| \to \infty \), the negative drift term dominates outside a compact interval. That is, there exists a constant \( R_2 > 0 \) such that for the compact set \( C := [s^\star - R_2, s^\star + R_2] \), we have
    \[
    \mathbb{E}[\widetilde{V}(S_{t+1}) - \widetilde{V}(S_t) \mid S_t = s] \leq -\beta \widetilde{V}(s)^{\frac{1}{2}} + \cdot \mathbf{1}_C(s) B, \quad \forall s \in \mathbb{R}.
    \]
    This in fact satisfies the polynomial drift condition, which was first given by \citet{jarner2002polynomial}, completing the proof.
    
The Markov chain \(\{ S_t \}\) is irreducible and aperiodic, as it is driven by the full support of the underlying signals \(X_t\) and \(Y_t\). Invoking the global drift condition established above, we conclude that the log-odds process \(\{ S_t \}\) admits a unique invariant probability measure \(\pi_S\). Moreover, the chain is Harris recurrent and ergodic, with the law \(\mathcal{L}(S_t)\) converging in total variation to \(\pi_S\) as \(t \to \infty\). Since the posterior \(\pi_t\) is a continuous transformation of \(S_t\) via the map \(\sigma\), it inherits the ergodicity of the log-odds process. As a result, the posterior distribution converges in distribution to the pushforward measure \(\pi_{\text{post}} = \pi_S \circ \sigma^{-1}\).

\end{proof}

\subsection{Proof of Theorem \ref{thm:ergodicity-angle-condition}}
\label{app:D:ergodicity-angle-condition-proof}

We now establish a sufficient condition for ergodicity of the stochastic process using the Foster–Lyapunov criterion. The following assumption is required for ergodicity:

This assumption is mild. When the reward density has bounded second moments, this assumption is usually satisfied. Now we are ready to prove the first sufficient condition for ergodicity in Theorem \ref{thm:ergodicity-angle-condition}.

\begin{proof}
    Define
    \[
    V(S) \;=\; 1 + \|S - S^\star\|^2.
    \]
    Then
    \[
    V(S_{t+1}) - V(S) = \|S_{t+1}-S^\star\|^2 - \|S-S^\star\|^2.
    \]
    Expanding,
    \[
    = 2\langle S - S^\star,\, S_{t+1}-S\rangle + \|S_{t+1}-S\|^2.
    \]
    
    Taking conditional expectation,
    \[
    \mathbb{E}[V(S_{t+1})-V(S)\mid S_t=S]
    = 2\langle S - S^\star,\, \mathbb{E}[S_{t+1}-S \mid S_t=S]\rangle 
    + \mathbb{E}[\|S_{t+1}-S\|^2 \mid S_t=S].
    \]
    
    By definition of the drift, 
    \[
    \mathbb{E}[S_{t+1}-S \mid S_t=S] = D \psi(S).
    \]
    Thus,
    \[
    \mathbb{E}[V(S_{t+1})-V(S)\mid S_t=S]
    = 2 \langle S - S^\star, D \psi(S)\rangle + \sigma^2(S),
    \]
    where $\sigma^2(S):=\mathbb{E}[\|S_{t+1}-S\|^2\mid S_t=S]$ is uniformly bounded under Assumption \ref{assumption:uniform-second-moment-bound}.  
    
    Condition \eqref{eq:angle-condition} ensures that the first term is strictly negative outside a compact set: there exists $\beta>0$ such that
    \[
    2 \langle S - S^\star, D \psi(S)\rangle \le -\beta \|S-S^\star\|, \qquad  \text{for large } \|S-S^\star\|.
    \]
    Hence
    \[
    \mathbb{E}[V(S_{t+1})-V(S)\mid S_t=S]
    \;\le\; -\,\beta V(S)^{\frac{1}{2}} + B,
    \]
    for some constant $B>0$. This is precisely the Foster–Lyapunov drift condition, which implies positive Harris recurrence and thus ergodicity of $\{S_t\}$. 
    \end{proof}

\subsection{Proof of Theorem \ref{thm:sufficient-condition-ergodicity-transience}}
\label{app:D:sufficient-condition-ergodicity-transience-proof}

First of all, observe that the drift vector can be factorized in a way that removes the reference model term:

\begin{equation}
    \xi(S)\;=\; G\,\big(\pi_{-M}(S)-\pi_{-M}^*\big).
    \label{eq:xi-factor}
\end{equation}
where $G := [d(a_1), \ldots, d(a_{M-1})] \in \mathbb{R}^{(M-1) \times (M-1)}$.
    
Then we can define a canonical potential function $V(S)$ and show that the mean field ODE is a generalized gradient system. Let $A(S)=\log\!\big(1+\sum_{k=1}^{M-1}e^{S^{(k)}}\big)$ and $V(S)=A(S)-\langle S,\pi_{-M}^*\rangle$. Then $\nabla V(S)=\pi_{-M}(S)-\pi_{-M}^*$ and $\nabla^2 V(S)=\mathrm{diag}(\pi_{-M}(S))-\pi_{-M}(S)\pi_{-M}(S)^\top\succeq 0$.

The mean field ODE is a generalized gradient system as follows:
\begin{equation}
    \dot S = \xi(S) = G \nabla V(S)
\end{equation}

The potential function $V(S)$ is convex in general. And the strict convexity is guaranteed by the uniqueness of the interior fixed point.

\begin{proposition}[Uniqueness of Interior Fixed Point]
    \label{prop:uniqueness_simple}
    The interior fixed point is unique if and only if $G$ has full rank.
    \end{proposition}

\begin{remark}
    The Jacobian of the mean field ODE is given by:
    \begin{equation}
        J_{\xi}(S) = G \nabla^2 V(S)
    \end{equation}

    Then the Jacobian is local stable (Hurwitz) if the symmetric part of $G$, $\Sym(G) = (G + G^\top)/2$, is negative definite.
\end{remark}

Now we can state and prove the main Theorem~\ref{thm:sufficient-condition-ergodicity-transience}

\begin{proof}
    \textbf{1. Condition for Ergodicity} We use the same potential $V$ as in Definition \ref{def:softmax-potential-main}:
    
    \[
    W(S):=1+V(S)-V\left(S^{\star}\right), \quad \nabla W(S)=\nabla V(S) .
    \]
    
    A one-step Taylor expansion around $S_t$ gives, for $Z_t=S_{t+1}-S_t$,
    \[
    V(S_t+Z_t)-V(S_t)=\nabla V(S_t)^{\top} Z_t+\frac{1}{2} Z_t^{\top} H(S_t+\theta Z_t) Z_t,
    \]
    
    for some $\theta \in(0,1)$, with $H=\nabla^2 V = \operatorname{diag}\left(\pi_{-M}\right)-\pi_{-M} \pi_{-M}^{\top}$. Taking conditional expectation and using $\mathbb{E}[Z_t \mid S]=G \nabla V(S)$,
    \[
    \mathbb{E}\left[W\left(S_{t+1}\right)-W\left(S_t\right) \mid S_t=S\right] \leq \nabla V(S)^{\top} G \nabla V(S) + \frac{1}{2}\|H\| \bar{\sigma}^2 \le  -\mu \|\nabla V(S)\|^2 +C,
    \]
    
    for constants $\mu>0$ (from $\operatorname{Sym}(G) \prec 0$ ) and $C=\frac{1}{4} \bar{\sigma}^2$. $H(S)$ is positive definite and uniformly bounded in spectral norm with sharp bound $\|H\|\le \frac{1}{2}$.
    
    Because $\pi^{\star} \in \Delta_M^{\circ}$, there exists a radius $R$ such that when $V(S) \geq R$ the vector $u(S)=\nabla V(S)= \pi_{-M}(S)-\pi_{-M}^{\star}$ is bounded away from 0. Hence there is $c_0>0$ and a compact set $C_R:=\{V \leq R\}$ with
    
    \[
    \|\nabla V(S)\| \geq c_0 \quad \text { on } C_R^{\mathrm{c}} .
    \]
    
    Therefore, for $S \notin C_R$,
    
    \[
    \mathbb{E}\left[W\left(S_{t+1}\right)-W\left(S_t\right) \mid S_t=S\right] \leq-\mu c_0^2+C=:-\varepsilon<0,
    \]
    
    this is achieved by Assumption \ref{as:small-noise} and after increasing $R$ if needed so that $\varepsilon>0$. Thus we have the standard Foster-Lyapunov drift condition. With Lemma \ref{lemma:irreducibility-and-aperiodicity-of-the-log-odds-process} it yields positive Harris recurrence and an invariant probability measure.
    
\textbf{2. Condition for Transience}
To prove transience, we must show that the chain has a non-zero probability of never returning to a compact set. We cannot use the previous Lyapunov function $W$ as it is unbounded. Here we consider the following function:
\[
U(S) := 1- \exp\big(-\eta W(S)\big),
\]
for a small parameter $\eta > 0$ to be chosen later. Note that $U(S)$ is bounded ($0 < U \le 1$) and $U(S) \to 1$ as $\|S\| \to \infty$.

We analyze the conditional expectation of $\Delta U$. Since $U = 1 - e^{-\eta W}$, we have $\Delta U = e^{-\eta W(S_t)}(1 - e^{-\eta \Delta W})$. Using the inequality $e^{-x} \le 1 - x + \frac{x^2}{2}$, which holds for all $x \in \mathbb{R}$ (since $e^{-x}$ is convex and the right-hand side is its second-order Taylor polynomial at $0$), we obtain:
\begin{equation}
    \mathbb{E}[U_{t+1} - U_t \mid S_t=S] \le U(S) \left( -\eta \mathbb{E}[\Delta W \mid S] + \frac{\eta^2}{2} \mathbb{E}[(\Delta W)^2 \mid S] \right).
\end{equation}

In order to get the Lyapunov drift condition, we have to lower bound the linear term $\mathbb{E}[\Delta W \mid S]$ and upper bound the quadratic term $\mathbb{E}[(\Delta W)^2 \mid S]$. Specifically, we want to show the following upper and lower bounds hold for some constants $C_{\text{noise}}$ and $K$:
\begin{align}
\mathbb{E}[\Delta W \mid S] &\ge \mu \|\nabla V(S)\|^2 - C_{\text{noise}}, \\
\mathbb{E}[(\Delta W)^2 \mid S] &\le K.
\end{align}

\textit{1. Lower Bound on $\mathbb{E}[\Delta W \mid S]$.} We first establish a lower bound on the drift $\mathbb{E}[\Delta W \mid S]$. Recall that $W(S)$ is the potential function and $Z_t = S_{t+1} - S_t$. We perform a second-order Taylor expansion of $W$ around $S_t$: 
\[
\Delta W = W(S_{t+1}) - W(S_t) = \nabla V(S_t)^\top Z_t + \frac{1}{2} Z_t^\top H(\tilde{S}_t) Z_t,
\]
where $\tilde{S}_t$ is some point on the line segment between $S_t$ and $S_{t+1}$, and $H = \nabla^2 V$ is the Hessian. Taking the conditional expectation with respect to $S_t = S$:
\[
\mathbb{E}_S[\Delta W] = \nabla V(S)^\top \mathbb{E}_S[Z_t] + \frac{1}{2} \mathbb{E}_S\left[ Z_t^\top H(\tilde{S}_t) Z_t \right].
\]
Using the property that $\mathbb{E}_S[Z_t] = G \nabla V(S)$ and the spectral condition $\Sym(G) \succ 0$:
\[
\begin{aligned}
\nabla V(S)^\top \mathbb{E}_S[Z_t] &= \nabla V(S)^\top G \nabla V(S) \\
&= \nabla V(S)^\top \left( \frac{G + G^\top}{2} \right) \nabla V(S) \\
&\ge \lambda_{\min}(\Sym(G)) \|\nabla V(S)\|^2 \\
&= \mu \|\nabla V(S)\|^2.
\end{aligned}
\]
We strictly lower bound the remainder term by subtracting its maximum possible magnitude. Using the Cauchy-Schwarz inequality for matrices ($x^\top A x \le \|A\| \|x\|^2$) and the global bound $\|H(S)\| \le L$:
\[
\begin{aligned}
\frac{1}{2} \mathbb{E}_S\left[ Z_t^\top H(\tilde{S}_t) Z_t \right] &\ge - \left| \frac{1}{2} \mathbb{E}_S\left[ Z_t^\top H(\tilde{S}_t) Z_t \right] \right| \\
&\ge - \frac{1}{2} \mathbb{E}_S\left[ \|H(\tilde{S}_t)\| \|Z_t\|^2 \right] \\
&\ge - \frac{1}{2} L \, \mathbb{E}_S\left[ \|Z_t\|^2 \right].
\end{aligned}\]
By Assumption \ref{as:small-noise} (and the bounded second moment assumption), $\mathbb{E}_S[\|Z_t\|^2] \le \bar{\sigma}^2$. Thus:
\[
\frac{1}{2} \mathbb{E}_S\left[ Z_t^\top H(\tilde{S}_t) Z_t \right] \ge - \frac{1}{2} L \bar{\sigma}^2.
\]
Adding the bounds from Step A and Step B gives the final inequality:
\[
\begin{aligned}
\mathbb{E}[\Delta W \mid S] &\ge \mu \|\nabla V(S)\|^2 - \frac{1}{2} L \bar{\sigma}^2 \\
&= \mu \|\nabla V(S)\|^2 - C_{\text{noise}}.
\end{aligned}
\]

\textit{2. Upper Bound on $\mathbb{E}[(\Delta W)^2 \mid S]$.}
Then we need to control the second moment $\mathbb{E}[(\Delta W)^2 \mid S]$ to ensure that the quadratic term does not dominate the linear term in the drift. We utilize the global Lipschitz property of the softmax potential.
    
Recall that $\nabla V(S) = \pi_{-M}(S) - \pi^\star_{-M}$. Since probability vectors lie in the simplex, the gradient is uniformly bounded:
\[
\|\nabla V(S)\| \le \|\pi_{-M}(S)\| + \|\pi^\star_{-M}\| \le 2\sqrt{M}.
\]
Consequently, $V$ is globally Lipschitz continuous with constant $L_V \approx 2\sqrt{M}$. By the Lipschitz property:
\[
|\Delta W| = |V(S_{t+1}) - V(S_t)| \le L_V \|S_{t+1} - S_t\| = L_V \|Z_t\|.
\]
Squaring both sides and taking the conditional expectation:
\[
\mathbb{E}\left[(\Delta W)^2 \mid S\right] \le L_V^2 \, \mathbb{E}\left[\|Z_t\|^2 \mid S\right] \le 2 L_V^2 \bar{\sigma}^2.
\]
This confirms that there exists a constant $K = 2 L_V^2 \bar{\sigma}^2$ such that:
\[
\mathbb{E}[(\Delta W)^2 \mid S] \le K.
\]
Substituting these into the expansion for $U$:
\[
\frac{\mathbb{E}[\Delta U \mid S]}{U(S)} \le -\eta \left(\mu \|\nabla V\|^2 - C_{\text{noise}}\right) + \frac{\eta^2}{2} K.
\]
Then:
\[
\mathbb{E}[\Delta U \mid S] \le U(S) \left( -\eta \left(\mu \|\nabla V\|^2 - C_{\text{noise}}\right) + O(\eta^2) \right).
\]
By Lemma \ref{lem:grad-away} and assumption~\ref{as:small-noise}, outside a sufficiently large compact set $C_R$, $\|\nabla V(S)\| \ge c_\infty$. Thus, for small enough $\eta$, the negative term dominates the constant noise term, yielding:
\[
\mathbb{E}[U(S_{t+1}) - U(S_t) \mid S_t = S] < 0 \quad \text{for all } S \notin C_R.
\]

The function $U$ is bounded and measurable, and $\mathbb{E}[\Delta U \mid S] \geq \varepsilon > 0$ outside the compact set $C_R$ (equivalently, $\mathbb{E}[\Delta(1-U) \mid S] \leq -\varepsilon < 0$). These are precisely the conditions of \citet[Thm.~8.4.2]{meynMarkovChainsStochastic1993}, which, combined with the $\psi$-irreducibility established in Lemma~\ref{lemma:irreducibility-and-aperiodicity-of-the-log-odds-process}, implies that the chain is transient.

Since the chain is transient on $\mathbb{R}^d$, it visits any compact set only finitely often. Consequently, for any radius $M > 0$, there exists a finite random time $T_M$ such that $\|S_t\| > M$ for all $t > T_M$. This implies: 
\[
\lim_{t \to \infty} \|S_t\| = \infty \quad \text{almost surely.}
\]
As $\|S\|\to\infty$, softmax posteriors approach the boundary of the simplex (at least one coordinate of $\pi$ vanishes), so $\mathrm{dist}(\pi_t,\partial\Delta_M)\to 0$ a.s.

\end{proof}

\section{Supporting Lemmas and Auxiliary Results}
\label{app:E}

This appendix collects supporting lemmas and auxiliary results used throughout the paper.

This assumption ensures that the agent can, over time, steer the log-odds vector in any direction in $\mathbb{R}^{M-1}$ by encountering a suitable sequence of observations. We now state the formal irreducibility result.

\subsection{Proof of Lemma~\ref{lemma:irreducibility-and-aperiodicity-of-the-log-odds-process}}
\label{app:A:markov}

\paragraph{Markov Property} Under Thompson Sampling, the posterior process $\{\pi_t\}$ is a Markov process on the belief simplex $\Delta_M$. The action selection depends only on the current posterior, which ensures the Markov property.

By construction we have $S_{t+1}=S_t+\mathbf{Z}_t$, where $\mathbf{Z}_t$ is drawn from $\nu_{S_t}$ depending only on $S_t$. Thus for any Borel set $B$,
\[
    \mathbb{P}(S_{t+1}\in B \mid \mathcal{F}_t) \;=\; \mathbb{P}(S_{t+1}\in B \mid S_t),
\]
so $(S_t)$ is a time-homogeneous Markov chain with transition kernel
\begin{equation}
    P(s,B) \;=\; \mathbb{P}(S_{t+1}\in B \mid S_t = s) \;=\; \int \mathbbm{1}\{s+z \in B\}\,\nu_s(dz).\label{eq:transition-kernel}
\end{equation}

\paragraph{Irreducibility}

To establish  irreducibility of the Markov chain $(S_t)_{t\geq 0}$, we examine the mapping from the data space to the sum of increments. Let $d = M-1$. By the Rank Condition in Assumption~\ref{assum:full-rank-Reversibility}, there exist $d$ smooth curves $\gamma_1, \dots, \gamma_d$ in the support $\mathcal{Z}$ and points $t_1^*, \dots, t_d^*$ such that the tangent vectors $\gamma_i'(t_i^*)$ are linearly independent.

Consider the summation map $\Phi: \mathbb{R}^d \to \mathbb{R}^d$ defined by:
\[
    \Phi(t_1, \dots, t_d) = \sum_{i=1}^d \gamma_i(t_i).
\]
The Jacobian matrix of this map at the point $\mathbf{t}^* = (t_1^*, \dots, t_d^*)$ is given by the column vectors of the tangents:
\[
    J_\Phi(\mathbf{t}^*) = \begin{bmatrix} | & & | \\ \gamma_1'(t_1^*) & \cdots & \gamma_d'(t_d^*) \\ | & & | \end{bmatrix}.
\]

By the Rank Condition, this Jacobian has full rank $d$. Since $\Phi(\mathbf{t}^*) = \mathbf{0}$ and the differential is surjective, the Submersion Theorem implies that the image of $\Phi$ contains an open neighborhood of the value $\mathbf{0}$. That is, there exists $\epsilon > 0$ such that the open ball $B_\epsilon(\mathbf{0})$ is contained in the support of the $L$-step distribution. Thus, $B_\epsilon(\mathbf{0}) \subset S_{\mathcal{Z}}$. 

We now show that $S_{\mathcal{Z}} = \mathbb{R}^d$. Let $x \in \mathbb{R}^d$ be any target vector. Since $B_\epsilon(\mathbf{0})$ is a neighborhood of the origin, for sufficiently large integer $N$, the scaled vector $x/N$ satisfies $\| x/N \| < \epsilon$.Thus, $x/N \in B_\epsilon(\mathbf{0}) \subset S_{\mathcal{Z}}$. Since $S_{\mathcal{Z}}$ is closed under addition, the sum of $N$ copies of this vector is also in $S_{\mathcal{Z}}$:\[ x = \underbrace{\frac{x}{N} + \dots + \frac{x}{N}}_{N \text{ times}} \in S_{\mathcal{Z}}. \] This proves that every point in $\mathbb{R}^d$ can be reached by a finite sum of increments. Thus, $S_{\mathcal{Z}} = \mathbb{R}^d$

We now show that the process can stay close to this trajectory with positive probability. Define the summation map $\Psi: \mathbb{R}^{M-1} \times (\mathbb{R}^{M-1})^k \to \mathbb{R}^{M-1}$ by $\Psi(s, w_1, \dots, w_k) = s + \sum_{j=1}^k w_j$. This map is continuous. We have $\Psi(u_0, z_1, \dots, z_k) = v_0 \in V$. Since $V$ is open, the preimage $\Psi^{-1}(V)$ is an open set in the product space containing the point $(u_0, z_1, \dots, z_k)$. Therefore, there exist open neighborhoods $O_{u} \ni u_0$ and $O_{z_j} \ni z_j$ for each $j=1,\dots,k$ such that:
\[
    O_{u} \oplus O_{z_1} \oplus \dots \oplus O_{z_k} \subset V.
\]
Without loss of generality, we can choose $O_u \subset B_\epsilon(u_0) \subset U$.

We compute the probability of the event that the chain follows this trajectories.
\begin{align*}
    P^{n_0+k}(s, B) &\ge \mathbb{P}\left( S_{n_0+k} \in V \right) \\
    &\ge \mathbb{P}\left( S_{n_0} \in O_u, \ Z_{n_0+1} \in O_{z_1}, \dots, Z_{n_0+k} \in O_{z_k} \right).
\end{align*}
By the Markov property, the independence of increments conditional on state and the fact that support is state-independent:
\[
    = \mathbb{P}(S_{n_0} \in O_u) \times \prod_{j=1}^k \mathbb{P}(Z \in O_{z_j}).
\]
$\mathbb{P}(S_{n_0} \in O_u) > 0$ because $O_u \subset U$ and $S_{n_0}$ has a density on $U$. For the increment probabilities, each $z_j$ lies on one of the support curves $\gamma_i$, say $z_j = \gamma_i(t_j)$. Each such curve is a smooth map from $\mathbb{R}$ to $\mathbb{R}^{M-1}$, and the reward distribution under each action has a density with full support on $\mathbb{R}$ (e.g., Gaussian), so the pushforward of the reward density through $\gamma_i$ assigns positive measure to any open neighborhood of $z_j$ on the curve. In particular, $\mathbb{P}(Z \in O_{z_j}) > 0$ for each $j$. Thus, $P^{n_0+k}(s, B) > 0$.

\paragraph{Aperiodicity} We prove that the $\psi$-irreducible Markov chain $(S_t)$ is aperiodic. It suffices to exhibit a measurable set $C$ with $\psi(C)>0$ and an integer $n$
such that for all $s\in C$,
\[
P^{n}(s,C)>0
\quad\text{and}\quad
P^{n+1}(s,C)>0,
\]
since then the period divides $\gcd(n,n+1)=1$ and must equal $1$.

\textit{1. $L$-step neighborhood around $0$.} By Assumption~\ref{assum:full-rank-Reversibility}, there exist increments
$z_1,\dots,z_L$ satisfying the Rank and Reversibility conditions.
As shown in the irreducibility proof, this implies the existence of $\varepsilon>0$
such that for any $s$ in a neighborhood of $0$,
\[
B_\varepsilon(0)
\subset
\mathrm{supp}\big(\nu_s^{(*L)}\big),
\]
where $\nu_s^{(*L)}$ denotes the $L$-step increment law.

\textit{2. Return in $kL$ steps.} Let $C := B_r(0)$ with $r<\varepsilon/4$.
Choose $k$ large enough so that $k\varepsilon > 2r$.
Since convolution enlarges the reachable set,
\[
B_{k\varepsilon}(0)
\subset
\mathrm{supp}\big(\nu_s^{(*(kL))}\big).
\]
Hence for any $s\in C$,
\[
C \subset s + B_{k\varepsilon}(0),
\]
and therefore
\[
P^{kL}(s,C) > 0
\qquad\forall s\in C.
\]

\textit{3. Return in $kL+1$ steps.}

Pick any one-step support point $v\in\mathcal Z$.
Increase $k$ if necessary so that $\|v\|<k\varepsilon/2$.
Then $-v\in B_{k\varepsilon}(0)$, so for any $s\in C$ there is positive probability that
\[
S_{kL} \in s - v + U
\]
for some small neighborhood $U$ of $0$.

Choose $C$ sufficiently small so that
\[
(s - v + U) + v \subset C
\qquad\text{for all } s\in C.
\]
Since $v$ lies in the one-step support and increments are continuous along the
support curves, there exists a neighborhood $U_v$ of $v$ such that
$\nu_s(U_v)>0$ for all $s\in C$.
Therefore,
\[
P^{kL+1}(s,C)>0
\qquad\forall s\in C.
\]

We have shown that for all $s\in C$,
\[
P^{kL}(s,C)>0
\quad\text{and}\quad
P^{kL+1}(s,C)>0.
\]
Since $\gcd(kL,kL+1)=1$, the period of the $\psi$-irreducible chain
must equal $1$. Thus $(S_t)$ is aperiodic.

\subsection{Proof of Lemma~\ref{lemma:dichotomy-of-the-log-odds-process}}

\begin{proof}
    Given our markov chain is Lebesgue-irreducible (Lemma \ref{lemma:irreducibility-and-aperiodicity-of-the-log-odds-process}), the result is a direct consequence of the theorem 8.0.1 in \citet{meynMarkovChainsStochastic1993}.
\end{proof}

\subsection{Proof of Lemma \ref{lem:grad-away}}

\begin{lemma}[Gradient bounded away from $0$ at infinity]\label{lem:grad-away}
    Let $\pi^\star\in\Delta_M^\circ$ and $V$ be as in Def.~\ref{def:softmax-potential-main}. There exists $c_\infty>0$ and $R<\infty$ such that
    \[
    \|\nabla V(S)\|=\|\pi_{-M}(S)-\pi^\star_{-M}\|\ \ge\ c_\infty\qquad\text{for all } S\notin C_R:=\{V\le R\}.
    \]
    In particular, one may take $c_\infty\le \mathrm{dist}\big(\pi^\star,\partial\Delta_M\big)$; a crude explicit bound is $c_\infty\le \min_{1\le j\le M}\pi^\star_j$.
    \end{lemma}
    
    \begin{proof}
    As $\|S\|\to\infty$, every sequence has a subsequence along which $\pi(S)$ converges to a point on some proper face of $\Delta_M$ (softmax saturation). Since $\pi^\star\in\Delta_M^\circ$, its Euclidean distance to any proper face is strictly positive. Hence the preimage of any sufficiently small ball around $\pi^\star$ is bounded, which yields the claim by continuity of $\nabla V$.
    \end{proof}

\subsection{Proof of Proposition \ref{prop:transience-no-equilibrium}}

\begin{proof}
    By strict separation of a point from a compact convex set, the hypothesis $0\notin\mathrm{conv}\{d(a_j)\}$ yields a vector $w\in\R^{M-1}$ and a constant $c>0$ with
    $w^\top d(a_j)\ge c$ for all $j$.
    Therefore, for every $S$,
    \[
    w^\top \xi(S)
    =\sum_{j=1}^M \pi_j(S)\,w^\top d(a_j)
    \ \ge\ \sum_{j=1}^M \pi_j(S)\,c
    \ =\ c.
    \]
    
    Define the one–step noise as $\varepsilon_{t+1}:=Z_{t+1}-\xi(S_t)$. By Assumption \ref{as:small-noise}, we have
    $\E[(w^\top \varepsilon_{t+1})^2\mid S_t]\le \|w\|^2\,\bar\sigma^2<\infty$.
    Let $M_t:=w^\top S_t$.
    From the decomposition
    \[
    M_{t+1}-M_t
    =\ w^\top Z_{t+1}
    =\ w^\top \xi(S_t)\ +\ w^\top \varepsilon_{t+1},
    \]
    we obtain
    \[
    M_t
    = M_0+\sum_{k=0}^{t-1} w^\top \xi(S_k)\;+\;\sum_{k=0}^{t-1} w^\top \varepsilon_{k+1}
    \ \ \ge\ M_0 + c\,t + \sum_{k=0}^{t-1} w^\top \varepsilon_{k+1}.
    \]
    The martingale difference array $\{w^\top\varepsilon_{t}\}$ is square-integrable with uniformly bounded conditional variances, hence by the strong law for martingale differences,
    \[
    \frac{1}{t}\sum_{k=0}^{t-1} w^\top \varepsilon_{k+1}\to 0 \quad a.s.
    \] 
    Dividing by $t$ gives
    \[
    \frac{1}{t}M_t\ \ge\ c + o(1)\ \to\ \mu \quad a.s.
    \]
    with $\mu\ge c>0$, so $M_t\to +\infty$.
    By Cauchy–Schwarz, $|M_t|\le \|w\|\,\|S_t\|$, hence $\|S_t\|\to\infty$.
    In particular, for any compact $K\subset\mathbb{R}^{M-1}$,
    $\sup_{x\in K} w^\top x<\infty$, so $S_t\in K$ can occur only finitely many times a.s.; thus the chain is transient.
    
    For the posterior, set $\pi_t=\psi(S_t)$.
    Since $\|S_t\|\to\infty$, consider any subsequence $t_n$ along which $u:=\lim_{n\to\infty} S_{t_n}/\|S_{t_n}\|$ exists.
    Write the softmax in coordinates $k=1,\dots,M-1$ with baseline $M$:
    \[
    \pi^{(k)}(S)=\frac{e^{S^{(k)}}}{1+\sum_{i=1}^{M-1} e^{S^{(i)}}},\quad
    \pi^{(M)}(S)=\frac{1}{1+\sum_{i=1}^{M-1} e^{S^{(i)}}}.
    \]
    If $\max_k u^{(k)}>0$, then $S^{(k)}_{t_n}$ for maximizers $k$ grow like $+\|S_{t_n}\|$, forcing $\pi^{(M)}(S_{t_n})\to 0$; if $\max_k u^{(k)}<0$, then all $S^{(k)}_{t_n}\to -\infty$ and $\pi^{(M)}(S_{t_n})\to 1$.
    In both cases, $\mathrm{dist}(\pi_{t_n},\partial\Delta_M)\to 0$.
    As every diverging sequence admits such a subsequence, we conclude $\mathrm{dist}(\pi_t,\partial\Delta_M)\to 0$ a.s., and any limit point lies on a face $I$.
    \end{proof}
    
\subsection{Proof of Proposition \ref{prop:dynamics_succession}}

\begin{proof}
By smoothness of the softmax function, \(\pi_{t+1}=\psi(S_t+Z_{t+1})=\psi(S_t)+J(\pi_t)Z_{t+1}+O(|Z_{t+1}|^2)\). Near face \(\Delta_I\), all Jacobian entries \(J(\pi_t)\) connecting to \(I^c\) are \(o(1)\), and the drift satisfies \(\xi(S_t)=\xi_I(\pi_{I,t})+o(1)\). Therefore, the full one-step update restricted to \(I\) matches the face scheme up to vanishing error, while the off-face increment vanishes.
\end{proof}

\subsection{Proof of Proposition \ref{prop:kernel_simplex}}

\begin{proof}
    Let $Z(A) \subset \mathbb{R}^m$ be the 1-dimensional kernel of $A$ (a line through the origin). Because the entry distribution is continuous, $\mathbb{P}(\exists i : z_i = 0) = 0$ for a direction vector $z \in Z(A) \setminus \{0\}$, so almost surely all coordinates of $z$ are nonzero.
    
    Observe that $\ker(A) \cap \Delta^m \neq \varnothing$ if and only if there exists a nonnegative vector in the line $Z(A)$ whose coordinates sum to $1$. Since scaling along $Z(A)$ is free, this occurs if and only if $z$ has all coordinates with the same sign. Then we can scale $z$ so that $\sum_i x_i = 1$. Conversely, if $z$ has mixed signs, no positive scaling makes all coordinates $\geq 0$. Thus the event of interest is the sign pattern of $z$ lies in $\{(+,\ldots,+), (-,\ldots,-)\}$.
    
    By symmetry: for any sign vector $s \in \{\pm 1\}^m$, let $D_s = \operatorname{diag}(s_1, \ldots, s_m)$. Because the entries of $A$ are i.i.d.\ and symmetric, $A D_s \stackrel{d}{=} A$. If $z$ spans $\ker(A)$, then $D_s z$ spans $\ker(A D_s)$. Hence the distribution of the sign pattern of the kernel direction is invariant under independent flips of the $m$ coordinate signs. 
    
    Since the boundary (some coordinate zero) has probability $0$, each of the $2^m$ orthants occurs with equal probability $2^{-m}$. Exactly two orthants correspond to ``all coordinates have the same sign.'' Therefore\footnote{One can also use the Wendel's Theorem to prove this proposition.}
    \[
    \mathbb{P}\big(\ker(A) \cap \Delta^m \neq \varnothing\big) = \mathbb{P}(\text{all signs equal}) = \frac{2}{2^m} = 2^{1-m}. \qedhere
    \]
\end{proof}

\subsection{Proof of Proposition \ref{prop:limiting_action_distributions}}
\begin{proof}
    The result follows directly from the convergence theorems established earlier. In cases (i) and (iii), the posterior converges to a degenerate distribution, leading to deterministic action selection. In case (ii), the posterior converges to a Bernoulli distribution, giving rise to randomized action selection with probability $p^*$. In case (iv), the ergodic behavior of the posterior leads to time-averaged action probabilities determined by the stationary distribution.
    \end{proof}

\subsection{Proof of Proposition \ref{prop:limiting_average_regret}}
\begin{proof}
The oracle policy achieves expected reward $g(1)$ per period. Thompson Sampling's expected reward per period in the limit is determined by its limiting action distribution:
\begin{itemize}
    \item Cases (i) and (iv): Action 1 selected with probability 1 $\Rightarrow$ expected reward $g(1)$ $\Rightarrow$ regret 0.
    \item Case (ii): Action 2 selected with probability 1 $\Rightarrow$ expected reward $g(2)$ $\Rightarrow$ regret $g(1) - g(2)$.
    \item Case (iii): Action selection converges to a Bernoulli distribution. With probability $p^*$, the algorithm converges to always selecting action 1 (regret 0), and with probability $1-p^*$, it converges to always selecting action 2 (regret $g(1) - g(2)$).
    \item Case (v): Expected reward is $\alpha^* g(1) + (1-\alpha^*) g(2)$ $\Rightarrow$ regret $(1-\alpha^*)[g(1) - g(2)]$.
\end{itemize}
\end{proof}

\subsection{Proof of Proposition \ref{prop:MM_limiting_action_dist}}
\begin{proof}
    Let $\pi_t \in \Delta_M$ denote the posterior belief at time $t$. Under Thompson Sampling, the decision maker samples an index $J_t \sim \pi_t$ and chooses action $A_t = \phi(\theta^{(J_t)})$. The conditional probability of choosing action $a$ given the belief $\pi_t$ is defined as:
    \begin{equation}
        p_a(\pi_t) \;:=\; \mathbb{P}(A_t = a \mid \pi_t) \;=\; \sum_{j=1}^M \pi_t(\theta^{(j)}) \, \mathbf{1}\{\phi(\theta^{(j)}) = a\}.
    \end{equation}
    The unconditional probability of selecting action $a$ is $\mathbb{P}(A_t = a) = \mathbb{E}[p_a(\pi_t)]$, where the expectation is taken over the stochastic evolution of the belief $\pi_t$.
    
    The classification results in Section \ref{sec:multi-model-setting} establish that the sequence of posterior beliefs $\{\pi_t\}_{t \ge 0}$ converges in distribution (weakly) to a limiting random variable $\pi_\infty$ supported on $\Delta_M$. The law of $\pi_\infty$ admits the decomposition:
    \begin{equation}
        \mathcal{L}(\pi_\infty) = \sum_{v \in V_\star} W_v \delta_{e_v} + \sum_{I \in \mathcal{F}_\star} W_I \mu_I,
    \end{equation}
    where $V_\star$ is the set of terminal vertices, $\mathcal{F}_\star$ is the set of ergodic faces, $\mu_I$ is the unique invariant distribution on face $\Delta_I$, and $W$ represent the basin of attraction probabilities.
    
    The function $\pi \mapsto p_a(\pi)$ is a linear functional of the belief vector; therefore, it is continuous and bounded on the compact simplex $\Delta_M$. By the definition of weak convergence (the Portmanteau theorem), the convergence of measures implies the convergence of expectations for bounded continuous functions:
    \begin{align}
        \alpha(a) &:= \lim_{t \to \infty} \mathbb{P}(A_t = a) = \lim_{t \to \infty} \mathbb{E}[p_a(\pi_t)] \nonumber \\
        &= \mathbb{E}[p_a(\pi_\infty)] = \int_{\Delta_M} p_a(\pi) \, \mathcal{L}(\pi_\infty)(d\pi).
    \end{align}
    Substituting the decomposition of $\mathcal{L}(\pi_\infty)$:
    \begin{equation}
        \alpha(a) = \sum_{v \in V_\star} W_v \int_{\Delta_M} p_a(\pi) \delta_{e_v}(d\pi) + \sum_{I \in \mathcal{F}_\star} W_I \int_{\Delta_M} p_a(\pi) \mu_I(d\pi).
    \end{equation}
    For the vertex terms, the measure is a point mass at the unit vector $e_v$. The action probability simplifies to:
    \begin{equation}
        p_a(e_v) = \sum_{j=1}^M \delta_{jv} \mathbf{1}\{\phi(\theta^{(j)}) = a\} = \mathbf{1}\{\phi(\theta^{(v)}) = a\}.
    \end{equation}
    For the face terms, the integral is taken with respect to the invariant measure $\mu_I$ supported on $\Delta_I$. Combining these yields the stated result:
    \begin{equation}
        \alpha(a) = \sum_{v \in V_\star} W_v \, \mathbf{1}\{\phi(\theta^{(v)})=a\} + \sum_{I \in \mathcal{F}_\star} W_I \int_{\Delta_I} p_a(\pi) \, \mu_I(d\pi).
    \end{equation}
\end{proof}

\subsection{Proof of Proposition \ref{prop:MM_limiting_regret}}
\begin{proof}
    Let $r(a) := g(a^\star) - g(a)$ be the expected per-period regret of action $a$. The average expected regret over horizon $T$ is:
    \begin{equation}
        \overline{R}_T := \frac{1}{T} \sum_{t=0}^{T-1} \mathbb{E}[g(a^\star) - g(A_t)] = \frac{1}{T} \sum_{t=0}^{T-1} \sum_{a \in \mathcal{A}} \mathbb{P}(A_t = a) \, r(a).
    \end{equation}
    By the linearity of expectation, we factor out the summation over actions. From Proposition \ref{prop:MM_limiting_action_dist}, we established that the sequence of action probabilities converges: $\lim_{t \to \infty} \mathbb{P}(A_t = a) = \alpha(a)$.
    
    Since the sequence $\{\mathbb{P}(A_t = a)\}_{t \ge 0}$ converges to $\alpha(a)$, its Cesàro mean (the time average) also converges to $\alpha(a)$. Therefore:
    \begin{align}
        \overline{R} &:= \lim_{T \to \infty} \overline{R}_T = \sum_{a \in \mathcal{A}} r(a) \left( \lim_{T \to \infty} \frac{1}{T} \sum_{t=0}^{T-1} \mathbb{P}(A_t = a) \right) \nonumber \\
        &= \sum_{a \in \mathcal{A}} \alpha(a) \, r(a).
    \end{align}
    Substituting the explicit form of $\alpha(a)$ derived in Proposition \ref{prop:MM_limiting_action_dist} into this equation yields the decomposed regret formula:
    \begin{equation}
         \overline{R} = \sum_{v \in V_\star} W_v \, r(\phi(\theta^{(v)})) + \sum_{I \in \mathcal{F}_\star} W_I \sum_{a \in \mathcal{A}} \left( \int_{\Delta_I} p_a(\pi) \mu_I(d\pi) \right) r(a).
    \end{equation}
    This general form encompasses all specific regimes (Interior Ergodic, Uniform Dominance, Vertex Selection, and Face-Ergodic) by restricting the sets $V_\star$ and $\mathcal{F}_\star$ accordingly.
\end{proof}

\subsection{Proof of Proposition \ref{prop:uniqueness_simple}}
    
    \begin{proof}
    The fixed point condition can be rewritten as the linear system $G \pi_{-M}^* = -d_M$ where $\pi_{-M}^* = (\pi_1^*, \ldots, \pi_{M-1}^*)^T$.
    
    ($\Rightarrow$) Suppose the interior fixed point is unique. If $G$ does not have full rank, then $\ker(G) \neq \{0\}$. For any non-zero $v \in \ker(G)$, we have $G v = 0$, so $G(\pi_{-M}^* + \epsilon v) = G\pi_{-M}^* = -d_M$ for any $\epsilon$. For sufficiently small $|\epsilon|$, both $\pi_{-M}^* + \epsilon v$ and $\pi_{-M}^* - \epsilon v$ yield valid interior solutions, contradicting uniqueness.
    
    ($\Leftarrow$) Suppose $G$ has full rank. Then $G$ is invertible, so the system $G \pi_{-M}^* = -d_M$ has a unique solution $\pi_{-M}^* = -G^{-1} d_M$. This uniquely determines $\pi_M^* = 1 - \mathbf{1}^T \pi_{-M}^*$, yielding a unique interior fixed point.
\end{proof}

\section{Numerical Simulations for Two-Arm Cases}
\label{app:simulations}

We present numerical simulations for the two-arm bandit problem discussed in Section \ref{sec:two_arm_dynamics}. We consider a Gaussian bandit setting with two models $\nu$ and $\gamma$. The true reward distributions are $R_t \mid A_t = i \sim \mathcal{N}(g(i), 1)$. The models believe $R_t \mid A_t = i \sim \mathcal{N}(\theta_i, 1)$ for $\theta \in \{\nu, \gamma\}$. In all cases, we set the time horizon $T=500$ and run $N=300$ Monte Carlo replications.

\subsection{Agreement Case}
In this case, both models agree on the optimal action ($\phi(\nu) = \phi(\gamma) = 1$). We set:
\begin{itemize}
    \item Model parameters: $\nu = (1.0, -1.0)$, $\gamma = (0.6, -0.6)$
    \item True rewards: $g = (0.95, -0.5)$
    \item Expected evidence: $\Delta_1 = 0.06, \Delta_2 = -0.06$
\end{itemize}
As predicted by Theorem \ref{thm:agreement_case}, the posterior converges to the model with the better fit (model $\nu$ as $\Delta_1 > 0$). Figure \ref{fig:case1_simulation} shows the posterior dynamics and average regret.

\begin{figure}[H]
    \centering
    \begin{subfigure}[b]{0.48\textwidth}
        \centering
        \includegraphics[width=\textwidth]{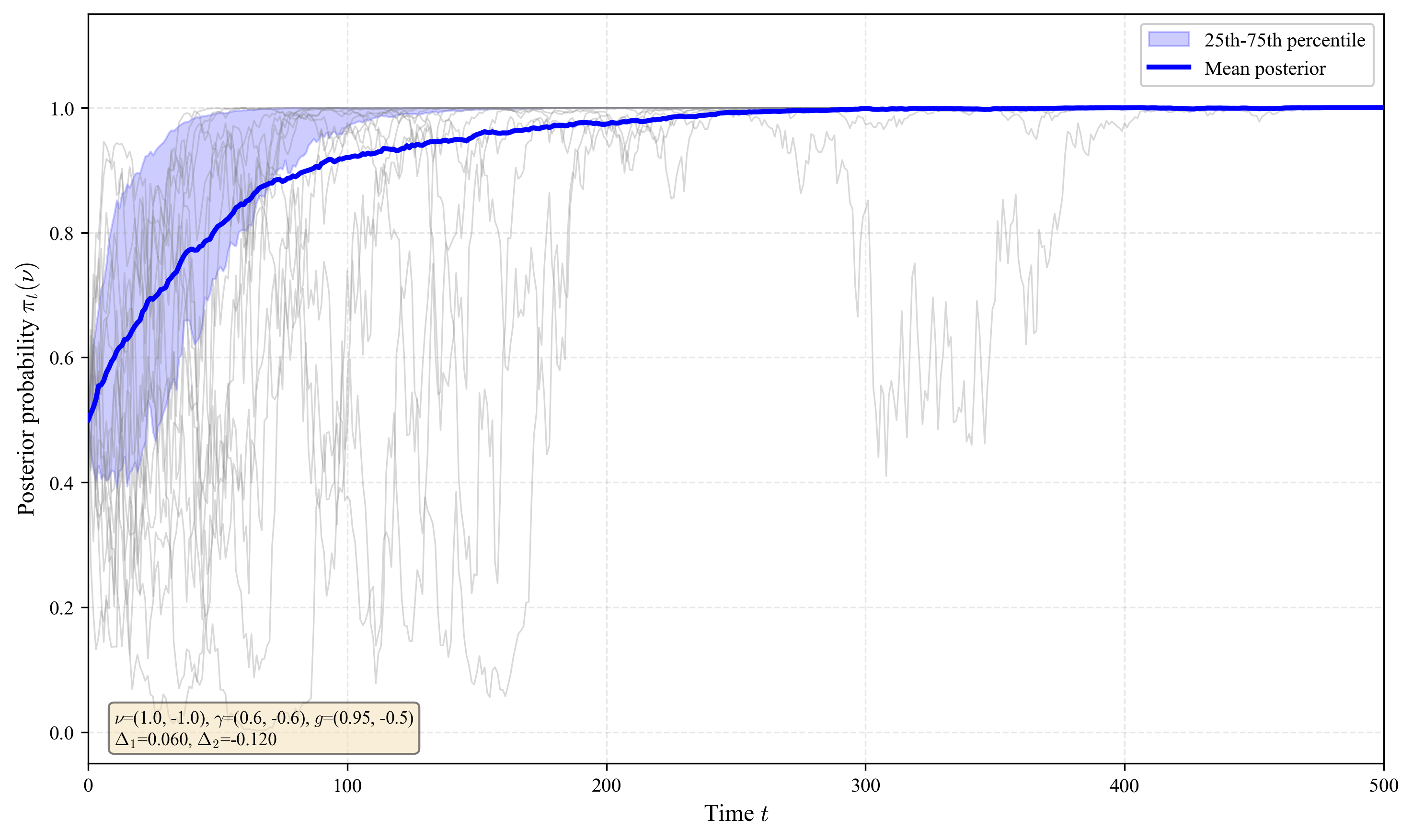}
        \caption{Posterior dynamics}
        \label{fig:case1_posterior}
    \end{subfigure}
    \hfill
    \begin{subfigure}[b]{0.48\textwidth}
        \centering
        \includegraphics[width=\textwidth]{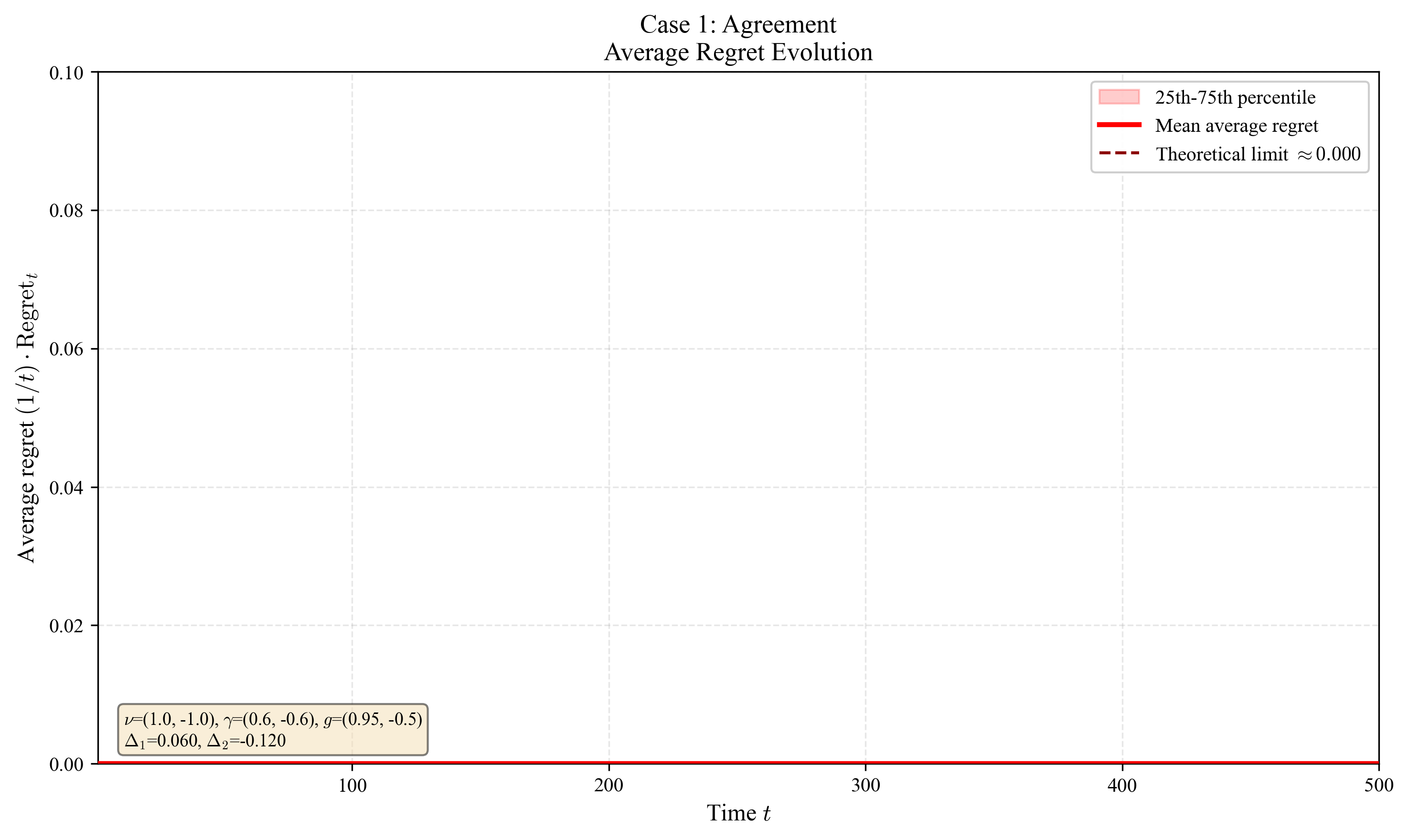}
        \caption{Average regret dynamics}
        \label{fig:case1_regret}
    \end{subfigure}
    \caption{Agreement case simulation results.}
    \label{fig:case1_simulation}
\end{figure}

\subsection{Self-Confirming Case}
In the self-confirming case, models disagree on the optimal action ($\phi(\nu) = 1, \phi(\gamma) = 2$). We set:
\begin{itemize}
    \item Model parameters: $\nu = (1.0, -1.0)$, $\gamma = (-1.0, 1.0)$
    \item True rewards: $g = (0.6, 0.2)$
    \item Expected evidence: $\Delta_1 = 1.2, \Delta_2 = -0.4$
\end{itemize}
Theorem \ref{thm:self_confirming} predicts convergence to a Bernoulli distribution. Since $\Delta_1 > 0 > \Delta_2$, each model receives positive feedback when its preferred action is chosen. Figure \ref{fig:case2_simulation} illustrates the posterior converging to either 0 or 1.

\begin{figure}[H]
    \centering
    \begin{subfigure}[b]{0.48\textwidth}
        \centering
        \includegraphics[width=\textwidth]{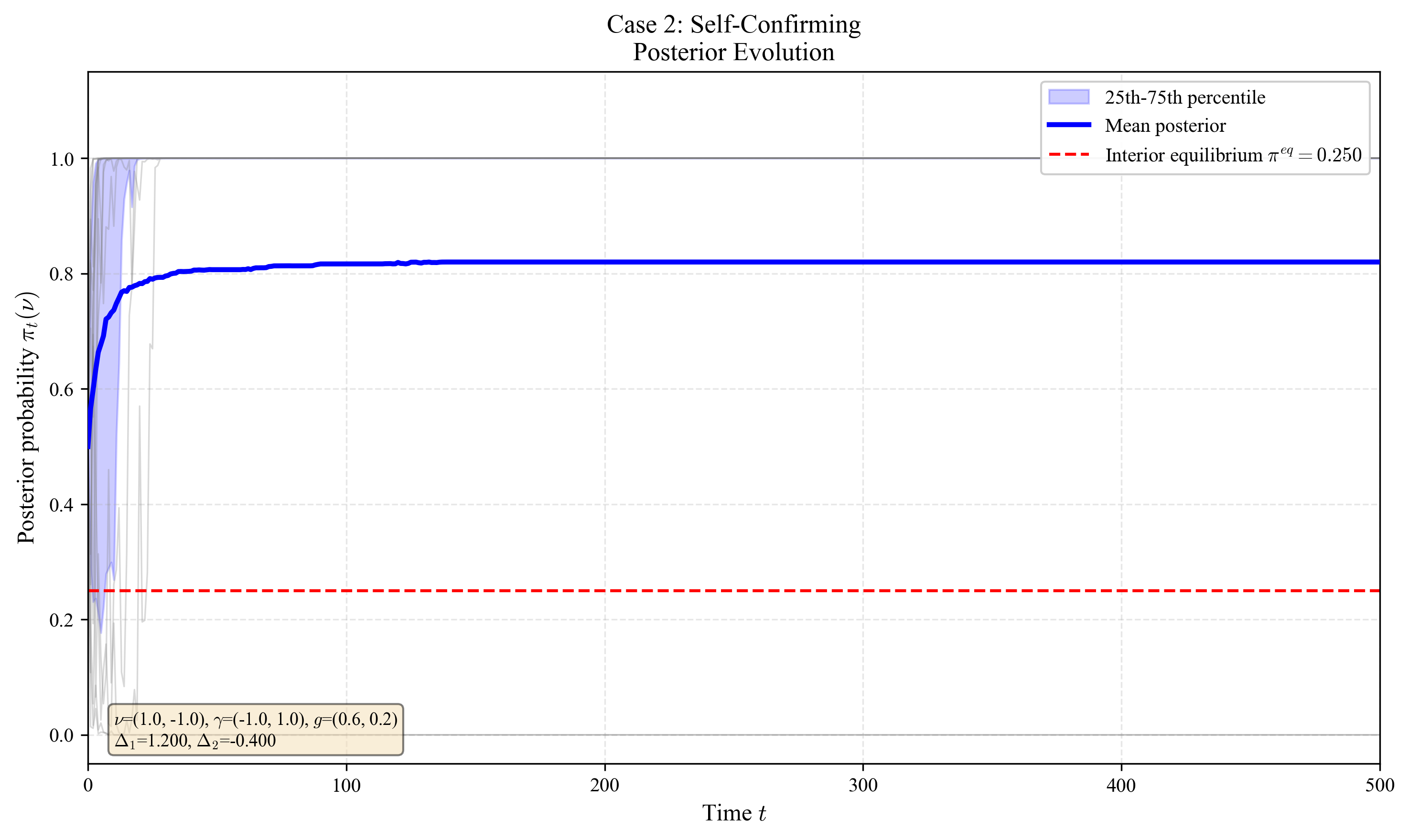}
        \caption{Posterior dynamics}
        \label{fig:case2_posterior}
    \end{subfigure}
    \hfill
    \begin{subfigure}[b]{0.48\textwidth}
        \centering
        \includegraphics[width=\textwidth]{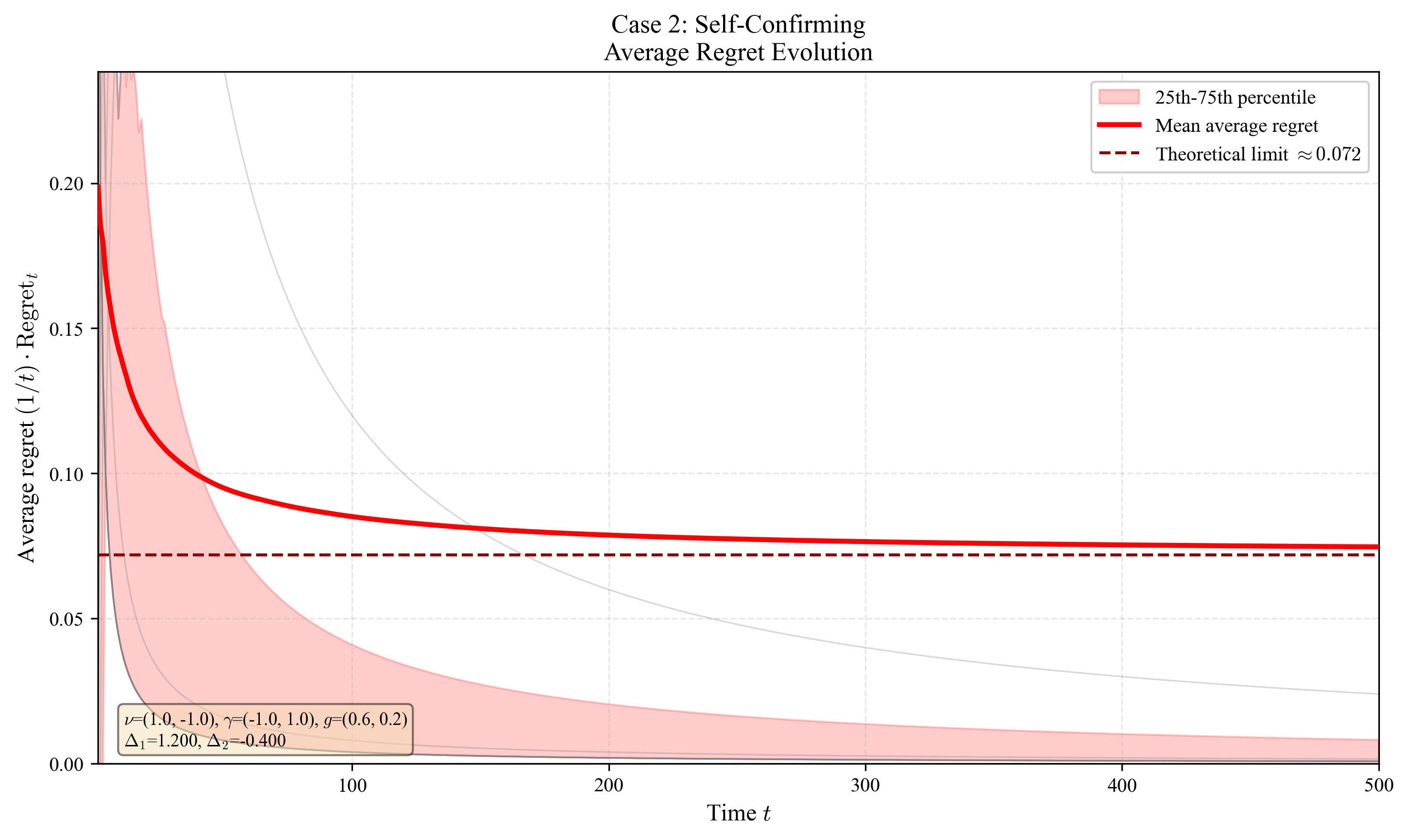}
        \caption{Average regret dynamics}
        \label{fig:case2_regret}
    \end{subfigure}
    \caption{Self-confirming case simulation results.}
    \label{fig:case2_simulation}
\end{figure}

\subsection{Uniform Dominance Case}
In the uniform dominance case, models disagree ($\phi(\nu) = 1, \phi(\gamma) = 2$), but model $\nu$ is favored under both actions. We set:
\begin{itemize}
    \item Model parameters: $\nu = (1.0, -1.0)$, $\gamma = (-1.0, 1.0)$
    \item True rewards: $g = (0.7, -0.7)$
    \item Expected evidence: $\Delta_1 = 1.4, \Delta_2 = 1.4$
\end{itemize}
Theorem \ref{thm:uniform_dominance} predicts almost sure convergence to $\pi = 1$. Figure \ref{fig:case3_simulation} confirms that model $\nu$ uniformly dominates.

\begin{figure}[H]
    \centering
    \begin{subfigure}[b]{0.48\textwidth}
        \centering
        \includegraphics[width=\textwidth]{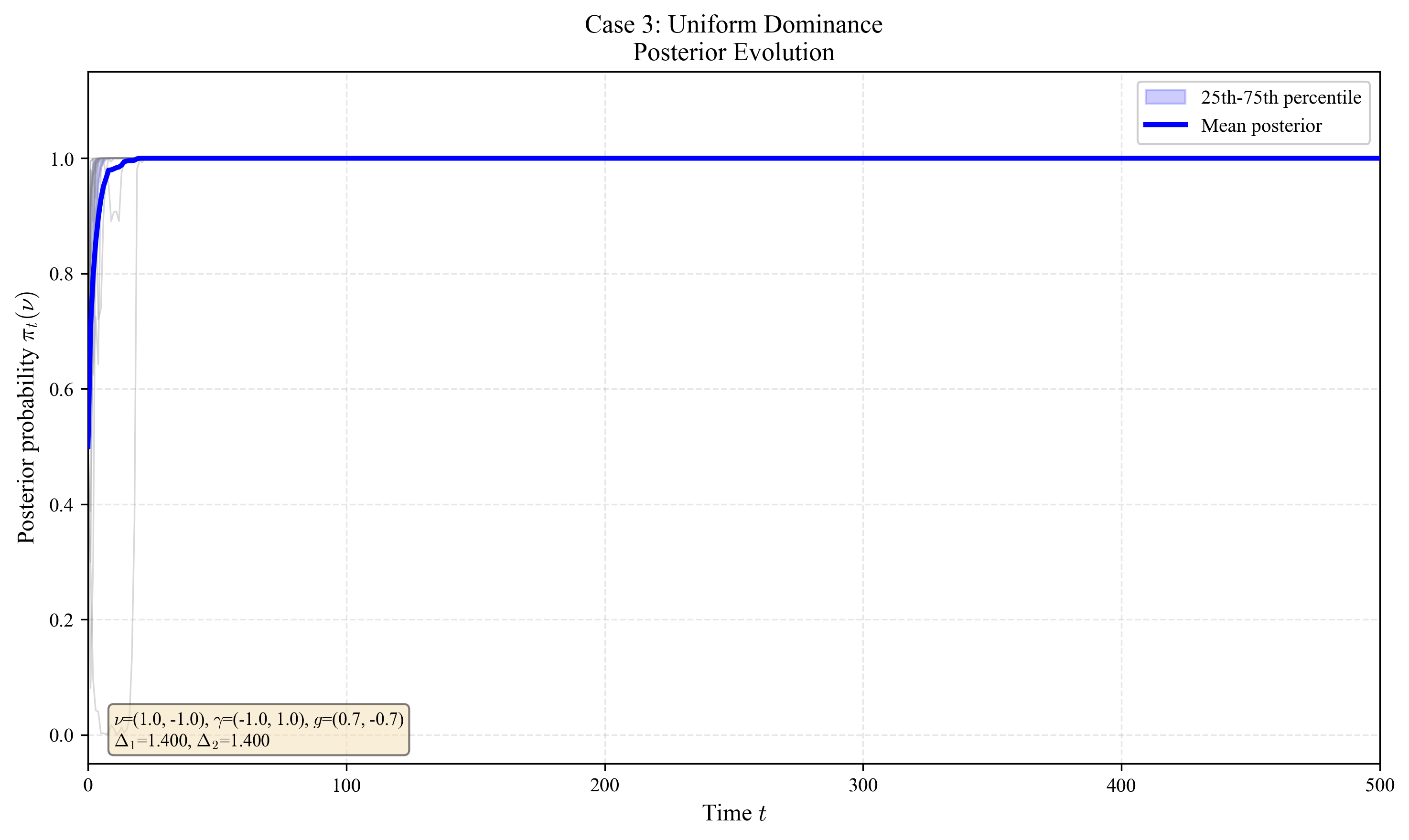}
        \caption{Posterior dynamics}
        \label{fig:case3_posterior}
    \end{subfigure}
    \hfill
    \begin{subfigure}[b]{0.48\textwidth}
        \centering
        \includegraphics[width=\textwidth]{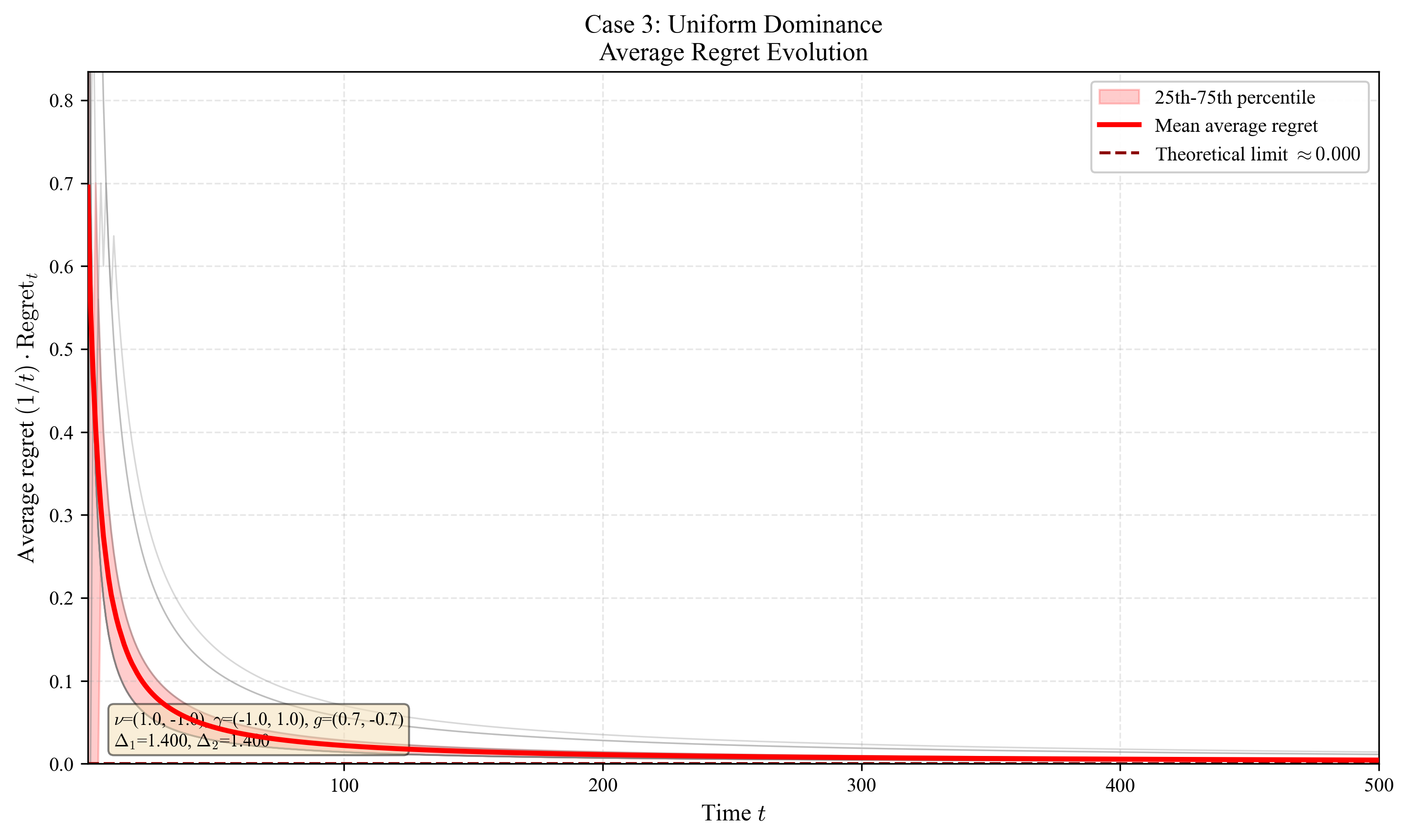}
        \caption{Average regret dynamics}
        \label{fig:case3_regret}
    \end{subfigure}
    \caption{Uniform dominance case simulation results.}
    \label{fig:case3_simulation}
\end{figure}

\subsection{Self-Defeating Case}
In the self-defeating case, models disagree ($\phi(\nu) = 1, \phi(\gamma) = 2$), and each model receives negative feedback when selected. We set:
\begin{itemize}
    \item Model parameters: $\nu = (1.0, -1.0)$, $\gamma = (-1.0, 1.0)$
    \item True rewards: $g = (-0.2, -0.6)$
    \item Expected evidence: $\Delta_1 = -0.4, \Delta_2 = 1.2$
\end{itemize}
Theorem \ref{thm:self_defeating} predicts convergence to a non-degenerate stationary distribution. Figure \ref{fig:case4_simulation} shows the persistent posterior oscillations and positive average regret.

\begin{figure}[H]
    \centering
    \begin{subfigure}[b]{0.48\textwidth}
        \centering
        \includegraphics[width=\textwidth]{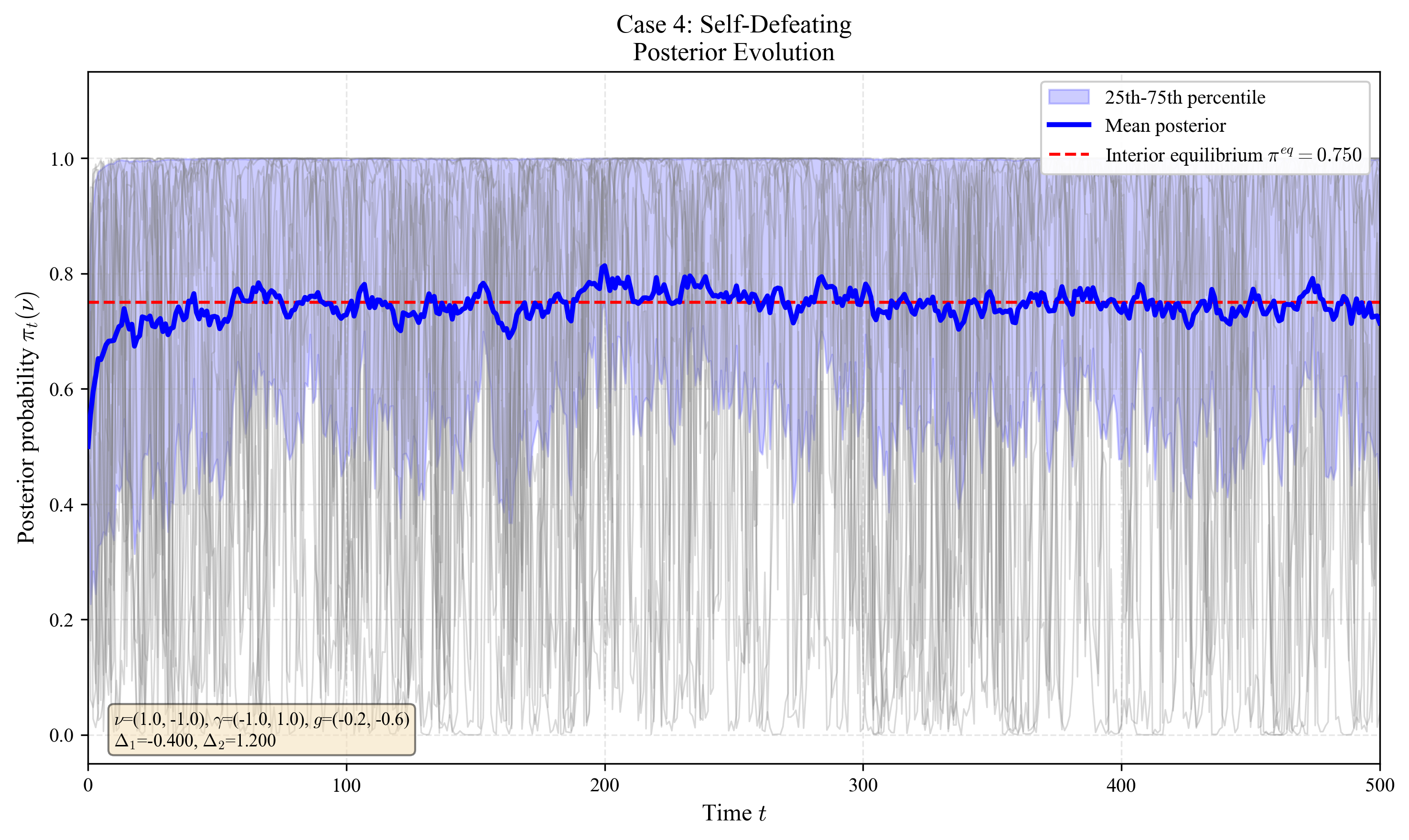}
        \caption{Posterior dynamics}
        \label{fig:case4_posterior}
    \end{subfigure}
    \hfill
    \begin{subfigure}[b]{0.48\textwidth}
        \centering
        \includegraphics[width=\textwidth]{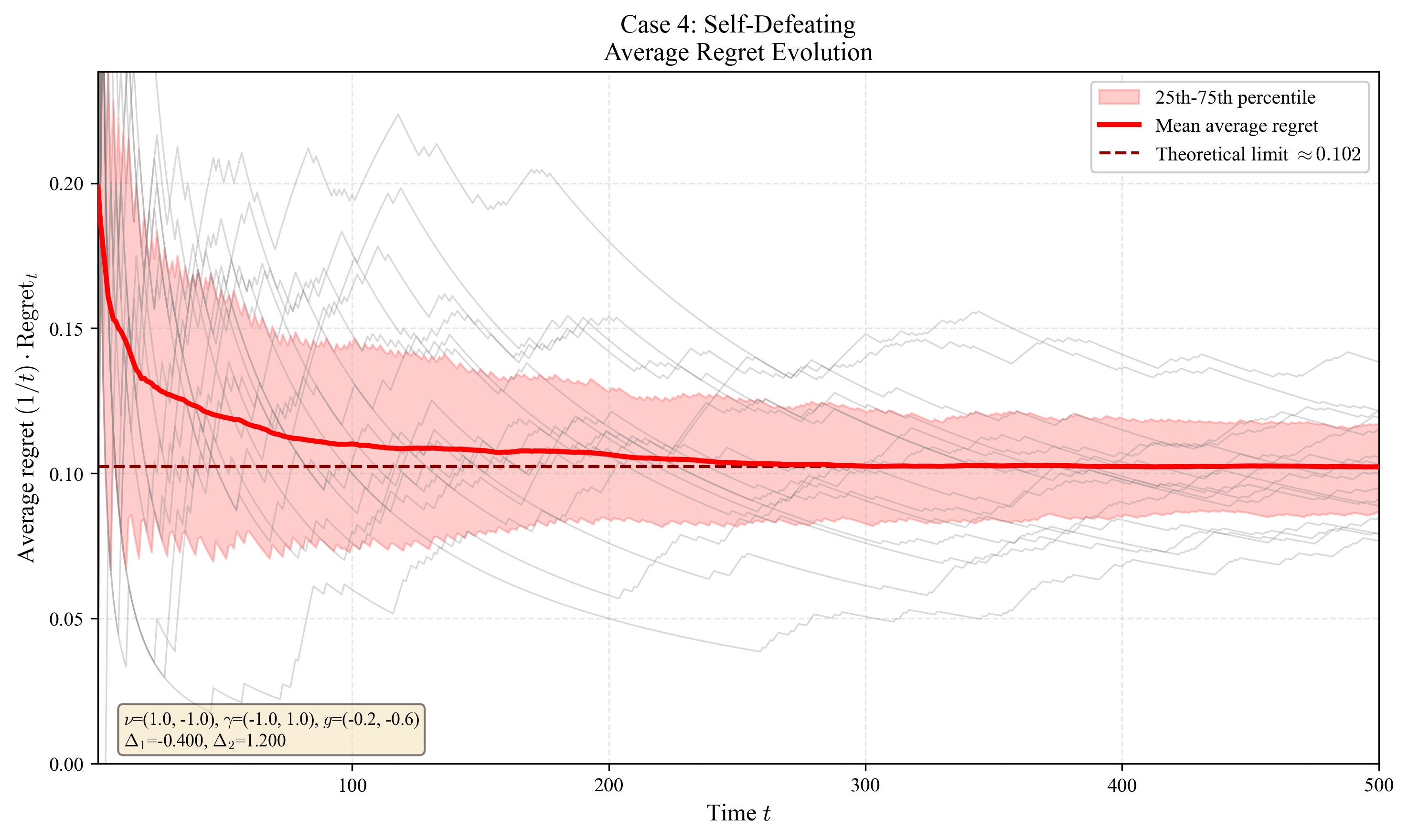}
        \caption{Average regret dynamics}
        \label{fig:case4_regret}
    \end{subfigure}
    \caption{Self-defeating case simulation results.}
    \label{fig:case4_simulation}
\end{figure}

\end{document}